\documentclass[11pt,a4paper]{article}

\usepackage[english]{babel}
\usepackage{german}
\usepackage{amssymb}
\usepackage{amsfonts}
\usepackage{amsmath}
\usepackage[latin1]{inputenc}
\usepackage{fullpage}
\usepackage{graphicx}
\usepackage{subfigure}
\usepackage{algorithm}
\usepackage[noend]{algorithmic}
\usepackage{amsmath,amssymb,cite}
\usepackage{lineno,footmisc,marvosym}
\usepackage{wasysym,vmargin,stackrel}
\usepackage{color,hyperref}
\usepackage{boxedminipage}
\usepackage{xspace}
\usepackage{todonotes}

\newcommand{\LL}{{\cal L}}
\newcommand{\PP}{{\cal P}}

\newcommand{\ProblemFormat}[1]{{\sc #1}}
\newcommand{\ProblemName}[1]{\ProblemFormat{#1}}

\newcommand{\bdsprob}{\ProblemName{Bounded Dominating Set} on Tolerance Graphs}
\newcommand{\rbdsprob}{\ProblemName{Restricted Bounded Dominating Set} on Tolerance Graphs}

\newcommand{\sscprob}{\ProblemName{Special 3-Set Cover}}

\newtheorem{theorem}{Theorem}

\newtheorem{definition}{Definition}
\newtheorem{observation}{Observation}

\newtheorem{lemma}{Lemma}

\newenvironment{proof}[1][Proof]{\noindent\textbf{#1.} }{\ \rule{0.5em}{0.5em}}
\newenvironment{proofofclaim}[1][Proof]{\noindent\textbf{#1.} }{\ensuremath{\square}}
\newtheorem{myclaim}{Claim}

\setmarginsrb{3.5cm}{2.25cm}{3.5cm}{3.05cm}{0.3cm}{0.3cm}{-0.3cm}{1.0cm}
   \renewenvironment{thebibliography}[1]{
 \begin{oldthebibliography}{#1}
 \setlength{\parskip}{0.0ex} \setlength{\itemsep}{1ex}}  {\end{oldthebibliography}}

\begin{document}

\title{New Geometric Representations and Domination Problems \\
on Tolerance and Multitolerance Graphs\thanks{%
A preliminary conference version of this work appeared in the \emph{%
Proceedings of the 32nd Symposium on Theoretical Aspects of Computer Science
(STACS)}, Munich, Germany, March 2015, pages 354-366.} \ \thanks{%
Partially supported by the EPSRC Grant EP/K022660/1.}}
\author{Archontia C.~Giannopoulou\thanks{%
Institute of Software Technology and Theoretical Computer Science, Technische Universit\"at Berlin, Germany. Email: \texttt{%
archontia.giannopoulou@gmail.com}} \ \thanks{%
The main part of this paper was prepared while the first author was
affiliated at the School of Engineering and Computing Sciences, Durham
University, UK.} \and George B.~Mertzios\thanks{%
School of Engineering and Computing Sciences, Durham University, UK. Email: 
\texttt{george.mertzios@durham.ac.uk}}}
\date{\vspace{-1.0cm}}
\maketitle

\begin{abstract}
Tolerance graphs model interval relations in such a way that intervals can
tolerate a certain amount of overlap without being in conflict. In one of
the most natural generalizations of tolerance graphs with direct
applications in the comparison of DNA sequences from different organisms,
namely \emph{multitolerance} graphs, two tolerances are allowed for each
interval --~one from the left and one from the right side. Several efficient
algorithms for optimization problems that are \textsf{NP}-hard in general
graphs have been designed for tolerance and multitolerance graphs. In spite
of this progress, the complexity status of some fundamental algorithmic
problems on tolerance and multitolerance graphs, such as the \emph{%
dominating set} problem, remained unresolved until now, three decades after
the introduction of tolerance graphs. In this article we introduce two new
geometric representations for tolerance and multitolerance graphs, given by
points and line segments in the plane. Apart from being important on their
own, these new representations prove to be a powerful tool for deriving both
hardness results and polynomial time algorithms. Using them, we surprisingly
prove that the dominating set problem can be solved in polynomial time on
tolerance graphs and that it is \textsf{APX}-hard on multitolerance graphs,
solving thus a longstanding open problem. This problem is the first one that
has been discovered with a different complexity status in these two graph
classes.\newline

\noindent \textbf{Keywords:} Tolerance graph, multitolerance graph,
geometric representation, dominating set problem, polynomial time algorithm, 
\textsf{APX}-hard.\newline

\noindent \textbf{AMS subject classification (2010).} Primary 05C62, 05C85;
Secondary 68R10.
\end{abstract}

\section{Introduction\label{introduction-sec}}

A graph $G=(V,E)$ on $n$ vertices is a \emph{tolerance} graph if there
exists a collection ${I=\{I_{v}\ |\ v\in V\}}$ of intervals on the real line
and a set ${t=\{t_{v}\ |\ v\in V\}}$ of positive numbers (the tolerances),
such that for any two vertices ${u,v\in V}$, ${uv\in E}$ if and only if ${%
|I_{u}\cap I_{v}|\geq \min \{t_{u},t_{v}\}}$, where $|I|$ denotes the length
of the interval $I$. The pair ${\langle I,t\rangle }$ is called a \emph{%
tolerance representation} of $G$. If $G$ has a tolerance representation~${%
\langle I,t\rangle }$, such that ${t_{v}\leq |I_{v}|}$ for every ${v\in V}$,
then $G$ is called a \emph{bounded tolerance} graph.

If we replace in the above definition ``min'' by~``max'', we obtain the
class of \emph{max-tolerance} graphs. Both tolerance and max-tolerance
graphs have attracted many research efforts~\cite%
{Bogart95,Bus06,Fel98,GolumbicMonma84,GolSi02,MSZ-Model-SIDMA-09,MSZ-SICOMP-11,Kaufmann06,Lehmann06,Multitol-Mertzios14}
as they find numerous applications, especially in bioinformatics, among
others~\cite{GolSi02,Kaufmann06,Lehmann06}; for a more detailed account see
the book on tolerance graphs~\cite{GolTol04}. One of their major
applications is in the comparison of DNA sequences from different organisms
or individuals by making use of a software tool like BLAST~\cite{Altschul90}%
. However, at some parts of the above genomic sequences in BLAST, we may
want to be more tolerant than at other parts, since for example some of them
may be biologically less significant or we have less confidence in the exact
sequence due to sequencing errors in more error prone genomic regions. This
concept leads naturally to the notion of \emph{multitolerance} graphs which
generalize tolerance graphs~\cite{Parra98,GolTol04,Multitol-Mertzios14}. The
main idea is to allow two different tolerances for each interval, one to
each of its sides. Then, every interval tolerates in its interior part the
intersection with other intervals by an amount that is a convex combination
of these two border-tolerances.

Formally, let ${I=[l,r]}$ be an interval on the real line and ${%
l_{t},r_{t}\in I}$ be two numbers between~$l$ and~$r$, called \emph{tolerant
points}. For every ${\lambda \in \lbrack 0,1]}$, we define the interval $%
I_{l_{t},r_{t}}(\lambda )={[l+(r_{t}-l)\lambda },$ $l_{t}+(r-l_{t})\lambda ]$%
, which is the convex combination of ${[l,l_{t}]}$ and ${[r_{t},r]}$.
Furthermore, we define the set $\mathcal{I}(I,l_{t},r_{t})=%
\{I_{l_{t},r_{t}}(\lambda )\ |\ \lambda \in \lbrack 0,1]\}$ of intervals.
That is, $\mathcal{I}(I,l_{t},r_{t})$ is the set of all intervals that we
obtain when we linearly transform~$[l,l_{t}]$ into~$[r_{t},r]$. For an
interval~$I$, the \emph{set of tolerance-intervals}~$\tau $ of $I$ is
defined either as $\tau =\mathcal{I}(I,l_{t},r_{t})$ for some values $%
l_{t},r_{t}\in I$ (the case of a \emph{bounded} vertex), or as $\tau =\{%
\mathbb{R}\}$ (the case of an \emph{unbounded} vertex). A graph $G=(V,E)$ is
a \emph{multitolerance} graph if there exists a collection $I=\{I_{v}\ |\
v\in V\}$ of intervals and a family $t=\{\tau _{v}\ |\ v\in V\}$ of sets of
tolerance-intervals, such that: for any two vertices~${u,v\in V}$, $uv\in E$
if and only if $Q_{u}\subseteq I_{v}$ for some $Q_{u}\in \tau _{u}$, or $%
Q_{v}\subseteq I_{u}$ for some $Q_{v}\in \tau _{v}$. Then, the pair~$\langle
I,t\rangle $ is called a \emph{multitolerance representation} of~$G$. If $G$
has a multitolerance representation with only bounded vertices, i.e., with $%
\tau_{v}\neq \{\mathbb{R}\}$ for every vertex $v$, then $G $ is called a 
\emph{bounded multitolerance} graph.

For several optimization problems that are \textsf{NP}-hard in general
graphs, such as the coloring, clique, and independent set problems,
efficient algorithms are known for tolerance and multitolerance graphs.
However, only few of them have been derived using the (multi)tolerance
representation (e.g.~\cite{Parra98,GolSi02}), while most of these algorithms
appeared as a consequence of the containment of tolerance and multitolerance
graphs to weakly chordal (and thus also to perfect) graphs~\cite{Spinrad03}.
To design efficient algorithms for (multi)tolerance graphs, it seems to be
essential to assume that a suitable representation of the graph is given
along with the input, as it has been recently proved that the recognition of
tolerance graphs is \textsf{NP}-complete~\cite{MSZ-SICOMP-11}. Recently two
new geometric intersection models in the 3-dimensional space have been
introduced for both tolerance graphs (the \emph{parallelepiped}
representation~\cite{MSZ-Model-SIDMA-09}) and multitolerance graphs (the 
\emph{trapezoepiped} representation~\cite{Multitol-Mertzios14}), which
enabled the design of very efficient algorithms for such problems, in most
cases with (optimal) $O(n\log n)$ running time~\cite%
{Multitol-Mertzios14,MSZ-Model-SIDMA-09}. In spite of this, the complexity
status of some algorithmic problems on tolerance and multitolerance graphs
still remains open, three decades after the introduction of tolerance graphs
in~\cite{GoMo82}. Arguably the two most famous and intriguing examples of
such problems are the \emph{minimum dominating set} problem and the \emph{%
Hamilton cycle} problem (see e.g.~\cite[page 314]{Spinrad03}). Both these
problems are known to be \textsf{NP}-complete on the greater class of weakly
chordal graphs~\cite{Booth82,Muller96} but solvable in polynomial time in
the smaller classes of bounded tolerance and bounded multitolerance (i.e.,
trapezoid) graphs~\cite{KratschStewart93,DeogunS94}. The reason that these
problems resisted solution attempts over the years seems to be that the
existing representations for (multi)tolerance graphs do not provide enough
insight to deal with these problems.

\paragraph{Our contribution.}

In this article we introduce a new geometric representation for
multitolerance graphs, which we call the \emph{shadow representation}, given
by a set of line segments and points in the plane. In the case of tolerance
graphs, this representation takes a very special form, in which all line
segments are horizontal, and therefore we call it the \emph{horizontal
shadow representation}. Note that both the shadow and the horizontal shadow
representations are \emph{not} intersection models for multitolerance graphs
and for tolerance graphs, respectively, in the sense that two line segments
may not intersect in the representation although the corresponding vertices
are adjacent. However, the main advantage of these two new representations
is that they provide substantially new insight for tolerance and
multitolerance graphs and they can be used to interpret optimization
problems (such as the dominating set problem and its variants) using
computational geometry terms.

Apart from being important on their own, these new representations enable us
to establish the complexity of the \emph{minimum dominating set} problem on
both tolerance and multitolerance graphs, thus solving a longstanding open
problem. Given a horizontal shadow representation of a tolerance graph $G$,
we present an algorithm that computes a minimum dominating set in polynomial
time. On the other hand, using the shadow representation, we prove that the
minimum dominating set problem is \textsf{APX}-hard on multitolerance graphs
by providing a reduction from a special case of the set cover problem. That
is, there exists no Polynomial Time Approximation Scheme (PTAS) for this
problem unless \textsf{P}=\textsf{NP}. This is the first problem that has
been discovered with a different complexity status in these two graph
classes. Therefore, given the (seemingly) small difference between the
definition of tolerance and multitolerance graphs, this dichotomy result
appears to be surprising.

\paragraph{Organization of the paper.}

In Section~\ref{preliminaries-sec} we briefly revise the 3-dimensional
intersection models for tolerance graphs~\cite{MSZ-Model-SIDMA-09} and
multitolerance graphs~\cite{Multitol-Mertzios14}, which are needed in order
to present our new geometric representations. In Section~\ref%
{representations-sec} we introduce our new geometric representation for
multitolerance graphs (the \emph{shadow representation}) and its special
case for tolerance graphs (the \emph{horizontal shadow representation}). In
Section~\ref{dominating-hard-multitolerance-sec} we prove that \textsc{%
Dominating Set} on multitolerance graphs is \textsf{APX}-hard. Then, in
Sections~\ref{Bounded-dominating-sec}-\ref{tolerance-domination-sec} we
present our polynomial algorithm for the dominating set problem on tolerance
graphs, using the horizontal shadow representation (cf.~Algorithms~\ref%
{bounded-dominating-tolerance-alg},~\ref{restricted-bounded-alg}, and~\ref%
{dominating-tol-alg}). In particular, we first present Algorithm~\ref%
{bounded-dominating-tolerance-alg} in Section~\ref{Bounded-dominating-sec},
which solves a variation of the dominating set problem on tolerance graphs,
called \textsc{Bounded Dominating Set}. Then we present Algorithm~\ref%
{restricted-bounded-alg} in Section~\ref{Restricted-domination-sec}, which
uses Algorithm~\ref{bounded-dominating-tolerance-alg} as a subroutine in
order to solve a slightly modified version of \textsc{Bounded Dominating Set}
on tolerance graphs, namely \textsc{Restricted Bounded Dominating Set}. In
Section~\ref{tolerance-domination-sec} we present our main algorithm
(Algorithm~\ref{dominating-tol-alg}) which solves \textsc{Dominating Set} on
tolerance graphs in polynomial time, using Algorithms~\ref%
{bounded-dominating-tolerance-alg} and~\ref{restricted-bounded-alg} as
subroutines. Finally, in Section \ref{conclusions-sec} we discuss the
presented results and some interesting further research questions.

\paragraph{Notation.}

In this article we consider simple undirected graphs with no loops or
multiple edges. In an undirected graph $G$ the edge between two vertices $u$
and $v$ is denoted by~$uv$, and in this case $u$ and $v$ are said to be 
\emph{adjacent} in $G$. We denote by $N(u)=\{v\in V:uv\in E\}$ the set of
neighbors of a vertex $u$ in $G$, and $N[u]=N(u)\cup \{u\}$. Given a graph ${%
G=(V,E)}$ and a subset ${S\subseteq V}$, $G[S]$ denotes the induced subgraph
of $G$ on the vertices in $S$. A subset $S\subseteq V$ is a \emph{dominating
set} of $G$ if every vertex $v\in V\setminus S$ has at least one neighbor in 
$S$. Finally, given a set $X\subseteq \mathbb{R}^{2}$ of points in the
plane, we denote by~$H_{\text{convex}}(X)$ the \emph{convex hull} defined by
the points of $X$, and by $\overline{X}=\mathbb{R}^{2}\setminus X$ the
complement of $X$ in $\mathbb{R}^{2}$. For simplicity of the presentation we
make the following notational convention throughout the paper: whenever we
need to compute a set $S$ with the smallest cardinality among a family $%
\mathcal{S}$ of sets, we write $S=\min \{\mathcal{S}\}$.

\section{Tolerance and multitolerance graphs\label{preliminaries-sec}}

In this section we briefly revise the 3-dimensional intersection model for
tolerance graphs~\cite{MSZ-Model-SIDMA-09} and its generalization to
multitolerance graphs~\cite{Multitol-Mertzios14}, together with some useful
properties of these models that are needed for the remainder of the paper.
Since the intersection model of~\cite{MSZ-Model-SIDMA-09} for tolerance
graphs is a special case of the intersection model of~\cite%
{Multitol-Mertzios14} for multitolerance graphs, we mainly focus below on
the more general model for multitolerance graphs.

Consider a multitolerance graph $G=(V,E)$ that is given along with a
multitolerance representation $R$. Let $V_{B}$ and $V_{U}$ denote the set of
bounded and unbounded vertices of $G$ in this representation, respectively.
Consider now two parallel lines~$L_{1}$ and $L_{2}$ in the plane. For every
vertex $v\in V=V_{B}\cup V_{U}$, we appropriately construct a trapezoid $%
\overline{T}_{v}$ with its parallel lines on~$L_{1}$ and~$L_{2}$,
respectively (for details of this construction of the trapezoids we refer to~%
\cite{Multitol-Mertzios14}). According to this construction, for every
unbounded vertex $v\in V_{U}$ the trapezoid $\overline{T}_{v}$ is trivial,
i.e.,~a line~\cite{Multitol-Mertzios14}. For every vertex $v\in V=V_{B}\cup
V_{U}$ we denote by $a_{v},b_{v},c_{v},d_{v}$ the lower left, upper right,
upper left, and lower right endpoints of the trapezoid $\overline{T}_{v}$,
respectively. Note that for every unbounded vertex $v\in V_{U}$ we have $%
a_{v}=d_{v}$ and $c_{v}=b_{v}$, since $\overline{T}_{v}$ is just a line
segment. An example is depicted in Figure~\ref{TuTv-modified-fig}, where $%
\overline{T}_{u}$ corresponds to a bounded vertex $u$ and~$\overline{T}_{v}$
corresponds to an unbounded vertex~$v$.

\begin{figure}[t]
\centering
\includegraphics[width=\textwidth]{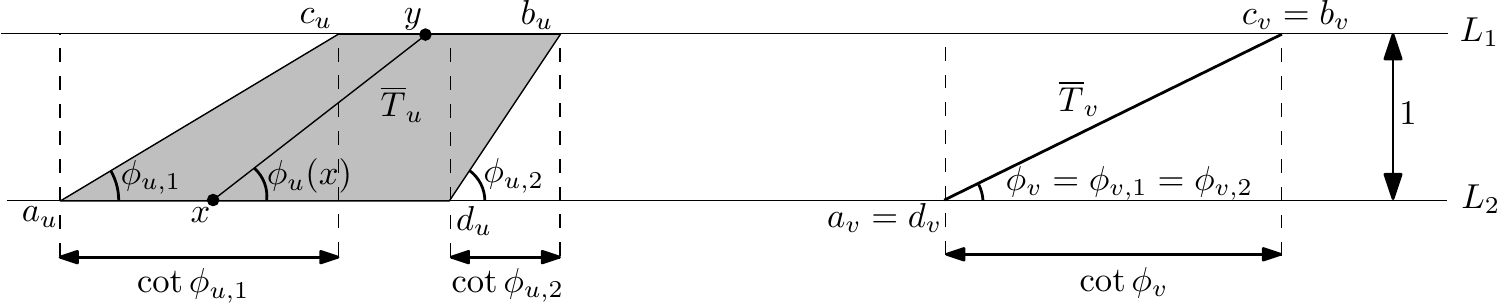}
\caption{The trapezoid $\overline{T}_{u}$ corresponds to the bounded vertex $%
u\in V_{B}$, while the line segment $\overline{T}_{v}$ corresponds to the
unbounded vertex $v\in V_{U}$.}
\label{TuTv-modified-fig}
\end{figure}

We now define the left and right angles of these trapezoids. For every angle 
$\phi $, the values $\tan \phi $ and $\cot \phi =\frac{1}{\tan \phi }$
denote the tangent and the cotangent of $\phi $, respectively. Furthermore, $%
\phi =arc\cot x$ is the angle $\phi $, for which $\cot \phi =x$.

\begin{definition}[\hspace{-0,01mm}\protect\cite{Multitol-Mertzios14}]
\label{Tv}For every vertex $v\in V=V_{B}\cup V_{U}$, the values $\phi
_{v,1}=arc\cot \left( c_{v}-a_{v}\right) $ and $\phi _{v,2}=arc\cot \left(
b_{v}-d_{v}\right) $ are the \emph{left angle} and the \emph{right angle} of~%
$\overline{T}_{v}$, respectively. Moreover, for every unbounded vertex $v\in 
{V_{U}}$, $\phi _{v}=\phi _{v,1}=\phi _{v,2}$ is the \emph{angle} of $%
\overline{T}_{v}$.
\end{definition}

Note here that, if $G$ is given along with a \emph{tolerance} representation 
$R$ (i.e., if $G$ is a tolerance graph), then for every bounded vertex $u$
we have that $\phi _{u,1}=\phi _{u,2}$, and thus the corresponding trapezoid 
$\overline{T}_{u}$ always becomes a \emph{parallelogram}~\cite%
{Multitol-Mertzios14} (see also~\cite{MSZ-Model-SIDMA-09}).

Without loss of generality we can assume that all endpoints and angles of
the trapezoids are distinct, i.e.,~$\{{a_{u},b_{u},c_{u},d_{u}\}\cap }\{{%
a_{v},b_{v},c_{v},d_{v}\}=\emptyset }$ and $\{{\phi _{u,1},\phi _{u,2}\}\cap 
}\{{\phi _{v,1},\phi _{v,2}\}=\emptyset }$ for every $u,v\in V$ with $u\neq
v $, as well as that ${0<\phi _{v,1},\phi _{v,2}<\frac{\pi }{2}}$ for all
angles $\phi _{v,1},\phi _{v,2}$~\cite{Multitol-Mertzios14}. It is important
to note here that this set of trapezoids $\{\overline{T}_{v}:v\in
V=V_{B}\cup V_{U}\}$ is \emph{not} an intersection model for the graph $G$,
as two trapezoids $\overline{T}_{v},\overline{T}_{w}$ may have a non-empty
intersection although $vw\notin E$. However the subset of trapezoids $\{%
\overline{T}_{v}:v\in V_{B}\}$ that corresponds to the \emph{bounded}
vertices (i.e.,~to the vertices of $V_{B}$) is an intersection model of the
induced subgraph $G[V_{B}]$.

In order to construct an intersection model for the whole graph $G$
(i.e.,~including also the set $V_{U}$ of the unbounded vertices), we exploit
the third dimension as follows. Let $\Delta =\max \{b_{v}:v\in V\}-\min
\{a_{u}:u\in V\}$ (where we consider the endpoints $b_{v}$ and $a_{u}$ as
real numbers on the lines $L_{1}$ and $L_{2}$, respectively). First, for
every unbounded vertex~$v\in V_{U}$ we construct the line~segment ${%
T_{v}=\{(x,y,z)}:{(x,y)\in \overline{T}_{v},z=\Delta -\cot \phi _{v}\}}$.
For every bounded vertex $v\in V_{B}$, denote by $\overline{T}_{v,1}$ and $%
\overline{T}_{v,2}$ the left and the right line segment of the trapezoid $%
\overline{T}_{v}$, respectively. We construct two line~segments ${%
T_{v,1}=\{(x,y,z):(x,y)\in \overline{T}_{v,1},z=\Delta -\cot \phi _{v,1}\}}$
and ${T_{v,2}=\{(x,y,z):(x,y)\in \overline{T}_{v,2},z=\Delta -\cot \phi
_{v,2}\}}$. Then, for every $v\in V_{B}$, we construct the 3-dimensional
object $T_{v}$ as the convex hull $H_{\text{convex}}({\overline{T}%
_{v},T_{v,1},T_{v,2})}$; this 3-dimensional object $T_{v}$ is called the 
\emph{trapezoepiped} of vertex $v\in V_{B}$. The resulting set $\{T_{v}:v\in
V=V_{B}\cup V_{U}\}$ of objects in the 3-dimensional space is called the 
\emph{trapezoepiped representation} of the multitolerance graph $G$~\cite%
{Multitol-Mertzios14}. This is an \emph{intersection model} of~$G$,
i.e.,~two vertices $v,w$ are adjacent if and only if $T_{v}\cap T_{w}\neq
\emptyset $. For a proof of this fact and for more details about the
trapezoepiped representation of multitolerance graphs we refer to~\cite%
{Multitol-Mertzios14}.

Recall that, if $G$ is a tolerance graph, given along with a tolerance
representation $R$, then $\phi _{u,1}=\phi _{u,2}$ for every bounded vertex $%
u$. Therefore, in the above construction, for every bounded vertex $u$ the
trapezoepiped $T_{u}$ becomes a \emph{parallelepiped}, and in this case the
resulting trapezoepiped representation is called a \emph{parallelepiped
representation}~\cite{Multitol-Mertzios14,MSZ-Model-SIDMA-09}.

An example of the construction of a trapezoepiped representation is given in
Figure~\ref{fig:3D}. A multitolerance graph $G$ with six vertices $%
\{v_{1},v_{2},\ldots ,v_{6}\}$ is depicted in Figure~\ref{fig:3Da}, while
the trapezoepiped representation of $G$ is illustrated in Figure~\ref%
{fig:3Db}. The set of bounded and unbounded vertices in this representation
are~${V_{B}=\{v_{3},v_{4},v_{6}\}}$ and~${V_{U}=\{v_{1},v_{2},v_{5}\}}$,
respectively. 
\begin{figure}[t]
\centering%
\subfigure[]{ \label{fig:3Da}
\includegraphics[scale=0.76]{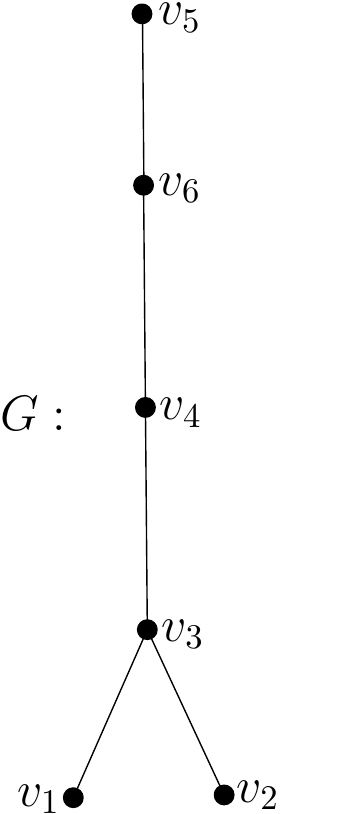}} \hspace{0,4cm} 
\subfigure[]{ \label{fig:3Db}
\hspace{-0.4cm}\includegraphics[scale=0.76]{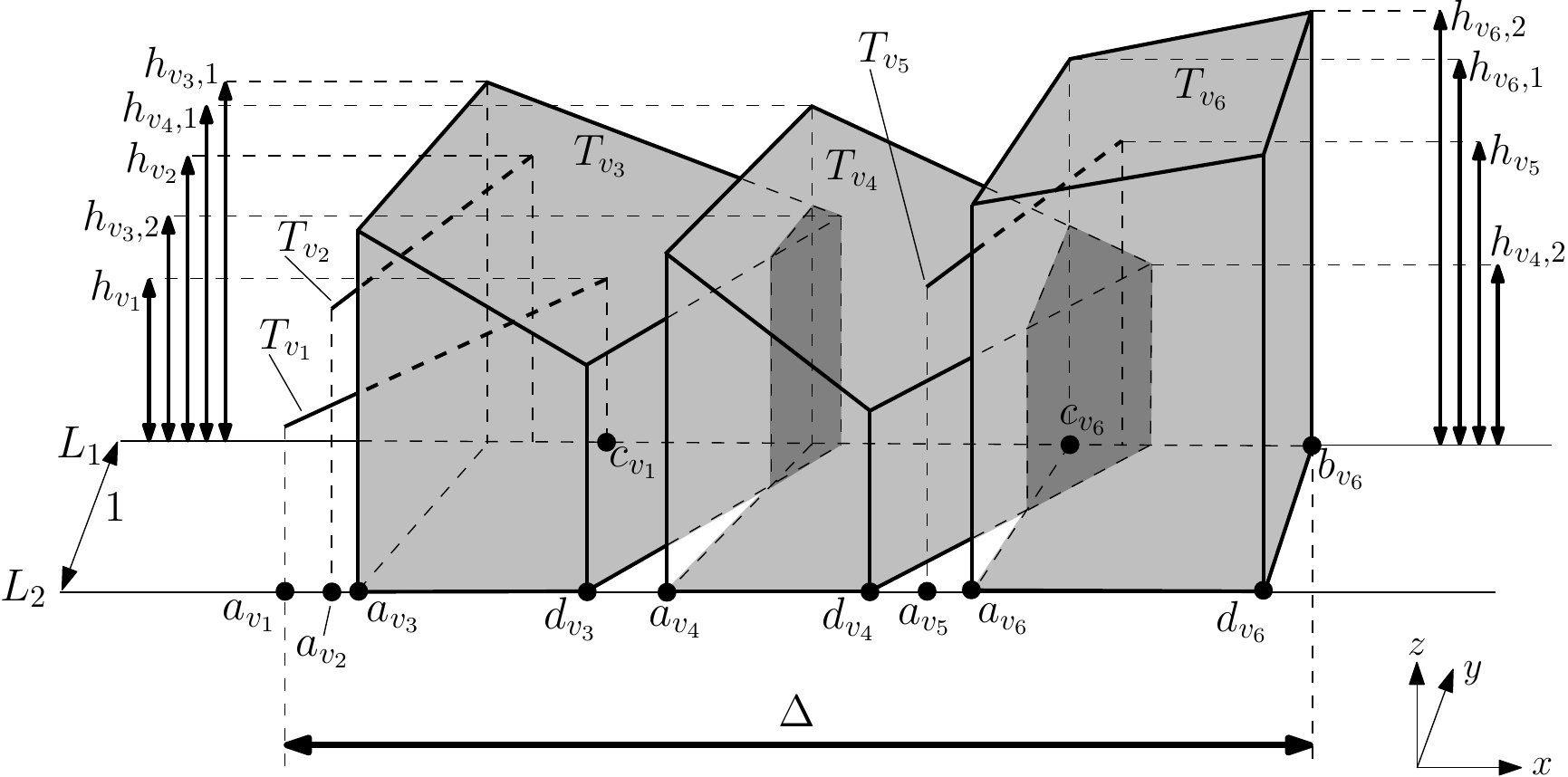}}
\caption{(a) A multitolerance graph $G$ and (b) a trapezoepiped
representation $R$ of $G$. Here, ${h_{v_{i},j}=\Delta -\cot \protect\phi %
_{v_{i},j}}$ for every bounded vertex ${v_{i}\in V_{B}}$ and ${j\in \{1,2\}}$%
, while ${h_{v_{i}}=\Delta -\cot \protect\phi _{v_{i}}}$ for every unbounded
vertex ${v_{i}\in V_{U}}$.}
\label{fig:3D}
\end{figure}

\begin{definition}[\hspace{-0,01mm}\protect\cite{Multitol-Mertzios14}]
\label{inevitable-canonical-def}An unbounded vertex $v\in V_{U}$ is \emph{%
inevitable} if replacing $T_{v}$ by $H_{\text{convex}}(T_{v},\overline{T}%
_{v})$ creates a new edge $uv$ in $G$; then $u$ is a \emph{hovering vertex}
of $v$ and the set $H(v)$ of all hovering vertices of $v$ is the \emph{%
hovering set} of $v$. A trapezoepiped representation of a multitolerance
graph $G$ is called \emph{canonical} if every unbounded vertex is inevitable.
\end{definition}

In the example of Figure~\ref{fig:3D}, $v_{2}$ and $v_{5}$ are inevitable
unbounded vertices, $v_{1}$ and $v_{4}$ are hovering vertices of $v_{2}$ and 
$v_{5}$, respectively, while $v_{1}$ is not an inevitable unbounded vertex.
Therefore, this representation is not canonical for the graph $G$. However,
if we replace~$T_{v_{1}}$ by $H_{\text{convex}%
}(T_{v_{1}},a_{v_{1}},c_{v_{1}})$, we get a canonical representation for $G$
in which vertex $v_{1}$ is bounded.

\begin{lemma}[\hspace{-0,01mm}\protect\cite{Multitol-Mertzios14}]
\label{neighbors-hovering}Let~$v\in V_{U}$ be an inevitable unbounded vertex
of a multitolerance graph~$G$. Then $N(v)\subseteq N(u)$ for every hovering
vertex~$u\in H(v)$ of~$v$.
\end{lemma}

\begin{lemma}[\hspace{-0,01mm}\protect\cite{Multitol-Mertzios14}]
\label{hovering-bounded}Let $R$ be a canonical representation of a
multitolerance graph~$G$ and~$v\in V_{U}$ be an (inevitable) unbounded
vertex of~$G$. Then there exists a hovering vertex~$u$ of~$v$, which is
bounded.
\end{lemma}

Recall that $\{\overline{T}_{v}:v\in V_{B}\}$ is an intersection model of
the induced subgraph $G[V_{B}]$ on the bounded vertices of $G$, i.e.,~$uv\in
E$ if and only if $\overline{T}_{u}\cap \overline{T}_{v}\neq \emptyset $
where $u,v\in V_{B}$. Furthermore, although $\{\overline{T}_{v}:v\in
V=V_{B}\cup V_{U}\}$ is \emph{not} an intersection model of $G$, it still
provides the whole information about the adjacencies of the vertices of $G$,
cf.~Lemma~\ref{neighbors-bounded-unbounded}. For Lemma~\ref%
{neighbors-bounded-unbounded} we need the next definition of the angles $%
\phi _{u}(x)$, where $u\in V_{B}$ and $a_{u}\leq x\leq d_{u}$, cf.~Figure~%
\ref{TuTv-modified-fig} for an illustration.

\begin{definition}[\hspace{-0,01mm}\protect\cite{Multitol-Mertzios14}]
\label{phi(x)}Let $u\in V_{B}$ be a bounded vertex and $%
a_{u},b_{u},c_{u},d_{u}$ be the endpoints of the trapezoid $\overline{T}_{u}$%
. Let $x\in \lbrack a_{u},d_{u}]$ and $y\in \lbrack c_{u},b_{u}]$ be two
points on the lines~$L_{2}$ and~$L_{1}$, respectively, such that $x=\lambda
a_{u}+(1-\lambda )d_{u}$ and $y=\lambda c_{u}+(1-\lambda )b_{u}$ for the
same value~${\lambda \in \lbrack 0,1]}$. Then $\phi _{u}(x)$ is the angle of
the line segment with endpoints $x$ and $y$ on the lines~$L_{2}$ and~$L_{1}$%
, respectively.
\end{definition}

\begin{lemma}[\hspace{-0,01mm}\protect\cite{Multitol-Mertzios14}]
\label{neighbors-bounded-unbounded}Let $u\in V_{B}$ and $v\in V_{U}$ in a
trapezoepiped representation of a multitolerance graph $G=(V,E)$. Let $a_{u}$%
, $d_{u}$, and $a_{v}=d_{v}$ be the endpoints of $\overline{T}_{u}$ and $%
\overline{T}_{v}$, respectively, on~$L_{2}$.~Then:

\begin{itemize}
\item if $a_{v}<a_{u}$, then $uv\in E$ if and only if $\overline{T}_{u}\cap 
\overline{T}_{v}\neq \emptyset $,

\item if $a_{u}<a_{v}<d_{u}$, then $uv\in E$ if and only if $\phi
_{v}\leq\phi _{u}(a_{v})$,

\item if $d_{u}<a_{v}$, then $uv\notin E$.
\end{itemize}
\end{lemma}

\section{The new geometric representations\label{representations-sec}}

In this section we introduce new geometric representations on the plane for
both tolerance and multitolerance graphs. The new representation of
tolerance graphs is called the \emph{horizontal shadow representation},
which is given by a set of points and horizontal line segments in the plane.
The horizontal shadow representation can be naturally extended to general
multitolerance graphs, in which case the line segments are not necessarily
horizontal; we call this representation of multitolerance graphs the \emph{%
shadow representation}. In the remainder of this section, we present the
shadow representation of general multitolerance graphs, since the horizontal
shadow representation of tolerance graphs is just the special case, in which
every line segment is horizontal.

\begin{definition}[shadow representation]
\label{shadow-representation-def}Let $G=(V,E)$ be a multitolerance graph, $R$
be a trapezoepiped representation of $G$, and $V_{B},V_{U}$ be the sets of
bounded and unbounded vertices of $G$ in~$R$, respectively. We associate the
vertices of $G$ with points and line segments in the plane as follows:

\begin{itemize}
\item for every $v\in V_{B}$, the points $p_{v,1}=(a_{v},\Delta -\cot \phi
_{v,1})$ and $p_{v,2}=(d_{v},\Delta -\cot \phi _{v,2})$ and the line segment 
$L_{v}=(p_{v,1},p_{v,2})$,

\item for every $v\in V_{U}$, the point $p_{v}=(a_{v},\Delta -\cot \phi_{v})$%
.
\end{itemize}

The tuple $(\mathcal{P},\mathcal{L})$, where $\mathcal{L}=\{L_{v}:v\in
V_{B}\}$ and $\mathcal{P}=\{p_{v}:v\in V_{U}\}$, is the \emph{shadow
representation} of~$G$. If $\phi _{v,1}=\phi _{v,2}$ for every $v\in V_{B}$,
then $(\mathcal{P},\mathcal{L})$ is the \emph{horizontal shadow
representation} of the tolerance graph $G$. Furthermore, the representation $%
(\mathcal{P},\mathcal{L})$ is \emph{canonical} if the initial trapezoepiped
representation $R$ is also canonical.
\end{definition}

Note by Definition~\ref{shadow-representation-def} that, given a
trapezoepiped (resp.~parallelepiped) representation of a multitolerance
(resp.~tolerance) graph $G$ with $n$ vertices, we can compute a shadow
(resp.~horizontal shadow) representation of $G$ in $O(n)$ time. As an
example for Definition~\ref{shadow-representation-def}, we illustrate in
Figure~\ref{shadow-representation-fig} the shadow representation $(\mathcal{P%
},\mathcal{L})$ of the multitolerance graph $G$ of Figure~\ref{fig:3D}.

\begin{figure}[h]
\centering 
\includegraphics[scale=0.8]{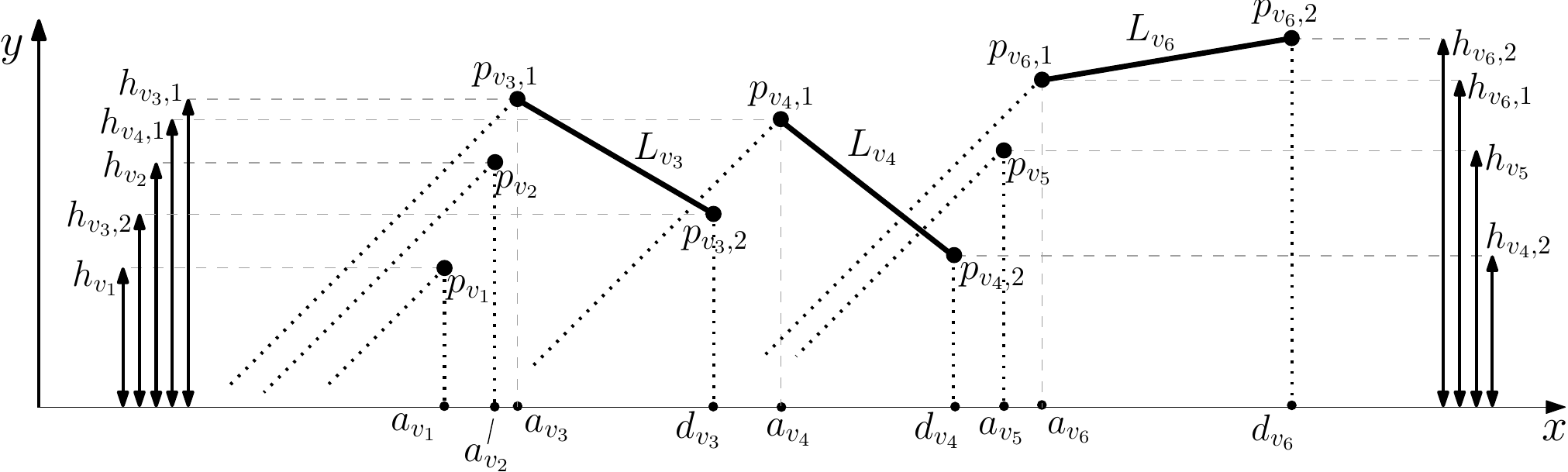}
\caption{The shadow representation $(\mathcal{P},\mathcal{L})$ of the
multitolerance graph $G$ of Figure~\protect\ref{fig:3D}. The unbounded
vertices $V_{U}=\{v_{1},v_{2},v_{5}\}$ and the bounded vertices $%
V_{B}=\{v_{3},v_{4},v_{6}\}$ are associated with the points $\mathcal{P}%
=\{p_{v_1},p_{v_2},p_{v_5}\}$ and with the line segments $\mathcal{L}%
=\{L_{v_1},L_{v_2},L_{v_5}\}$, respectively.}
\label{shadow-representation-fig}
\end{figure}

\begin{observation}
\label{segment-middle-obs}In Definition~\ref{shadow-representation-def}, $%
L_{v}=\{(x,\Delta -\cot \phi _{v}(x)):a_{v}\leq x\leq d_{v}\}$ for every
bounded vertex $v\in V_{B}$ of the multitolerance graph $G$.
\end{observation}

Now we introduce the notions of the \emph{shadow} and the \emph{reverse
shadow} of points and of line segments in the plane; an example is
illustrated in Figure~\ref{shadow-point-segment-fig}.

\begin{definition}[shadow]
\label{shadows-def}For an arbitrary point $t=(t_{x},t_{y})\in \mathbb{R}^{2}$
the \emph{shadow} of $t$ is the region $S_{t}=\{(x,y)\in \mathbb{R}%
^{2}:x\leq t_{x},\ y-x\leq t_{y}-t_{x}\}$. Furthermore, for every line
segment $L_{u}$, where $u\in V_{B}$, the \emph{shadow} of $L_{u}$ is the
region $S_{u}=\bigcup\nolimits_{t\in L_{u}}S_{t}$.
\end{definition}

\begin{definition}[reverse shadow]
\label{reverse-shadows-def}For an arbitrary point $t=(t_{x},t_{y})\in 
\mathbb{R}^{2}$ the \emph{reverse shadow} of $t$ is the region $%
F_{t}=\{(x,y)\in \mathbb{R}^{2}:x\geq t_{x},\ y-x\geq t_{y}-t_{x}\}$.
Furthermore, for every line segment $L_{i}$, where $u\in V_{B}$, the \emph{%
reverse shadow} of $L_{i}$ is the region $F_{i}=\bigcup\nolimits_{t\in
L_{i}}F_{t}$.
\end{definition}

\begin{figure}[h]
\centering%
\subfigure[]{ \label{shadow-point-fig}
\includegraphics[scale=0.69]{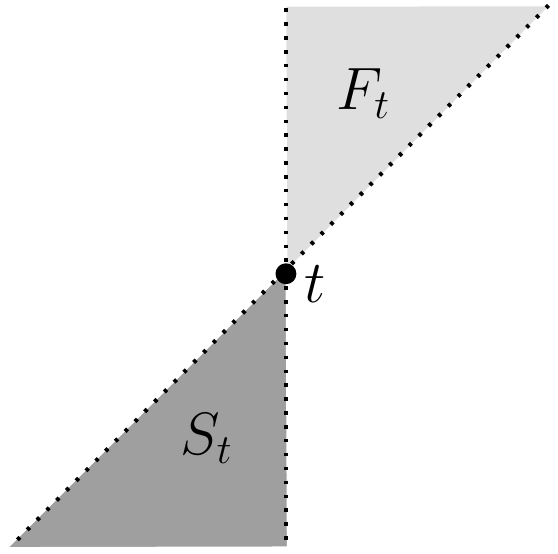}} \hspace{0,2cm} 
\subfigure[]{ \label{shadow-segment-fig}
\hspace{-0.4cm}\includegraphics[scale=0.68]{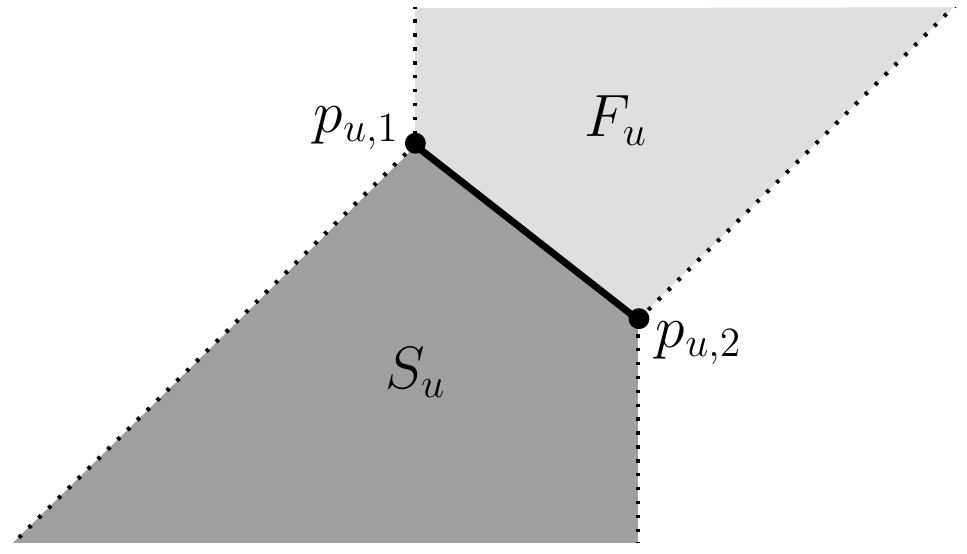}}
\caption{The shadow and the reverse shadow of (a)~a point $t\in\mathbb{R}%
^{2} $ and (b)~a line segment $L_u$. }
\label{shadow-point-segment-fig}
\end{figure}

\begin{lemma}
\label{slope-line-segments-lem}Let $G$ be a multitolerance graph and $(%
\mathcal{P},\mathcal{L})$ be a shadow representation of $G$. Let $u\in V_{B}$
be a bounded vertex of $G$ such that the corresponding line segment $L_{u}$
is not trivial, i.e.,~$L_{u}$ is not a single point. Then the angle of the
line segment $L_{u}$ with a horizontal line (i.e.,~parallel to the $x$-axis)
is at most $\frac{\pi }{4}$ and at least $-\frac{\pi }{2}$.
\end{lemma}

\begin{proof}
The two endpoints of $L_{u}$ are the points $(a_{u},\Delta -\cot \phi
_{u,1}) $ and $(d_{u},\Delta -\cot \phi _{u,2})$. For the purposes of the
proof, denote by $\psi $ the angle of the line segment $L_{u}$ with a
horizontal line (i.e.,~parallel to the $x$-axis). To prove that $\psi \geq -%
\frac{\pi }{2}$ it suffices to observe that $a_{u}\leq d_{u}$ (cf.~Figure~%
\ref{TuTv-modified-fig}). To prove that $\psi \leq \frac{\pi }{4}$ it
suffices to show that $(\Delta -\cot \phi _{u,2})-(\Delta -\cot \phi
_{u,1})\leq d_{u}-a_{u}$, or equivalently to show that $(\Delta
-(b_{u}-d_{u}))-(\Delta -(c_{u}-a_{u}))\leq d_{u}-a_{u}$. The latter
inequality is equivalent to $b_{u}\geq c_{u}$, which is always true
(cf.~Figure~\ref{TuTv-modified-fig}).
\end{proof}

\medskip

Recall now that two unbounded vertices $u,v\in V_{U}$ are never adjacent.
The connection between a multitolerance graph $G$ and a shadow
representation of it is the following. Two bounded vertices $u,v\in V_{B}$
are adjacent if and only if $L_{u}\cap S_{v}\neq \emptyset $ or $L_{v}\cap
S_{u}\neq \emptyset $, cf.~Lemma~\ref{shadow-correctness-lem-1}. A bounded
vertex $u\in V_{B}$ and an unbounded vertex $v\in V_{U}$ are adjacent if and
only if $p_{v}\in S_{u}$, cf.~Lemma~\ref{shadow-correctness-lem-2}.

\begin{lemma}
\label{shadow-correctness-lem-1}Let $(\mathcal{P},\mathcal{L})$ be a shadow
representation of a multitolerance graph $G$. Let $u,v\in V_{B}$ be two
bounded vertices of $G$. Then $uv\in E$ if and only if $L_{v}\cap S_{u}\neq
\emptyset $ or $L_{u}\cap S_{v}\neq \emptyset $.
\end{lemma}

\begin{proof}
Let $R$ be the trapezoepiped representation of $G$, from which the shadow
representation $(\mathcal{P},\mathcal{L})$ is constructed, cf.~Definition~%
\ref{shadow-representation-def}.

($\Rightarrow $) Let $uv\in E$. Assume first that the intervals $%
[a_{u},d_{u}]$ and $[a_{v},d_{v}]$ of the $x$-axis share at least one common
point, say $t_{x}$. If $\phi _{v}(t_{x})\leq \phi _{u}(t_{x})$, then the
point $(t_{x},\Delta -\cot \phi _{v}(t_{x}))$ of the line segment $L_{v}$
belongs to the shadow $S_{u}$ of the line segment $L_{u}$, i.e.,~$L_{v}\cap
S_{u}\neq \emptyset $. Otherwise, symmetrically, if $\phi _{v}(t)>\phi
_{u}(t)$ then $L_{u}\cap S_{v}\neq \emptyset $.

Assume now that $[a_{u},d_{u}]$ and $[a_{v},d_{v}]$ are disjoint,
i.e.,~either $d_{u}<a_{v}$ or $d_{v}<a_{u}$. Without loss of generality we
may assume that $d_{u}<a_{v}$, as the other case is symmetric. Then, as $%
uv\in E$ by assumption, it follows that $\overline{T}_{u}\cap \overline{T}%
_{v}\neq \emptyset $ in the trapezoepiped representation $R$ of $G$. Thus $%
b_{u}\geq c_{v}$, since we assumed that $d_{u}<a_{v}$. Therefore $\cot \phi
_{u}=b_{u}-d_{u}\geq c_{v}-d_{u}=\cot \phi _{v,1}+(a_{v}-d_{u})$. That is, $%
(\Delta -\cot \phi _{u,2})-d_{u}\leq (\Delta -\cot \phi _{v,1})-a_{v}$, and
thus the point $(d_{u},\Delta -\cot \phi _{u,2})$ of the line segment $L_{u}$
belongs to the shadow $S_{t}$ of the point $t=(a_{v},\Delta -\cot \phi
_{v,1})$ of the line segment $L_{v}$. Therefore $L_{u}\cap S_{v}\neq
\emptyset $.

($\Leftarrow $) Let $L_{v}\cap S_{u}\neq \emptyset $ or $L_{u}\cap S_{v}\neq
\emptyset $. Assume first that the intervals $[a_{u},d_{u}]$ and $%
[a_{v},d_{v}]$ of the $x$-axis share at least one common point, say $t_{x}$.
Then $t_{x}\in \lbrack a_{u},d_{u}]\cap \lbrack a_{v},d_{v}]$, and thus the
trapezoids $\overline{T}_{u}$ and $\overline{T}_{v}$ in the trapezoepiped
representation $R$ have a common point on the line $L_{2}$, i.e.,~$\overline{%
T}_{u}\cap \overline{T}_{v}\neq \emptyset $. Therefore, since both $u$ and $%
v $ are bounded vertices, it follows that $uv\in E$.

Assume now that $[a_{u},d_{u}]$ and $[a_{v},d_{v}]$ are disjoint,
i.e.,~either $d_{v}<a_{u}$ or $d_{u}<a_{v}$. Without loss of generality we
may assume that $d_{v}<a_{u}$, as the other case is symmetric. Then $%
L_{u}\cap S_{v}=\emptyset $, and thus $L_{v}\cap S_{u}\neq \emptyset $.
Therefore, by Lemma~\ref{slope-line-segments-lem}, it follows that the point 
$t=(d_{v},\Delta -\cot \phi _{v,2})$ of $L_{v}$ must belong to $S_{u}$. In
particular, this point $t$ of $L_{v}$ must belong to the shadow $%
S_{t^{\prime }}$ of the point $t^{\prime }=(a_{u},\Delta -\cot \phi _{u,1})$
of $L_{u}$. That is, $(\Delta -\cot \phi _{v,2})-d_{v}\leq (\Delta -\cot
\phi _{u,1})-a_{u}$. It follows that $(b_{v}-d_{v})=\cot \phi _{v,2}\geq
\cot \phi _{u,1}+(a_{u}-d_{v})=(c_{u}-a_{u})+(a_{u}-d_{v})$, and thus $%
b_{v}\geq c_{u}$. Therefore, since $d_{v}<a_{u}$, it follows that $\overline{%
T}_{u}\cap \overline{T}_{v}\neq \emptyset $, and thus $uv\in E$.
\end{proof}

\begin{lemma}
\label{shadow-correctness-lem-2}Let $(\mathcal{P},\mathcal{L})$ be a shadow
representation of a multitolerance graph $G$. Let $v\in V_{U}$ and $u\in
V_{B}$ be two vertices of $G$. Then $uv\in E$ if and only if $p_{v}\in S_{u}$%
.
\end{lemma}

\begin{proof}
Let $R$ be the trapezoepiped representation of $G$, from which the shadow
representation $(\mathcal{P},\mathcal{L})$ is constructed, cf.~Definition~%
\ref{shadow-representation-def}. Furthermore recall that $%
p_{v}=(a_{v},\Delta -\cot \phi _{v})$ by Definition~\ref%
{shadow-representation-def}.

($\Rightarrow $) Let $uv\in E$. If $d_{u}<a_{v}$, then $uv\notin E$ by Lemma~%
\ref{neighbors-bounded-unbounded}, which is a contradiction. Therefore $%
a_{v}<d_{u}$. Assume first that $a_{u}<a_{v}<d_{u}$. Then Lemma~\ref%
{neighbors-bounded-unbounded} implies that $\phi _{v}\leq \phi _{u}(a_{v})$.
Thus it follows by Observation~\ref{segment-middle-obs} that $p_{v}\in S_{u}$%
. Assume now that $a_{v}<a_{u}$. Then Lemma~\ref{neighbors-bounded-unbounded}
implies that $\overline{T}_{u}\cap \overline{T}_{v}\neq \emptyset $. Thus $%
b_{v}\geq c_{u}$, since $a_{v}<a_{u}$. Therefore $\cot \phi
_{v}=(b_{v}-a_{v})\geq (a_{u}-a_{v})+(c_{u}-a_{u})=(a_{u}-a_{v})+\cot \phi
_{u,1}$. That is, $(\Delta -\cot \phi _{v})-a_{v}\leq (\Delta -\cot \phi
_{u,1})-a_{u}$, and thus the point $p_{v}=(a_{v},\Delta -\cot \phi _{v})$
belongs to the shadow $S_{t}$, where $t=(a_{u},\Delta -\cot \phi _{u,1})\in
L_{u}$, i.e.,~$p_{v}\in S_{u}$.

($\Leftarrow $) Let $p_{v}\in S_{u}$. Then clearly $a_{v}\leq d_{u}$. Assume
first that $a_{u}\leq a_{v}\leq d_{u}$. Then, since $p_{v}\in S_{u}$, it
follows by Observation~\ref{segment-middle-obs} that $\Delta -\cot \phi
_{v}\leq \Delta -\cot \phi _{u}(a_{v})$, and thus $\phi _{v}\leq \phi
_{u}(a_{v})$. Therefore Lemma~\ref{neighbors-bounded-unbounded} implies that 
$uv\in E$.

Assume now that $a_{v}<a_{u}$. Then, since $p_{v}\in S_{u}$, it follows that 
$p_{v}\in S_{t}$, where $t=(a_{u},\Delta -\cot \phi _{u,1})\in L_{u}$. Thus $%
(\Delta -\cot \phi _{v})-a_{v}\leq (\Delta -\cot \phi _{u,1})-a_{u}$. That
is, $(b_{v}-a_{v})=\cot \phi _{v}\geq (a_{u}-a_{v})+\cot \phi
_{u,1}=(a_{u}-a_{v})+(c_{u}-a_{u})$, and thus $b_{v}\geq c_{u}$. Therefore,
since $a_{v}<a_{u}$, it follows that $\overline{T}_{u}\cap \overline{T}%
_{v}\neq \emptyset $, and thus $uv\in E$ by Lemma~\ref%
{neighbors-bounded-unbounded}.
\end{proof}

\medskip

Lemmas~\ref{shadow-correctness-lem-1} and~\ref{shadow-correctness-lem-2}
show how adjacencies between vertices can be seen in a shadow representation 
$(\mathcal{P},\mathcal{L})$ of a multitolerance graph $G$. The next lemma
describes how the hovering vertices of an unbounded vertex $v\in V_{U}$ (cf.
Definition~\ref{inevitable-canonical-def}) can be seen in a shadow
representation $(\mathcal{P},\mathcal{L})$.

\begin{lemma}
\label{shadow-hovering-lem}Let $(\mathcal{P},\mathcal{L})$ be a shadow
representation of a multitolerance graph $G$. Let $v\in V_{U}$ be an
unbounded vertex of $G$ and $u\in V\setminus \{v\}$ be another arbitrary
vertex. If $u\in V_{B}$ (resp.~$u\in V_{U}$), then $u$ is a hovering vertex
of $v$ if and only if $L_{u}\cap S_{v}\neq \emptyset $ (resp.~$p_{u}\in
S_{v} $).
\end{lemma}

\begin{proof}
Let $G=(V,E)$ and $R$ be the trapezoepiped representation of $G$, from which
the shadow representation $(\mathcal{P},\mathcal{L})$ is constructed, cf.
Definition~\ref{shadow-representation-def}.

($\Leftarrow $) Let $u$ be a hovering vertex of $v$. That is, if we replace
in the trapezoepiped representation $R$ the line segment $T_{v}$ by $H_{%
\text{convex}}(T_{v},\overline{T}_{v})$ (i.e.,~if we make $v$ a bounded
vertex) then the vertices $u$ and $v$ become adjacent in the resulting
trapezoepiped representation $R^{\prime }$. Denote the new graph by $%
G^{\prime }=(V,E\cup \{uv\})$, i.e.,~$R^{\prime }$ is a trapezoepiped
representation of $G^{\prime }$. Note here that, since both $T_{v}$ and $%
\overline{T}_{v}$ are line segments, $H_{\text{convex}}(T_{v},\overline{T}%
_{v})$ is a degenerate trapezoepiped which is 2-dimensional.

Consider the shadow representation $(\mathcal{P}^{\prime },\mathcal{L}%
^{\prime })$ of $G^{\prime }$ that is obtained by this new trapezoepiped
representation $R^{\prime }$. Note that $\mathcal{P}^{\prime }=\mathcal{P}%
\setminus \{p_{v}\}$ and $\mathcal{L}^{\prime }=\mathcal{L}\cup \{L_{v}\}$,
where $L_{v}$ is a trivial line segment that consists of only one point $%
p_{v}$. Assume first that $u\in V_{U}$. Then, since $v$ is bounded and $v$
is adjacent to $u$ in $G^{\prime }$, Lemma~\ref{shadow-correctness-lem-2}
implies that $p_{u}\in S_{v}$. Assume now that $u\in V_{B}$. Then, since $v$
is bounded and $v$ is adjacent to $u$ in $G^{\prime }$, Lemma~\ref%
{shadow-correctness-lem-1} implies that $L_{v}\cap S_{u}\neq \emptyset $ or $%
L_{u}\cap S_{v}\neq \emptyset $. That is, $p_{v}\in S_{u}$ or $L_{u}\cap
S_{v}\neq \emptyset $, since $L_{v}=\{p_{v}\}$. If $p_{v}\in S_{u}$ then $u$
and $v$ are adjacent in $G$, by Lemma~\ref{shadow-correctness-lem-2}, which
is a contradiction. Therefore $L_{u}\cap S_{v}\neq \emptyset $.

($\Rightarrow $) Consider the shadow representation $(\mathcal{P}^{\prime },%
\mathcal{L}^{\prime })$ that is obtained by the shadow representation $(%
\mathcal{P},\mathcal{L})$ of $G$, such that $\mathcal{P}^{\prime }=\mathcal{P%
}\setminus \{p_{v}\}$ and $\mathcal{L}^{\prime }=\mathcal{L}\cup \{L_{v}\}$,
where $L_{v}$ is a trivial line segment that consists of only one point $%
p_{v}$. Then $(\mathcal{P}^{\prime },\mathcal{L}^{\prime })$ is a shadow
representation of some multitolerance graph $G^{\prime }$, where the bounded
vertices $V_{B}^{\prime }$ of $G^{\prime }$ correspond to the line segments
of $\mathcal{L}^{\prime }$ and the unbounded vertices $V_{U}^{\prime }$ of $%
G^{\prime }$ correspond to the points of $\mathcal{P}^{\prime }$.
Furthermore note that $V_{B}^{\prime }=V_{B}\cup \{v\}$ and $V_{U}^{\prime
}=V_{U}\setminus \{v\}$.

Assume first that $u\in V_{B}^{\prime }$ and $L_{u}\cap S_{v}\neq \emptyset $%
. Then, since both $u,v\in V_{B}^{\prime }$, Lemma~\ref%
{shadow-correctness-lem-1} implies that $u$ and $v$ are adjacent in $%
G^{\prime }$. Thus, since $u$ is not adjacent to $v$ in $G$, it follows that 
$u$ is a hovering vertex of $v$. Assume now that $u\in V_{U}^{\prime }$ and $%
p_{u}\in S_{v}$. Then, since both $v\in V_{B}^{\prime }$, Lemma~\ref%
{shadow-correctness-lem-2} implies that $u$ and $v$ are adjacent in $%
G^{\prime }$. Thus, similarly, $u$ is a hovering vertex of $v$.
\end{proof}

\medskip

In the example of Figure~\ref{shadow-representation-fig} the shadows of the
points in $\mathcal{P}$ and of the line segments in $\mathcal{L}$ are shown
with dotted lines. For instance, $p_{v_{2}}\in S_{v_{3}}$ and $%
p_{v_{2}}\notin S_{v_{4}}$, and thus the unbounded vertex $v_{2}$ is
adjacent to the bounded vertex $v_{3}$ but not to the bounded vertex $v_{4}$%
. Furthermore $L_{v_{3}}\cap S_{v_{4}}\neq \emptyset $, and thus $v_{3}$ and 
$v_{4}$ are adjacent. On the other hand, $L_{v_{3}}\cap
S_{v_{6}}=L_{v_{6}}\cap S_{v_{3}}=\emptyset $, and thus $v_{3}$ and $v_{4}$
are not adjacent. Finally $p_{v_{1}}\in S_{v_{2}}$ and $L_{v_{4}}\cap
S_{v_{5}}\neq \emptyset $, and thus $v_{1}$ is a hovering vertex of $v_{2}$
and $v_{4}$ is a hovering vertex of $v_{5}$. These facts can be also checked
in the trapezoepiped representation of the same multitolerance graph $G$ in
Figure~\ref{fig:3Db}.

\section{Dominating set is \textsf{APX}-hard on multitolerance graphs\label%
{dominating-hard-multitolerance-sec}}

In this section we prove that the dominating set problem on multitolerance
graphs is \textsf{APX}-hard. 
Let us first recall that an optimization problem $P_{1}$ is \emph{$L$%
-reducible} to an optimization problem $P_{2}$~\cite{Williamson11} if there
exist two functions $f$ and $g$, which are computable in polynomial time,
and two constants $\alpha ,\beta >0$ such that:

\begin{itemize}
\item for any instance $\mathcal{I}$ of $P_{1}$, $f(\mathcal{I})$ is an
instance of $P_{2}$ and $\text{OPT}(f(\mathcal{I}))\leq \alpha\cdot \text{OPT%
}(\mathcal{I})$, and

\item for any feasible solution $D$ of $f(\mathcal{I})$, $g(D)$ is a
feasible solution of $\mathcal{I}$, and it holds that $|c(g(D))-\text{OPT}(%
\mathcal{I})|\leq \beta \cdot |c(D)-\text{OPT}(f(\mathcal{I}))|$, where $%
c(D) $ and $c(g(D))$ denote the costs of the solutions $D$ and $g(D)$,
respectively.
\end{itemize}

Let us now define a special case of the unweighted set cover problem, namely
the \textsc{Special 3-Set Cover (S3SC)} problem~\cite{ChanG14}.

\begin{theorem}[\hspace{-0,01mm}\protect\cite{ChanG14}]
\label{thm:apxssc} \textsc{Special 3-Set Cover} is \textsf{APX}-hard.
\end{theorem}

\vspace{0,2cm} \noindent \fbox{ 
\begin{minipage}{0.96\textwidth}
 \begin{tabular*}{\textwidth}{@{\extracolsep{\fill}}lr} \sscprob & \\ \end{tabular*}
 
 \vspace{1.2mm}
{\bf{Input:}} A pair $({\cal U},{\cal S})$ consisting of a universe ${\cal U}=A\cup W\cup X\cup Y\cup Z$, and  a family ${\cal S}$ of subsets of ${\cal U}$ such that:
\begin{itemize}\item the sets $A$, $W$, $X$, $Y$, $Z$ are disjoint,
\item $A=\{a_{i}: i\in [n]\}$, $W=\{w_{i}: i\in [m]\}$, $X=\{x_{i}: i\in [m]\}$, $Y=\{y_{i}: i\in [m]\}$, $Z=\{z_{i}: i\in [m]\}$,
\item $2n=3m$,
\item for all $t\in [n]$, the element $a_{t}$ belongs to exactly two sets of ${\cal S}$, and
\item ${\cal S}$ has $5m$ sets; for every $t\in [m]$ there exist integers $1\leq i<j<k<n$ such that ${\cal S}$ contains the sets $\{a_{i},w_{t}\}$, $\{w_{t},x_{t}\},\{a_{j},x_{t},y_{t}\}$, $\{y_{t},z_{t}\}$, 
$\{a_{k},z_{t}\}$.\end{itemize}
{\bf{Output:}} A subset ${\cal S}_{0} \subseteq {\cal S}$ of minimum size such that 
every element in ${\cal U}$ belongs to at least one set of ${\cal S}_{0}$.
\end{minipage}} \vspace{0,2cm}

\begin{theorem}
\label{apx-hard-thm} \textsc{Dominating Set} is \textsf{APX}-hard on
Multitolerance Graphs.
\end{theorem}

\begin{proof}
From Theorem~\ref{thm:apxssc} it is enough to prove that \textsc{Special
3-Set Cover} is $L$-reducible to \textsc{Dominating Set} on Multitolerance
Graphs.\footnote{%
This proof is inspired by the proof of Theorem 1.1(C5) in~\cite{ChanG14}.}

Given an instance $\mathcal{I}=(\mathcal{U},\mathcal{S})$ of \textsc{Special
3-Set Cover} as above we construct a multitolerance graph $f(\mathcal{I})=(%
\mathcal{P},\mathcal{L})$, where $\mathcal{P}$ and $\mathcal{L}$ are the
sets of points and line segments in the shadow representation of $f(\mathcal{%
I})$, as follows. For every element $a_{i}\in A$, we create the point $%
p_{a_{i}}$ of $\mathcal{P}$ on the line $\{(z,-z):z>0\}$. Furthermore, for
every element $q\in W\cup X\cup Y\cup Z$, we create the point $p_{q}$ of $%
\mathcal{P}$ on the line $\{(t,\tan (\frac{\pi }{6})t):t<0\}$, such that for
every $i\in \lbrack m]$ the points that correspond to the elements $w_{i}$, $%
x_{i}$, $y_{i}$, and $z_{i}$ appear consecutively on this line (cf.~Figure~%
\ref{apx-reduction-fig}). Then, since every set of $\mathcal{S}$ contains at
most one element of $A$ and at most two elements of $W\cup X\cup Y\cup Z$,
it can be easily verified that we can construct for every set $Q_{j}\in 
\mathcal{S}$, $j\in \lbrack 5m]$, a line segment $L_{j}$ such that the
points of $\mathcal{P}$ that are contained within its shadow $S_{j}$ are
exactly the points of $\mathcal{P}$ that correspond to the elements of $%
Q_{j} $ (cf.~Figure~\ref{apx-reduction-fig}). Furthermore we construct an
additional line segment $L_{5m+1}$, with left endpoint $l_{5m+1}$ and right
endpoint $r_{5m+1}$, respectively, such that $l_{5m+1}$ (resp. $r_{5m+1}$)
lies below and to the left (resp. below and to the right) of every endpoint
of $\mathcal{P}\cup \{L_{1},L_{2},\ldots ,L_{5m}\}$. Then note that the line
segment $L_{5m+1}$ corresponds to a hovering vertex of every point $p\in 
\mathcal{P}$ in the multitolerance graph $f(\mathcal{I})$, cf.~Lemma~\ref%
{shadow-hovering-lem}. Moreover the line segment $L_{5m+1}$ is a neighbor to
all other line segments $\{L_{1},L_{2},\ldots ,L_{5m}\}$ in the
multitolerance graph $f(\mathcal{I})$, cf.~Lemma~\ref%
{shadow-correctness-lem-1}. Finally we add the line segment $L_{5m+2}$ such
that $L_{5m+1}$ is its only neighbor, cf.~Figure~\ref{apx-reduction-fig}.
This concludes the construction of the new instance $f(\mathcal{I})$.

\begin{figure}[tbp]
\centering 
\includegraphics[scale=0.7]{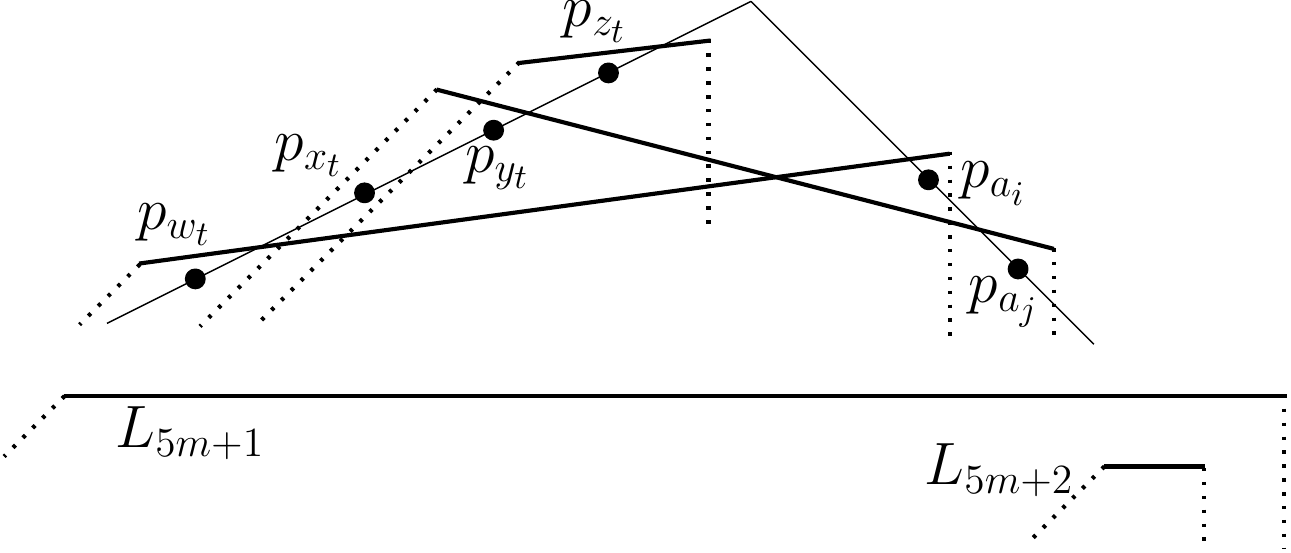}
\caption{The construction of the shadow representation in Theorem~\protect
\ref{apx-hard-thm}.}
\label{apx-reduction-fig}
\end{figure}

\begin{myclaim}
\label{clm:redpr1} $\text{OPT}(f(\mathcal{I}))\leq \text{OPT}(\mathcal{I})+1$%
, and thus $\text{OPT}(f(\mathcal{I}))\leq 2\cdot \text{OPT}(\mathcal{I})$.
\end{myclaim}

\begin{proofofclaim}[Proof of Claim~\protect\ref{clm:redpr1}]
Let $\mathcal{S}_{0}\subseteq \mathcal{S}$ be an optimum solution of an
instance $\mathcal{I}$ to \textsc{Special 3-Set Cover} and let $D$ be the
subset of $\mathcal{L}$ in the instance $f(\mathcal{I})$ of \textsc{%
Dominating Set}, where a line segment $L$ of $f(\mathcal{I})$ belongs to $D$
if and only if the corresponding set of $\mathcal{I}$ belongs to $S$. Let
now $D^{\prime }=D\cup \{L_{5m+1}\}$. As $S$ is an optimum solution of $%
\mathcal{I}$ it follows that all the elements of $\mathcal{U}$ belong to
some set of $S$ and from the construction of $f(\mathcal{I})$ it follows
that all points of $\mathcal{P}$ are contained inside the shadows of the
line segments in $D$. Thus, every point of $\mathcal{P}$ has a neighbor in $%
D $. Notice also that from the construction of $L_{5m+1}$ all line segments
of $\mathcal{L}$ have $L_{5m+1}$ as a neighbor. Therefore, as $|D|=|S|$ and $%
L_{5m+1}\notin D$, $D^{\prime }=D\cup \{L_{5m+1}\}$ is a solution to $f(%
\mathcal{I})$ of size $\text{OPT}(\mathcal{I})+1$. As \textsc{Dominating Set}
is a minimization problem we obtain that $\text{OPT}(f(\mathcal{I}))\leq
|D^{\prime }|=\text{OPT}(\mathcal{I})+1$.
\end{proofofclaim}

\medskip

We now define the function $g$ which, given a feasible solution $D$ of $f(%
\mathcal{I})$, returns a feasible solution $g(D)$ of $\mathcal{I}$. Let $D$
be a feasible solution of $f(\mathcal{I})$.

If $L_{5m+1}$ does not belong to $D$ then $L_{5m+2}$ belongs to $D$, since $%
L_{5m+1}$ the only neighbor of $L_{5m+2}$. By replacing $L_{5m+2}$ by $%
L_{5m+1}$ we obtain a solution of $f(\mathcal{I})$ of the same size. Thus,
without loss of generality we may assume that $L_{5m+1}$ belongs to $D$.
Furthermore, by the minimality of $D$ it follows that $D$ does not contain $%
L_{5m+2}$. Recall that all line segments $\{L_{1},L_{2},\ldots ,L_{5m}\}$
have $L_{5m+1}$ as a neighbor in $D$ and that every point $p$ of $f(\mathcal{%
I})$ is contained in the shadow of some line segment $L_{p}\in
\{L_{1},L_{2},\ldots ,L_{5m}\}$ in $f(\mathcal{I})$. Thus, for every point $%
p\in \mathcal{P}\cap D$, the set $(D\setminus \{p\})\cup \{L_{p}\}$ is also
a solution of $f(\mathcal{I})$ and has size at most $|D|$. Therefore,
without loss of generality we may also assume that $D$ only contains line
segments. As $L_{5m+1}\in D$ is not a neighbor of any point of $\mathcal{P}$
in $f(\mathcal{I})$, the set $D\setminus \{L_{5m+1}\}$ contains all
neighbors of the points of $f(\mathcal{I})$. Let $\mathcal{S}_{0}\subseteq 
\mathcal{S}$ contain all sets from $\mathcal{S}$ that correspond to the line
segments of $D\setminus \{L_{5m+1}\}$. From the construction of $f(\mathcal{I%
})$ we obtain that each element of $\mathcal{U}$ in $\mathcal{I}$ belongs to
at least one set of $\mathcal{S}_{0}$. We define $g(D)$ to be that set $%
\mathcal{S}_{0}$. Finally, notice that $|\mathcal{S}_{0}|\leq |D|-1$. This
implies the following simple observation.

\begin{observation}
\label{obs:solszs} If $D$ is a solution of $f(\mathcal{I})$, then $g(D)$ is
a solution of $\mathcal{I}$ and $c(g(D))\leq c(D)-1$.
\end{observation}

\begin{myclaim}
\label{clm:redpr22} $\text{OPT}(f(\mathcal{I}))= \text{OPT}(\mathcal{I})+1$.
\end{myclaim}

\begin{proofofclaim}[Proof of Claim~\protect\ref{clm:redpr22}]
Let $D$ be an optimum solution of $f(\mathcal{I})$. From Observation~\ref%
{obs:solszs}, we obtain that there exists a solution $S$ of $\mathcal{I}$
such that $|S|\leq \text{OPT}(f(\mathcal{I}))-1$. As \textsc{Special 3-Set
Cover} is a minimization problem it follows that $\text{OPT}(\mathcal{I}%
)\leq |S| \leq \text{OPT}(f(\mathcal{I}))-1$ and thus, $\text{OPT}(\mathcal{I%
})+1\leq \text{OPT}(f(\mathcal{I}))$. We now obtain the desired result from
Claim~\ref{clm:redpr1}.
\end{proofofclaim}

\medskip

We finally prove that $c(g(D))-\text{OPT}(\mathcal{I})\leq c(D)-\text{OPT}(f(%
\mathcal{I}))$. Notice that this is enough to prove the reduction for $%
\alpha =2$ (Claim~\ref{clm:redpr1}) and $\beta =1$. Claim~\ref{clm:redpr22}
yields that $c(g(D))-\text{OPT}(\mathcal{I})=c(g(D))-\text{OPT}(f(\mathcal{I}%
))+1$, and thus it follows by Observation~\ref{obs:solszs} that 
\begin{equation*}
c(g(D))-\text{OPT}(f(\mathcal{I}))+1\leq c(D)-1-\text{OPT}(f(\mathcal{I}%
))+1=c(D)-\text{OPT}(f(\mathcal{I})).
\end{equation*}%
This completes the proof of the theorem.
\end{proof}

\section{Bounded dominating set on tolerance graphs\label%
{Bounded-dominating-sec}}

In this section we use the \emph{horizontal shadow representation} of
tolerance graphs (cf.~Section~\ref{representations-sec}) to provide a
polynomial time algorithm for a variation of the minimum dominating set
problem on tolerance graphs, namely \textsc{Bounded Dominating Set},
formally defined below. This problem variation may be interesting on its
own, but we use our algorithm for \textsc{Bounded Dominating Set} as a
subroutine in our algorithm for the minimum dominating set problem on
tolerance graphs, cf.~Sections~\ref{Restricted-domination-sec} and~\ref%
{tolerance-domination-sec}. Note that, given a horizontal shadow
representation $(\mathcal{P},\mathcal{L})$ of a tolerance graph $G=(V,E)$,
the representation $(\mathcal{P},\mathcal{L})$ defines a partition of the
vertex set $V$ into the set $V_{B}$ of bounded vertices and the set $V_{U}$
of unbounded vertices. Indeed, every point of $\mathcal{P}$ corresponds to
an unbounded vertex in $V_{U}$ and every line segment of $\mathcal{L}$
corresponds to a bounded vertex of $V_{B}$. We denote $\mathcal{P}%
=\{p_{1},p_{2},\ldots ,p_{|\mathcal{P}|}\}$ and $\mathcal{L=}%
\{L_{1},L_{2},\ldots ,L_{|\mathcal{L}|}\}$, where $|\mathcal{P}|+|\mathcal{L}%
| = |V_{U}|+|V_{B}| = |V|$.

In this section we only deal with tolerance graphs and their horizontal shadow
representations. Thus, from now on, all line segments $\{L_{i}:1\leq i\leq |%
\mathcal{L}|\}$ will be assumed to be \emph{horizontal}. Furthermore, with a
slight abuse of notation, for any two elements $x_{1},x_{2}\in \mathcal{P}%
\cup \mathcal{L}$, we may say in the following that $x_{1}$ is adjacent with 
$x_{2}$ (or $x_{1}$ is a neighbor of $x_{2}$) if the vertices that
correspond to $x_{1}$ and $x_{2}$ are adjacent in the graph $G$. Moreover,
whenever $\mathcal{P}_{1}\subseteq \mathcal{P}_{2}\subseteq \mathcal{P}$ and 
$\mathcal{L}_{1}\subseteq \mathcal{L}_{2}\subseteq \mathcal{L}$, we may say
in the following that the set $\mathcal{P}_{1}\cup \mathcal{L}_{1}$ \emph{%
dominates} $\mathcal{P}_{2}\cup \mathcal{L}_{2}$ if the vertices that
correspond to $\mathcal{P}_{1}\cup \mathcal{L}_{1}$ are a dominating set of
the subgraph of $G$ induced by the vertices corresponding to $\mathcal{P}%
_{2}\cup \mathcal{L}_{2}$.

\vspace{0,2cm} \noindent \fbox{ 
\begin{minipage}{0.96\textwidth}
 \begin{tabular*}{\textwidth}{@{\extracolsep{\fill}}lr} \bdsprob & \\ \end{tabular*}
 
 \vspace{1.2mm}
{\bf{Input:}} A horizontal shadow representation $(\PP,\LL)$ of a tolerance graph $G$. \\
{\bf{Output:}} A set $Z\subseteq \LL$ of minimum size that dominates $(\PP,\LL)$, or the announcement that~$\LL$ does not dominate $(\PP,\LL)$.
\end{minipage}} \vspace{0,2cm}

Before we proceed with our polynomial time algorithm for \textsc{Bounded
Dominating Set} on tolerance graphs, we first provide some necessary
notation and terminology.

\subsection{Notation and terminology\label%
{terminology-bounded-domination-subsec}}

For an arbitrary point $t=(t_{x},t_{y})\in \mathbb{R}^{2}$ we define two
(infinite) lines passing through $t$:

\begin{itemize}
\item the vertical line $\Gamma _{t}^{\text{vert}}=\{(t_{x},s)\in \mathbb{R}%
^{2}:s\in \mathbb{R}\}$, i.e.,~the line that is parallel to the $y$-axis, and

\item the diagonal line $\Gamma _{t}^{\text{diag}}=\{(s,s+(t_{y}-t_{x}))\in 
\mathbb{R}^{2}:s\in \mathbb{R}\}$, i.e.,~the line that is parallel to the
main diagonal $\{(s,s)\in \mathbb{R}^{2}:s\in \mathbb{R}\}$.
\end{itemize}

The lines $\Gamma _{t}^{\text{vert}}$ and $\Gamma _{t}^{\text{diag}}$ are
illustrated in~Figure~\ref{At-Bt-fig} (see also Figure~\ref{shadow-point-fig}%
). For every point $t=(t_{x},t_{y})\in \mathbb{R}^{2}$, each of the lines $%
\Gamma _{t}^{\text{vert}},\Gamma _{t}^{\text{diag}}$ separates $\mathbb{R}%
^{2}$ into two regions. With respect to the line $\Gamma _{t}^{\text{vert}}$
we define the regions $\mathbb{R}_{\text{left}}^{2}(\Gamma _{t}^{\text{vert}%
})=\{(x,y)\in \mathbb{R}^{2}:x\leq t_{x}\}$ and $\mathbb{R}_{\text{right}%
}^{2}(\Gamma _{t}^{\text{vert}})=\{(x,y)\in \mathbb{R}^{2}:x\geq t_{x}\}$ of
points to the left and to the right of $\Gamma _{t}^{\text{vert}}$,
respectively. Similarly, with respect to the line $\Gamma _{t}^{\text{diag}}$%
, we define the regions ${\mathbb{R}_{\text{left}}^{2}(\Gamma _{t}^{\text{%
diag}})=\{(x,y)\in \mathbb{R}^{2}:y-x\geq t_{y}-t_{x}\}}$ and ${\mathbb{R}_{%
\text{right}}^{2}(\Gamma _{t}^{\text{diag}})=\{(x,y)\in \mathbb{R}%
^{2}:y-x\leq t_{y}-t_{x}\}}$ of points to the left and to the right of $%
\Gamma _{t}^{\text{diag}}$, respectively.

Furthermore, for an arbitrary point $t=(t_{x},t_{y})\in \mathbb{R}^{2}$ we
define the region $A_{t}$ (resp.~$B_{t}$) that contains all points that are
both to the right (resp.~to the left) of $\Gamma _{t}^{\text{vert}}$ and to
the right (resp.~to the left) of $\Gamma _{t}^{\text{diag}}$. That is,%
\vspace{-0,2cm}%
\begin{eqnarray*}
A_{t} &=&\mathbb{R}_{\text{right}}^{2}(\Gamma _{t}^{\text{vert}})\cap 
\mathbb{R}_{\text{right}}^{2}(\Gamma _{t}^{\text{diag}}), \\
B_{t} &=&\mathbb{R}_{\text{left}}^{2}(\Gamma _{t}^{\text{vert}})\cap \mathbb{%
R}_{\text{left}}^{2}(\Gamma _{t}^{\text{diag}}).
\end{eqnarray*}%
An example of the regions $A_{t}$ and $B_{t}$ is given in Figure~\ref%
{At-Bt-fig}, where $A_{t}$ (resp.~$B_{t}$) is the \emph{shaded region} of $%
\mathbb{R}^{2}$ that is to the right (resp.~to the left) of the point $t$.
Consider an arbitrary horizontal line segment $L_{i}\in \mathcal{L}$. We
denote by $l_{i}$ and $r_{i}$ its left and its right endpoint, respectively;
note that possibly $l_{i}=r_{i}$. Denote by $\mathcal{A}=\{l_{i},r_{i}:1\leq
i\leq |\mathcal{L}|\}$ the set of all endpoints of all line segments of $%
\mathcal{L}$. Furthermore denote by $\mathcal{B}=\{\Gamma _{t}^{\text{diag}%
}\cap \Gamma _{t^{\prime }}^{\text{vert}}:t,t^{\prime }\in \mathcal{A}\}$
the set of all intersection points of the vertical and the diagonal lines
that pass from points of $\mathcal{A}$. Note that $\mathcal{A\subseteq B}$.

Given a horizontal shadow representation $(\mathcal{P},\mathcal{L}\mathcal{)}
$ we always assume that the points $p_{1},p_{2},\ldots ,p_{|\mathcal{P}|}$
are ordered increasingly with respect to their $x$-coordinates. Similarly we
assume that the horizontal line segments $L_{1},L_{2},\dots ,L_{|\mathcal{L}%
|}$ are ordered increasingly with respect to the $x$-coordinates of their
endpoint $r_{i}$. That is, if $i<j$ then $p_{i}\in \mathbb{R}_{\text{left}%
}^{2}(\Gamma _{p_{j}}^{\text{vert}})$ and $r_{i}\in \mathbb{R}_{\text{left}%
}^{2}(\Gamma _{r_{j}}^{\text{vert}})$. Notice that, without loss of
generality, we may assume that all points of $\mathcal{P}$ and all endpoints
of the horizontal line segments in $\mathcal{L}$ have different $x$%
-coordinates.

\begin{definition}
\label{left-right-crossing-pair-def} Let $L_{i},L_{i^{\prime}},\in \mathcal{L%
}$ and let $L_{j},L_{j^{\prime}}\in \mathcal{L}$, where possibly ${%
i^{\prime}=i}$ and possibly ${j^{\prime}=j}$. The pair $(j,j^{\prime })$ is
a \emph{left-crossing pair} if~${l_{j}\in S_{l_{j^{\prime}}}}$. Furthermore
the pair $(i,i^{\prime })$ is a \emph{right-crossing pair} if $r_{i^{\prime
}}\in S_{r_{i}}$. For every left-crossing pair $(j,j^{\prime })$ we define%
\begin{equation*}
\mathcal{L}_{j,j^{\prime }}^{\text{right}}=\{x\in \mathcal{P}\cup \mathcal{L}%
:x\subseteq A_{t}\text{,\ where }t=\Gamma _{l_{j}}^{\text{vert}}\cap \Gamma
_{l_{j^{\prime }}}^{\text{diag}}\}
\end{equation*}%
and for every right-crossing pair $(i,i^{\prime })$ we define%
\begin{equation*}
\mathcal{L}_{i,i^{\prime }}^{\text{left}}=\{x\in \mathcal{P}\cup \mathcal{L}%
:x\subseteq B_{t}\text{,\ where }t=\Gamma _{r_{i}}^{\text{vert}}\cap \Gamma
_{r_{i^{\prime }}}^{\text{diag}}\}.
\end{equation*}%
Finally, for every line segment $L_{q}\in \mathcal{L}$ we define%
\begin{equation*}
\mathcal{L}_{q}^{\text{right}}=\{x\in \mathcal{P}\cup \mathcal{L}:x\subseteq 
\mathbb{R}_{\text{right}}^{2}(\Gamma _{l_{q}}^{\text{diag}})\}.
\end{equation*}
\end{definition}

Examples of left-crossing and right-crossing pairs (cf.~Definition~\ref%
{left-right-crossing-pair-def}) are illustrated in Figure~\ref%
{crossing-pairs-fig}.

\begin{figure}[t]
\centering%
\subfigure[]{ \label{At-Bt-fig}
\includegraphics[scale=0.68]{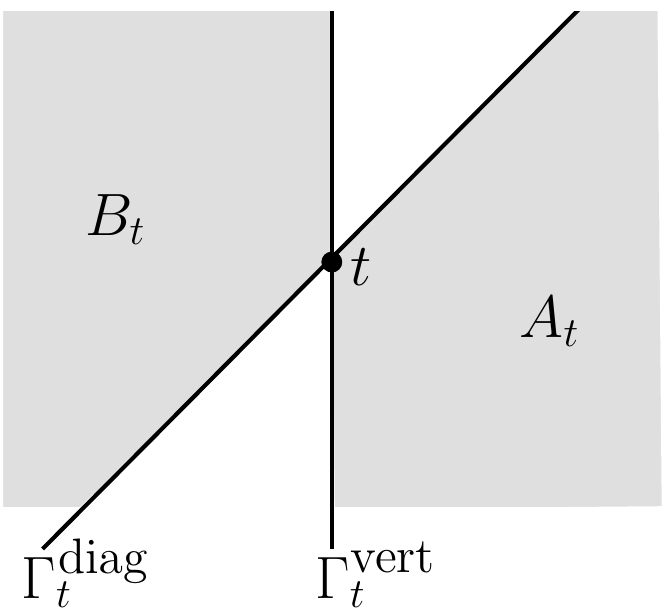}} \hspace{1,0cm} 
\subfigure[]{ \label{left-crossing-pair-fig}
\includegraphics[scale=0.68]{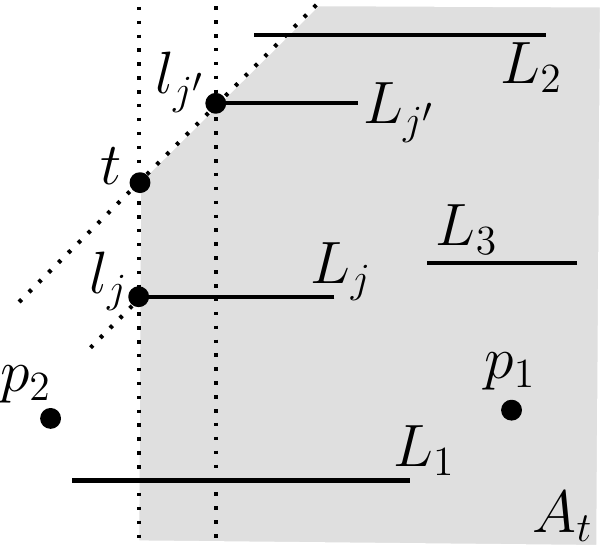}} \hspace{1,0cm} 
\subfigure[]{ \label{right-crossing-pair-fig}
\includegraphics[scale=0.68]{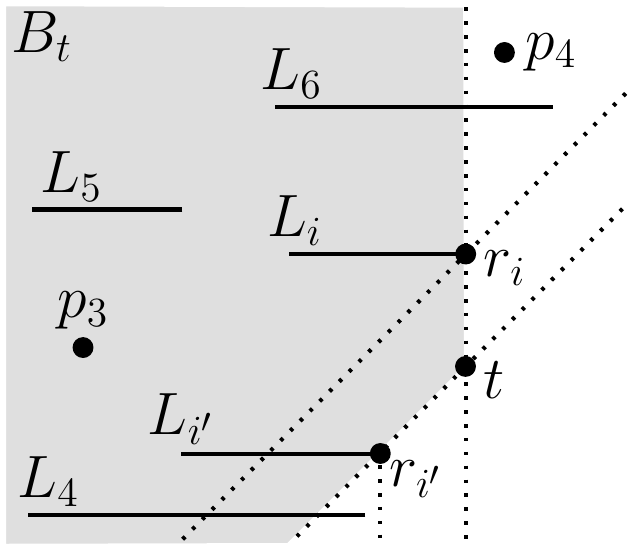}}
\caption{(a)~The regions $A_t, B_t$ and the lines $\Gamma_{t}^{\text{vert}},
\Gamma_{t}^{\text{diag}}$. (b) A left-crossing pair $(j,j^{\prime })$, where 
$L_{3}, p_{1} \in \mathcal{L}_{j,j^{\prime }}^{\text{right}}$ and $L_{1},
L_{2}, p_{2} \protect\notin \mathcal{L}_{j,j^{\prime }}^{\text{right}}$.
(c)~A~right-crossing pair $(i,i^{\prime })$, where $L_{5}, p_{3} \in 
\mathcal{L}_{i,i^{\prime }}^{\text{left}}$ and $L_{4}, L_{6}, p_{4} \protect%
\notin \mathcal{L}_{i,i^{\prime }}^{\text{left}}$.}
\label{crossing-pairs-fig}
\end{figure}

\begin{definition}
\label{start-end-pair-arbitrary-set-def}Let $S\subseteq \mathcal{P}\cup 
\mathcal{L}$ be an arbitrary set. Let $(i,i^{\prime })$ be a right-crossing
pair and $(j,j^{\prime })$ be a left-crossing pair. If $L_{i},L_{i^{\prime
}}\in S$ and $S\subseteq \mathcal{L}_{i,i^{\prime }}^{\text{left}}$, then $%
(i,i^{\prime })$ is the \emph{end-pair} of the set $S$. If $%
L_{j},L_{j^{\prime }}\in S$ and $S\subseteq \mathcal{L}_{j,j^{\prime }}^{%
\text{right}}$, then $(j,j^{\prime })$ is the \emph{start-pair} of the set $%
S $.
\end{definition}

\begin{definition}
\label{leftmost-segment-arbitrary-set-def}Let $S\subseteq \mathcal{P}\cup 
\mathcal{L}$ be an arbitrary set. The line segment $L_{q}\in S$ is the \emph{%
diagonally leftmost line segment} in $S$ if there exists a line segment $%
L_{j}\in \mathcal{L}\cap S$ such that $(j,q)$ is the start-pair of $S$.
\end{definition}

\begin{observation}
\label{unique-start-end-pair-L-obs}Every non-empty set $S\subseteq \mathcal{L%
}$ has a unique end-pair, a unique start-pair, and a unique diagonally
leftmost line segment.
\end{observation}

\subsection{The algorithm\label{bounded-alg-subsec}}

In this section we present our algorithm for \textsc{Bounded Dominating Set}
on tolerance graphs, cf.~Algorithm~\ref{bounded-dominating-tolerance-alg}.
Given a horizontal shadow representation $(\mathcal{P},\mathcal{L})$ of a
tolerance graph $G$, we first add two dummy line segments $L_{0}$ and $L_{|%
\mathcal{L}|+1}$ (with endpoints $l_{0},r_{0}$ and $l_{|\mathcal{L}|+1},r_{|%
\mathcal{L}|+1}$, respectively) such that all elements of $\mathcal{P\cup L}$
are contained in $A_{r_{0}}$ and in $B_{l_{|\mathcal{L}|+1}}$. Let $\mathcal{%
L}^{\prime }=\mathcal{L}\cup \{L_{0},L_{|\mathcal{L}|+1}\}$. Note that $(%
\mathcal{P},\mathcal{L}^{\prime })$ is a horizontal shadow representation of
some tolerance graph $G^{\prime }$, where the bounded vertices $%
V_{B}^{\prime }$ of $G^{\prime }$ correspond to the line segments of $%
\mathcal{L}^{\prime }$ and the unbounded vertices $V_{U}^{\prime }$ of $%
G^{\prime }$ correspond to the points of $\mathcal{P}$. Furthermore note
that $V_{B}^{\prime }=V_{B}\cup \{v_{0},v_{|\mathcal{L}|+1}\}$ and $%
V_{U}^{\prime }=V_{U}$, where $v_{0}$ and $v_{|\mathcal{L}|+1}$ are the
(isolated) bounded vertices of $G^{\prime }$ that correspond to the line
segments $L_{0}$ and $L_{|\mathcal{L}|+1}$, respectively. Finally observe
now that the set $V_{B}^{\prime }$ dominates the augmented graph $G^{\prime
} $ if and only if the set $V_{B}$ dominates the graph $G$; moreover, a set $%
S\subseteq V_{B}$ dominates $G$ if and only if $S\cup \{v_{0},v_{|\mathcal{L}%
|+1}\}$ dominates $G^{\prime }$.

For simplicity of the presentation, we refer in the following to the
augmented set $\mathcal{L}^{\prime }$ of horizontal line segments by $%
\mathcal{L}$. In the remainder of this section we will write $\mathcal{L}%
=\{L_{1},L_{2},\ldots ,L_{|\mathcal{L}|}\}$ with the understanding that the
first and the last line segments $L_{1}$ and $L_{|\mathcal{L}|}$ of $%
\mathcal{L}$ are dummy. Furthermore, we will refer to the augmented
tolerance graph $G^{\prime }$ by $G$.

For every pair of points $(a,b)\in \mathcal{A}\times \mathcal{B}$ such that $%
b\in \mathbb{R}_{\text{right}}^{2}(\Gamma _{a}^{\text{diag}})$, define $%
X(a,b)$ to be the set of all points of $\mathcal{P}$ and all line segments
of $\mathcal{L}$ that are contained in the region $B_{b}\setminus \Gamma
_{b}^{\text{vert}}$ and to the right of the line $\Gamma _{a}^{\text{diag}}$%
, cf.~Figure~\ref{X-a-b-fig}. That is, 
\begin{eqnarray}
R(a,b) &=&\left( B_{b}\setminus \Gamma _{b}^{\text{vert}}\right) \cap 
\mathbb{R}_{\text{right}}^{2}(\Gamma _{a}^{\text{diag}}),
\label{region-R(a,b)-def-eq} \\
X(a,b) &=&\{x\in \mathcal{P}\cup \mathcal{L}:x\subseteq R(a,b)\}.
\label{X(a,b)-def-eq}
\end{eqnarray}

\begin{figure}[h]
\centering\includegraphics[scale=0.68]{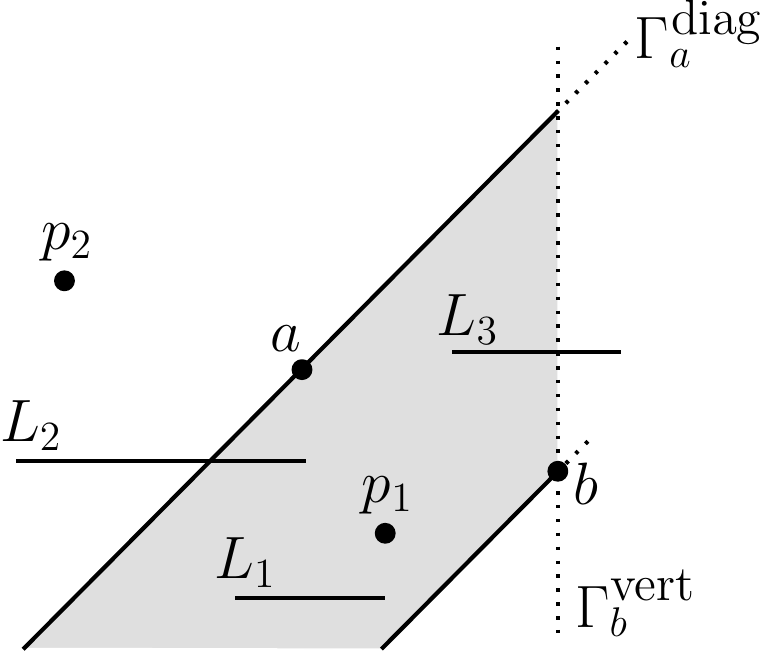}
\caption{The shaded region contains the points of $R(a,b)\subseteq \mathbb{R}%
^{2}$, where $(a,b)\in \mathcal{A}\times \mathcal{B}$. The set $X(a,b)$
contains all elements of $\mathcal{P}\cup \mathcal{L}$ that lie within $%
R(a,b)$. In this example, $L_{1},p_{1}\in X(a,b)$ and $L_{2},L_{3},p_{2}%
\protect\notin X(a,b)$.}
\label{X-a-b-fig}
\end{figure}

Now we present the main definition of this section, namely the quantity $%
BD_{(\mathcal{P},\mathcal{L})}(a,b,q,i,i^{\prime })$ for the \textsc{Bounded
Dominating Set} problem on tolerance graphs.

\begin{definition}
\label{BD-def}Let $(a,b)\in \mathcal{A}\times \mathcal{B}$ be a pair of
points such that $b\in \mathbb{R}_{\text{right}}^{2}(\Gamma _{a}^{\text{diag}%
})$. Let $(i,i^{\prime })$ be a right-crossing pair and $L_{q}$ be a line
segment such that $L_{q}\in \mathcal{L}_{i,i^{\prime }}^{\text{left}}$ and $%
L_{i},L_{i^{\prime }}\in \mathcal{\mathcal{L}}_{q}^{\text{right}}$.
Furthermore let $b\in \mathbb{R}_{\text{left}}^{2}(\Gamma _{r_{i}}^{\text{%
vert}})$. Then $BD_{(\mathcal{P},\mathcal{L})}(a,b,q,i,i^{\prime })$ is a
dominating set $Z\subseteq \mathcal{L}$ of $X(a,b)$ with the \emph{smallest}
size, such that:

\begin{itemize}
\item $(i,i^{\prime })$ is the end-pair of $Z$ and

\item $L_{q}$ is the diagonally leftmost line segment of $Z$.
\end{itemize}

If such a dominating set $Z\subseteq \mathcal{L}$ of $X(a,b)$ does not
exist, we define $BD_{(\mathcal{P},\mathcal{L})}(a,b,q,i,i^{\prime })=\bot $
and $|BD_{(\mathcal{P},\mathcal{L})}(a,b,q,i,i^{\prime })|=\infty $.
\end{definition}

Note that always $L_{q},L_{i},L_{i^{\prime }}\in BD_{(\mathcal{P},\mathcal{L}%
)}(a,b,q,i,i^{\prime })$. Furthermore some of the line segments $%
L_{q},L_{i},L_{i^{\prime }}$ may coincide, i.e.,~the set $%
\{L_{q},L_{i},L_{i^{\prime }}\}$ may have one, two, or three distinct
elements. However, since $b\in \mathbb{R}_{\text{left}}^{2}(\Gamma _{r_{i}}^{%
\text{vert}})$ in Definition~\ref{BD-def}, it follows that $L_{i}\nsubseteq
B_{b}\setminus \Gamma _{b}^{\text{vert}}$, and thus $L_{i}\notin X(a,b)$.
For simplicity of the presentation we may refer to the set $BD_{(\mathcal{P},%
\mathcal{L})}(a,b,q,i,i^{\prime })$ as $BD_{G}(a,b,q,i,i^{\prime })$, where $%
(\mathcal{P},\mathcal{L})$ is the horizontal shadow representation of the
tolerance graph $G$, or just as $BD(a,b,q,i,i^{\prime })$ whenever the
horizontal shadow representation $(\mathcal{P},\mathcal{L})$ is clear from
the context.

\begin{observation}
\label{obs:botcs}$BD(a,b,q,i,i^{\prime })\neq \bot $ if and only if $%
\mathcal{L\cap \mathcal{L}}_{q}^{\text{right}}\mathcal{\cap L}_{i,i^{\prime
}}^{\text{left}}$ is a dominating set of $X(a,b)$.
\end{observation}

\begin{observation}
\label{DB-init-feasible-obs}$BD(a,b,q,i,i^{\prime
})=\{L_{q},L_{i},L_{i^{\prime }}\}$ if and only if $\{L_{q},L_{i},L_{i^{%
\prime }}\}$ dominates $X(a,b)$.
\end{observation}

\begin{observation}
\label{R-a-b-obs}If $R(a,b)\subseteq S_{i}$ then $BD(a,b,q,i,i^{\prime
})=\{L_{q},L_{i},L_{i^{\prime }}\}$.
\end{observation}

Due to Observations~\ref{obs:botcs}-\ref{R-a-b-obs}, without loss of
generality we assume below (in Lemmas~\ref{bounded-dom-correctness-lem-0}-%
\ref{bounded-dom-correctness-lem-2}) that $BD(a,b,q,i,i^{\prime })\neq \bot $
and that $BD(a,b,q,i,i^{\prime })\neq \{L_{q},L_{i},L_{i^{\prime }}\}$, and
thus also $R(a,b)\nsubseteq S_{i}$ (cf.~Observation~\ref{R-a-b-obs}). We
provide our recursive computations for $BD(a,b,q,i,i^{\prime })$ in Lemmas~%
\ref{bounded-dom-correctness-lem-0},~\ref{bounded-dom-correctness-lem-1},
and~\ref{bounded-dom-correctness-lem-2}. In Lemma~\ref%
{bounded-dom-correctness-lem-0} we consider the case where $b\in S_{l_{i}}$
and in Lemmas~\ref{bounded-dom-correctness-lem-1} and~\ref%
{bounded-dom-correctness-lem-2} we consider the case where~$b\notin
S_{l_{i}} $.

\begin{lemma}
\label{bounded-dom-correctness-lem-0}Suppose that $BD(a,b,q,i,i^{\prime
})\neq \bot $ and that $BD(a,b,q,i,i^{\prime })\neq
\{L_{q},L_{i},L_{i^{\prime }}\}$, where $R(a,b)\nsubseteq S_{i}$. If $b\in
S_{l_{i}}$ then%
\begin{equation}
BD(a,b,q,i,i^{\prime })=BD(a,b^{\ast },q,i,i^{\prime }),
\label{recursion-eq-0}
\end{equation}%
where $b^{\ast }=\Gamma _{b}^{\text{vert}}\cap \Gamma _{l_{i}}^{\text{diag}}$%
.
\end{lemma}

\begin{proof}
Define the point $b^{\ast }=\Gamma _{b}^{\text{vert}}\cap \Gamma _{l_{i}}^{%
\text{diag}}$ of the plane. If $a\in S_{l_{i}}$ then $R(a,b)\subseteq S_{i}$%
, which is a contradiction. Thus $a\notin S_{l_{i}}$, and therefore $%
R(a,b^{\ast })\subseteq R(a,b)$. Consider now an element $x\in
X(a,b)\setminus X(a,b^{\ast })$. Then $x\cap S_{i}\neq \emptyset $, and thus 
$x$ is dominated by the line segment $L_{i}$. Therefore, for every set $Z$
of line segments such that $L_{i}\in Z$, we have that $Z$ dominates the set $%
X(a,b)$ if and only if $Z$ dominates the set $X(a,b^{\ast })$. Therefore $%
BD(a,b,q,i,i^{\prime })=BD(a,b^{\ast },q,i,i^{\prime })$.
\end{proof}

\medskip

Due to Lemma~\ref{bounded-dom-correctness-lem-0}, without loss of generality
we may assume in the following (in Lemmas~\ref%
{bounded-dom-correctness-lem-1-first-direction}-\ref%
{bounded-dom-correctness-lem-2}) that $b\notin S_{l_{i}}$. In order to
provide our second recursive computation for $BD(a,b,q,i,i^{\prime })$ in
Lemma~\ref{bounded-dom-correctness-lem-1} (cf.~Eq.~(\ref{recursion-eq-1})),
we first prove in the next lemma that the set at the right hand side of Eq.~(%
\ref{recursion-eq-1}) is indeed a dominating set of $X(a,b)$, in which $%
L_{q} $ is the diagonally leftmost line segment and $(i,i^{\prime })$ is the
end-pair.

\begin{lemma}
\label{bounded-dom-correctness-lem-1-first-direction}Suppose that $%
BD(a,b,q,i,i^{\prime })\neq \bot $ and that $BD(a,b,q,i,i^{\prime })\neq
\{L_{q},L_{i},L_{i^{\prime }}\}$, where $R(a,b)\nsubseteq S_{i}$ and $%
b\notin S_{l_{i}}$. Let $c\in \mathbb{R}^{2}$ and $L_{q^{\prime
}},L_{j},L_{j^{\prime }}\in \mathcal{L}$ such that:

\begin{enumerate}
\item $L_{q^{\prime }}\in \left( \mathcal{\mathcal{L}}_{q}^{\text{right}%
}\cap \mathcal{L}_{i,i^{\prime }}^{\text{left}}\right) \setminus \{L_{i}\}$,

\item $(j,j^{\prime })$ is a right-crossing pair of $\left( \mathcal{%
\mathcal{L}}_{q}^{\text{right}}\cap \mathcal{L}_{i,i^{\prime }}^{\text{left}%
}\right) \setminus \{L_{i}\}$, where $j^{\prime }=i^{\prime }$ whenever $%
i\neq i^{\prime }$,

\item $L_{q^{\prime }}\in \mathcal{L}_{j,j^{\prime }}^{\text{left}}$ and $%
L_{j},L_{j^{\prime }}\in \mathcal{\mathcal{L}}_{q^{\prime }}^{\text{right}}$,

\item $c=\Gamma _{r_{j}}^{\text{vert}}\cap \Gamma _{b}^{\text{diag}}$ if $%
r_{j}\in \mathbb{R}_{\text{left}}^{2}(\Gamma _{b}^{\text{vert}})$, and $c=b$
otherwise, and

\item the set $X(a,b)\setminus X(a,c)$ is dominated by $\{L_{j},L_{j^{%
\prime}}\}$.
\end{enumerate}

If $BD(a,c,q^{\prime },j,j^{\prime })\neq \bot $ then $\{L_{q},L_{i}\}\cup
BD(a,c,q^{\prime },j,j^{\prime })$ is a dominating set of $X(a,b)$, in which 
$L_{q}$ is the diagonally leftmost line segment and $(i,i^{\prime })$ is the
end-pair.
\end{lemma}

\begin{proof}
Assume that ${BD(a,c,q^{\prime },j,j^{\prime })\neq \bot}$. Since $%
X(a,b)\setminus X(a,c)$ is dominated by $\{L_{j},L_{j^{\prime}}\}$ by the
assumptions of the lemma, it follows that ${\{L_{q},L_{i}\}\cup
BD(a,c,q^{\prime },j,j^{\prime })}$ is a dominating set of~$X(a,b)$.

We now prove that $(i,i^{\prime })$ is the end-pair of $\{L_{q},L_{i}\}\cup
BD(a,c,q^{\prime },j,j^{\prime })$. First recall by the assumptions of the
lemma that $L_{j},L_{j^{\prime }}\in \mathcal{L\cap L}_{i,i^{\prime }}^{%
\text{left}}$ and note that $\mathcal{L}_{j,j^{\prime }}^{\text{left}%
}\subseteq \mathcal{L}_{i,i^{\prime }}^{\text{left}}$. Therefore, since $%
BD(a,c,q^{\prime },j,j^{\prime })\subseteq \mathcal{L\cap L}_{j,j^{\prime
}}^{\text{left}}$ by definition, it follows that $BD(a,c,q^{\prime
},j,j^{\prime })\subseteq \mathcal{L\cap L}_{i,i^{\prime }}^{\text{left}}$.
Let first $i^{\prime }=i$. Then clearly $L_{i}=L_{i^{\prime }}\in
\{L_{q},L_{i}\}\cup BD(a,c,q^{\prime },j,j^{\prime })\subseteq \mathcal{%
L\cap L}_{i,i}^{\text{left}}$, and thus in this case $(i,i^{\prime })=(i,i)$
is the end-pair of $\{L_{q},L_{i}\}\cup BD(a,c,q^{\prime },j,j^{\prime })$.
Let now $i^{\prime }\neq i$. Then $j^{\prime }=i^{\prime }$ by the
assumptions of the lemma, and thus $BD(a,c,q^{\prime },j,j^{\prime
})=BD(a,c,q^{\prime },j,i^{\prime })$. Then $L_{i},L_{i^{\prime }}\in
\{L_{q},L_{i}\}\cup BD(a,c,q^{\prime },j,j^{\prime })\subseteq \mathcal{%
L\cap L}_{i,i^{\prime }}^{\text{left}}$, and thus again $(i,i^{\prime })$ is
the end-pair of $\{L_{q},L_{i}\}\cup BD(a,c,q^{\prime },j,j^{\prime })$.

Finally, since $L_{q^{\prime }}\in \left( \mathcal{\mathcal{L}}_{q}^{\text{%
right}}\cap \mathcal{L}_{i,i^{\prime }}^{\text{left}}\right) \setminus
\{L_{i}\}$ by the assumptions of the lemma, it follows that $L_{q^{\prime
}}\subseteq \mathbb{R}_{\text{right}}^{2}(\Gamma _{l_{q}}^{\text{diag}})$,
cf.~Definition~\ref{left-right-crossing-pair-def}. Therefore, since $%
L_{q^{\prime }}$ is by definition the diagonally leftmost line segment of $%
BD(a,c,q^{\prime },j,j^{\prime })$, it follows that $L_{q}$ is the
diagonally leftmost line segment of $\{L_{q},L_{i}\}\cup $ $BD(a,c,q^{\prime
},j,j^{\prime })$. This completes the proof of the lemma.
\end{proof}

\medskip

Given the statement of Lemma~\ref%
{bounded-dom-correctness-lem-1-first-direction}, we are now ready to provide
our second recursive computation for $BD(a,b,q,i,i^{\prime })$ in the next
lemma.

\begin{lemma}
\label{bounded-dom-correctness-lem-1}Suppose that $BD(a,b,q,i,i^{\prime
})\neq \bot $ and that $BD(a,b,q,i,i^{\prime })\neq
\{L_{q},L_{i},L_{i^{\prime }}\}$, where $R(a,b)\nsubseteq S_{i}$ and $%
b\notin S_{l_{i}}$. If $BD(a,b,q,i,i^{\prime })\setminus L_{i}$ dominates
all elements of $\{x\in X(a,b):x\cap (S_{i}\cup F_{i})\neq \emptyset \}$
then 
\begin{equation}
BD(a,b,q,i,i^{\prime })=\{L_{q},L_{i}\}\cup \min_{c,q^{\prime },j,j^{\prime
}}\{BD(a,c,q^{\prime },j,j^{\prime })\},  \label{recursion-eq-1}
\end{equation}%
where the minimum is taken over all $c,q^{\prime },j,j^{\prime }$ that
satisfy the Conditions 1-5 of Lemma~\ref%
{bounded-dom-correctness-lem-1-first-direction}.
\end{lemma}

\begin{proof}
Let $Z\subseteq \mathcal{L\cap \mathcal{L}}_{q}^{\text{right}}\mathcal{\cap 
\mathcal{L}}_{i,i^{\prime }}^{\text{left}}$ be a dominating set of $X(a,b)$
such that $L_{q}$ is the diagonally leftmost line segment of $Z$ and $%
(i,i^{\prime })$ is the end-pair of $Z$. Suppose that $|Z|=|BD(a,b,q,i,i^{%
\prime })|$ and that all elements of $\{x\in X(a,b):x\cap (S_{i}\cup
F_{i})\neq \emptyset \}$ are dominated by $Z\setminus L_{i}$. Recall that $%
L_{i}\notin X(a,b)$. Thus, $Z\setminus \{L_{i}\}$ is a dominating set of $%
X(a,b)$. Let $(j,j^{\prime })$ denote the end-pair of $Z\setminus \{L_{i}\}$%
. Then all elements of $X(a,b)$ that are contained in $\mathbb{R}_{\text{%
right}}^{2}(\Gamma _{r_{j}}^{\text{vert}})$ must be dominated by $%
\{L_{j},L_{j^{\prime }}\}$. Define 
\begin{equation*}
c=%
\begin{cases}
\Gamma _{r_{j}}^{\text{vert}}\cap \Gamma _{b}^{\text{diag}}\text{,} & \text{%
if }r_{j}\in \mathbb{R}_{\text{left}}^{2}(\Gamma _{b}^{\text{vert}}) \\ 
b\text{,} & \text{otherwise}%
\end{cases}%
.
\end{equation*}%
That is, the set $X(a,b)\setminus X(a,c)$ is dominated by $%
\{L_{j},L_{j^{\prime }}\}$. Let $L_{q^{\prime }}$ denote the diagonally
leftmost line segment of $Z\setminus \{L_{i}\}$. Note that, if $L_{q}\neq
L_{i}$ then $L_{q^{\prime }}=L_{q}$. Furthermore note that $L_{q^{\prime
}}\in \mathcal{L}_{j,j^{\prime }}^{\text{left}}$ and $L_{j},L_{j^{\prime
}}\in \mathcal{\mathcal{L}}_{q^{\prime }}^{\text{right}}$. Since $Z\subseteq 
\mathcal{L\cap \mathcal{L}}_{q}^{\text{right}}\mathcal{\cap \mathcal{L}}%
_{i,i^{\prime }}^{\text{left}}$, it follows that $(j,j^{\prime })$ is a
right-crossing pair of $\left( \mathcal{\mathcal{L}}_{q}^{\text{right}}\cap 
\mathcal{L}_{i,i^{\prime }}^{\text{left}}\right) \setminus \{L_{i}\}$ and
that $L_{q^{\prime }}\in \left( \mathcal{\mathcal{L}}_{q}^{\text{right}}\cap 
\mathcal{L}_{i,i^{\prime }}^{\text{left}}\right) \setminus \{L_{i}\}$.
Furthermore, if $i\neq i^{\prime }$ then $L_{i^{\prime }}\in Z\setminus
\{L_{i}\}$, and thus, by the choice of the right-crossing pair $(j,j^{\prime
})$ as the end-pair of $Z\setminus \{L_{i}\}$, it follows that $j^{\prime
}=i^{\prime }$.

Since $L_{j},L_{j^{\prime }}\in \mathcal{L}_{i,i^{\prime }}^{\text{left}%
}\setminus \{L_{i}\}$, note that $L_{i}\notin BD(a,b,q^{\prime },j,j^{\prime
})$. Moreover note that $X(a,c)\subseteq X(a,b)$, and thus $Z\setminus
\{L_{i}\}$ is also a dominating set of $X(a,c)$. Therefore, since $%
(j,j^{\prime })$ is the end-pair of $Z\setminus \{L_{i}\}$, it follows that 
\begin{equation*}
|\{L_{q}\}\cup BD(a,c,q^{\prime },j,j^{\prime })|=|BD(a,c,q^{\prime
},j,j^{\prime })|\leq |Z\setminus \{L_{i}\}|\text{, if }L_{q}\neq L_{i}
\end{equation*}%
and that%
\begin{equation*}
|BD(a,c,q^{\prime },j,j^{\prime })|\leq |Z\setminus \{L_{i}\}|\text{, if }%
L_{q}=L_{i}.
\end{equation*}%
That is, in both cases where $L_{q}\neq L_{i}$ or $L_{q}=L_{i}$, we have that%
\begin{eqnarray}
|\{L_{q},L_{i}\}\cup BD(a,b,q^{\prime },j,j^{\prime })| &=&1+\left\vert
\left( \{L_{q}\}\cup BD(a,c,q^{\prime },j,j^{\prime })\right) \setminus
\{L_{i}\}\right\vert  \notag \\
&=&1+\left\vert BD(a,c,q^{\prime },j,j^{\prime })\right\vert
\label{bounded-inequality-eq-1} \\
&\leq &1+\left\vert Z\setminus \{L_{i}\}\right\vert
=|Z|=|BD(a,b,q,i,i^{\prime })|.  \notag
\end{eqnarray}

Finally Lemma~\ref{bounded-dom-correctness-lem-1-first-direction} implies
that, if $BD(a,c,q^{\prime },j,j^{\prime })\neq \bot $, then $%
\{L_{q},L_{i}\}\cup $ $BD(a,c,q^{\prime },j,j^{\prime })$ is a dominating
set of $X(a,b)$, in which $L_{q}$ is the diagonally leftmost line segment
and $(i,i^{\prime })$ is the end-pair. Therefore $|BD(a,b,q,i,i^{\prime
})|\leq |\{L_{q},L_{i}\}\cup BD(a,b,q^{\prime},j,j^{\prime })|$, and thus it
follows by Eq.~(\ref{bounded-inequality-eq-1}) that $|BD(a,b,q,i,i^{\prime
})| = |\{L_{q},L_{i}\}\cup BD(a,b,q^{\prime},j,j^{\prime })|$
\end{proof}

\medskip

In order to provide our third recursive computation for $BD(a,b,q,i,i^{%
\prime })$ in Lemma~\ref{bounded-dom-correctness-lem-2} (cf.~Eq.~(\ref%
{recursion-eq-2})), we first prove in Lemmas~\ref%
{bounded-dom-correctness-lem-2-first-direction-lem-1} and~\ref%
{bounded-dom-correctness-lem-2-first-direction-lem-2} that the set at the
right hand side of Eq.~(\ref{recursion-eq-2}) is indeed a dominating set of $%
X(a,b)$, in which $L_{q}$ is the diagonally leftmost line segment and $%
(i,i^{\prime })$ is the end-pair.

\begin{lemma}
\label{bounded-dom-correctness-lem-2-first-direction-lem-1}Suppose that $%
BD(a,b,q,i,i^{\prime })\neq \bot $ and that $BD(a,b,q,i,i^{\prime })\neq
\{L_{q},L_{i},L_{i^{\prime }}\}$, where $R(a,b)\nsubseteq S_{i}$ and $%
b\notin S_{l_{i}}$. Let $c\in \mathbb{R}^{2}$ such that:

\begin{enumerate}
\item $c\in \mathcal{B}\cap R(a,b)$ and $c\in \mathbb{R}_{\text{right}%
}^{2}(\Gamma _{l_{i}}^{\text{vert}})\setminus F_{l_{i}}$,

\item $\mathcal{P}\cap X(a,b)\cap F_{c}\cap F_{i}=\emptyset $.
\end{enumerate}

If $BD(a,c,q,i,i^{\prime })\neq \bot $ and $BD(c,b,q,i,i^{\prime })\neq \bot 
$, then $BD(a,c,q,i,i^{\prime })\cup BD(c,b,q,i,i^{\prime })$ is a
dominating set of $X(a,b)$, in which $L_{q}$ is the diagonally leftmost line
segment and $(i,i^{\prime })$ is the end-pair.
\end{lemma}

\begin{proof}
Assume that $BD(a,c,q,i,i^{\prime })\neq \bot $ and $BD(c,b,q,i,i^{\prime
})\neq \bot $. First note that, since $c\in R(a,b)$ by assumption, it
follows that $X(a,c)\cup X(c,b)\subseteq X(a,b)$, cf.~Eq.~(\ref%
{X(a,b)-def-eq}). Furthermore, since $c\in R(a,b)\subseteq B_{b}$ and $c\in 
\mathbb{R}_{\text{right}}^{2}(\Gamma _{l_{i}}^{\text{vert}})\setminus
F_{l_{i}}$ by the assumption, it follows that also $b\in \mathbb{R}_{\text{%
right}}^{2}(\Gamma _{l_{i}}^{\text{vert}})\setminus F_{l_{i}}$. Now recall
that $b\in \mathbb{R}_{\text{left}}^{2}(\Gamma _{r_{i}}^{\text{vert}})$ by
Definition~\ref{BD-def}, and thus also $c\in \mathbb{R}_{\text{left}%
}^{2}(\Gamma _{r_{i}}^{\text{vert}})$. Therefore, since $c\in \mathbb{R}_{%
\text{right}}^{2}(\Gamma _{l_{i}}^{\text{vert}})\setminus F_{l_{i}}$ by the
assumption, it follows that $S_{c}\cap \Gamma _{c}^{\text{diag}}\subseteq
S_{i}\cup F_{i}$. Moreover, since $c\in \mathbb{R}_{\text{right}}^{2}(\Gamma
_{l_{i}}^{\text{vert}})$ and $b\in \mathbb{R}_{\text{left}}^{2}(\Gamma
_{r_{i}}^{\text{vert}})$, it follows that $F_{c}\cap R(a,b)\subseteq
S_{i}\cup F_{i}$.

The line segments of $\mathcal{L}\cap X(a,b)$ can be partitioned into the
following sets:%
\begin{eqnarray*}
\mathcal{L}_{1} &=&\mathcal{L}\cap X(a,c) \\
\mathcal{L}_{2} &=&\mathcal{L}\cap X(c,b) \\
\mathcal{L}_{3} &=&\{L_{k}\in \mathcal{L}\cap X(a,b):L_{k}\cap F_{c}\neq
\emptyset \} \\
\mathcal{L}_{4} &=&\{L_{k}\in \mathcal{L}\cap X(a,b):L_{k}\cap S_{c}\cap
\Gamma _{c}^{\text{diag}}\neq \emptyset \}
\end{eqnarray*}

Since $BD(a,c,q,i,i^{\prime })\neq \bot $ and $BD(c,b,q,i,i^{\prime })\neq
\bot $ by assumption, it follows that the line segments of $\mathcal{L}_{1}$
are all dominated by $BD(a,c,q,i,i^{\prime })$ and the line segments of $%
\mathcal{L}_{2}$ are all dominated by $BD(c,b,q,i,i^{\prime })$.
Furthermore, since $F_{c}\cap R(a,b)\subseteq S_{i}\cup F_{i}$ as we proved
above, it follows that all line segments of $\mathcal{L}_{3}$ are dominated
by the line segment $L_{i}$. Moreover, since $S_{c}\cap \Gamma _{c}^{\text{%
diag}}\subseteq S_{i}\cup F_{i}$ as we proved above, it follows that all
line segments of $\mathcal{L}_{4}$ are dominated also by the line segment $%
L_{i}$. That is, all line segments of $\mathcal{L}\cap X(a,b)=\mathcal{L}%
_{1}\cup \mathcal{L}_{2}\cup \mathcal{L}_{3}\cup \mathcal{L}_{4}$ are
dominated by $BD(a,c,q,i,i^{\prime })\cup BD(c,b,q,i,i^{\prime })$.

Since $\mathcal{P}\cap X(a,b)\cap F_{c}\cap F_{i}=\emptyset $ by the
assumption, the points of $\mathcal{P}\cap X(a,b)$ can be partitioned into
the following sets:%
\begin{eqnarray*}
\mathcal{P}_{1} &=&\mathcal{P}\cap X(a,c) \\
\mathcal{P}_{2} &=&\mathcal{P}\cap X(c,b) \\
\mathcal{P}_{3} &=&\mathcal{P}\cap X(a,b)\cap F_{c}\cap S_{i}
\end{eqnarray*}%
It is easy to see that the points of $\mathcal{P}_{1}$ are all dominated by $%
BD(a,c,q,i,i^{\prime })$ and that the points of $\mathcal{P}_{2}$ are all
dominated by $BD(c,b,q,i,i^{\prime })$. Furthermore the points of $\mathcal{P%
}_{3}$ are dominated by the line segment $L_{i}$. Thus all points of $%
\mathcal{P}\cap X(a,b)=\mathcal{P}_{1}\cup \mathcal{P}_{2}\cup \mathcal{P}%
_{3}$ are dominated by $BD(a,c,q,i,i^{\prime })\cup BD(c,b,q,i,i^{\prime })$%
. Summarizing, $BD(a,c,q,i,i^{\prime })\cup BD(c,b,q,i,i^{\prime })$ is a
dominating set of $X(a,b)$.

Furthermore, since $(i,i^{\prime })$ is the end-pair of both $%
BD(a,c,q,i,i^{\prime })$ and $BD(c,b,q,i,i^{\prime })$, it follows that $%
(i,i^{\prime })$ is also the end-pair of $BD(a,c,q,i,i^{\prime })\cup
BD(c,b,q,i,i^{\prime })$. Similarly, since $L_{q}$ is the diagonally
leftmost line segment of both $BD(a,c,q,i,i^{\prime })$ and $%
BD(c,b,q,i,i^{\prime })$, it follows that $L_{q}$ is also the diagonally
leftmost line segment of $BD(a,c,q,i,i^{\prime })\cup BD(c,b,q,i,i^{\prime
}) $. This completes the proof of the lemma.
\end{proof}

\begin{lemma}
\label{bounded-dom-correctness-lem-2-first-direction-lem-2}Suppose that $%
BD(a,b,q,i,i^{\prime })\neq \bot $ and that $BD(a,b,q,i,i^{\prime })\neq
\{L_{q},L_{i},L_{i^{\prime }}\}$, where $R(a,b)\nsubseteq S_{i}$ and $%
b\notin S_{l_{i}}$. Let $c^{\prime }\in \mathbb{R}^{2}$ and $L_{q^{\prime
}}\in \mathcal{L}$ such that:

\begin{enumerate}
\item $c^{\prime }\in \mathcal{B}\cap R(a,b)$ and $c^{\prime }\in F_{l_{i}}$,

\item $L_{i},L_{i^{\prime }}\in \mathcal{\mathcal{L}}_{q^{\prime }}^{\text{%
right}}$,

\item $L_{q^{\prime }}\in \mathcal{\mathcal{L}}_{q}^{\text{right}}\cap 
\mathcal{L}_{i,i^{\prime }}^{\text{left}}$ and $l_{q^{\prime }}\in F_{l_{i}}$%
,

\item $c^{\prime }\in \Gamma _{l_{q^{\prime }}}^{\text{diag}}$ or $c^{\prime
}\in \Gamma _{b}^{\text{diag}}$, and

\item $\mathcal{P}\cap X(a,b)\cap F_{c^{\prime }}=\emptyset $.
\end{enumerate}

If $BD(a,c^{\prime },q,i,i^{\prime })\neq \bot $ and $BD(c^{\prime
},b,q^{\prime },i,i^{\prime })\neq \bot $ then $BD(a,c^{\prime
},q,i,i^{\prime })\cup BD(c^{\prime },b,q^{\prime },i,i^{\prime })$ is a
dominating set of $X(a,b)$, in which $L_{q}$ is the diagonally leftmost line
segment and $(i,i^{\prime })$ is the end-pair.
\end{lemma}

\begin{proof}
Assume that $BD(a,c^{\prime },q,i,i^{\prime })\neq \bot $ and $BD(c^{\prime
},b,q^{\prime },i,i^{\prime })\neq \bot $. First note that, since $c^{\prime
}\in R(a,b)$ by assumption, it follows that $X(a,c^{\prime })\cup
X(c^{\prime },b)\subseteq X(a,b)$, cf.~Eq.~(\ref{X(a,b)-def-eq}). Since $%
c^{\prime }\in F_{l_{i}}$ by assumption, it follows that $F_{c^{\prime
}}\subseteq F_{l_{i}}\subseteq S_{i}\cup F_{i}$. Moreover, if $c^{\prime
}\in \Gamma _{l_{q^{\prime }}}^{\text{diag}}$ then $S_{c^{\prime }}\cap
\Gamma _{c^{\prime }}^{\text{diag}}\subseteq \Gamma _{l_{q^{\prime }}}^{%
\text{diag}}$, and thus $S_{c^{\prime }}\cap \Gamma _{c^{\prime }}^{\text{%
diag}}\subseteq S_{q^{\prime }}\cup F_{q^{\prime }}$.

Similarly to the proof of Lemma~\ref%
{bounded-dom-correctness-lem-2-first-direction-lem-1}, the line segments of $%
\mathcal{L}\cap X(a,b)$ can be partitioned into the following sets:%
\begin{eqnarray*}
\mathcal{L}_{1} &=&\mathcal{L}\cap X(a,c^{\prime }), \\
\mathcal{L}_{2} &=&\mathcal{L}\cap X(c^{\prime },b), \\
\mathcal{L}_{3} &=&\{L_{k}\in \mathcal{L}\cap X(a,b):L_{k}\cap F_{c^{\prime
}}\neq \emptyset \}, \\
\mathcal{L}_{4} &=&\{L_{k}\in \mathcal{L}\cap X(a,b):L_{k}\cap S_{c^{\prime
}}\cap \Gamma _{c^{\prime }}^{\text{diag}}\neq \emptyset \}.
\end{eqnarray*}

Since $BD(a,c^{\prime },q,i,i^{\prime })\neq \bot $ and $BD(c^{\prime
},b,q^{\prime },i,i^{\prime })\neq \bot $ by assumption, it follows that the
line segments of $\mathcal{L}_{1}$ are all dominated by $BD(a,c^{\prime
},q,i,i^{\prime })$ and that the line segments of $\mathcal{L}_{2}$ are all
dominated by $BD(c^{\prime },b,q^{\prime },i,i^{\prime })$. Furthermore,
since $F_{c^{\prime }}\subseteq S_{i}\cup F_{i}$ as we proved above, it
follows that all line segments of $\mathcal{L}_{3}$ are dominated by the
line segment $L_{i}$. If $c^{\prime }\in \Gamma _{b}^{\text{diag}}$ then $%
\mathcal{L}_{4}=\emptyset $. Suppose that $c^{\prime }\in \Gamma
_{l_{q^{\prime }}}^{\text{diag}}$. Then, since $S_{c^{\prime }}\cap \Gamma
_{c^{\prime }}^{\text{diag}}\subseteq S_{q^{\prime }}\cup F_{q^{\prime }}$
as we proved above, it follows that all line segments of $\mathcal{L}_{4}$
are dominated by the line segment $L_{q^{\prime }}$. That is, in both cases
where $c^{\prime }\in \Gamma _{b}^{\text{diag}}$ or $c^{\prime }\in \Gamma
_{l_{q^{\prime }}}^{\text{diag}}$, all line segments of $\mathcal{L}\cap
X(a,b)=\mathcal{L}_{1}\cup \mathcal{L}_{2}\cup \mathcal{L}_{3}\cup \mathcal{L%
}_{4}$ are dominated by $BD(a,c^{\prime },q,i,i^{\prime })\cup BD(c^{\prime
},b,q^{\prime },i,i^{\prime })$.

Since $c^{\prime }\in F_{l_{i}}$ and $\mathcal{P}\cap X(a,b)\cap
F_{c^{\prime }}=\emptyset $ by the assumption, it follows that the points of 
$\mathcal{P}\cap X(a,b)$ can be partitioned into the following sets:%
\begin{eqnarray*}
\mathcal{P}_{1} &=&\mathcal{P}\cap X(a,c^{\prime }), \\
\mathcal{P}_{2} &=&\mathcal{P}\cap X(c^{\prime },b).
\end{eqnarray*}%
It is easy to see that the points of $\mathcal{P}_{1}$ are all dominated by $%
BD(a,c^{\prime },q,i,i^{\prime })$ and that the points of $\mathcal{P}_{2}$
are all dominated by $BD(c^{\prime },b,q^{\prime },i,i^{\prime })$.
Summarizing, $BD(a,c^{\prime },q,i,i^{\prime })\cup BD(c^{\prime
},b,q^{\prime },i,i^{\prime })$ is a dominating set of $X(a,b)$.

Since $(i,i^{\prime })$ is the end-pair of both $BD(a,c^{\prime
},q,i,i^{\prime })$ and $BD(c^{\prime },b,q^{\prime },i,i^{\prime })$, it
follows that $(i,i^{\prime })$ is also the end-pair of $BD(a,c^{\prime
},q,i,i^{\prime })\cup BD(c^{\prime },b,q^{\prime },i,i^{\prime })$. Now
note that $L_{q}$ is the diagonally leftmost line segment of $BD(a,c^{\prime
},q,i,i^{\prime })$ and $L_{q^{\prime }}$ is the diagonally leftmost line
segment of $BD(c^{\prime },b,q^{\prime },i,i^{\prime })$. Therefore, since $%
L_{q^{\prime }}\in \mathcal{\mathcal{L}}_{q}^{\text{right}}\cap \mathcal{L}%
_{i,i^{\prime }}^{\text{left}}$ by assumption, it follows that $L_{q}$
remains the diagonally leftmost line segment of $BD(a,c^{\prime
},q,i,i^{\prime })\cup BD(c^{\prime },b,q^{\prime },i,i^{\prime })$. This
completes the proof of the lemma.
\end{proof}

\medskip

Given the statements of Lemmas~\ref%
{bounded-dom-correctness-lem-2-first-direction-lem-1} and~\ref%
{bounded-dom-correctness-lem-2-first-direction-lem-2}, we are now ready to
provide our third recursive computation for $BD(a,b,q,i,i^{\prime })$ in the
next lemma.

\begin{lemma}
\label{bounded-dom-correctness-lem-2}Suppose that $BD(a,b,q,i,i^{\prime
})\neq \bot $ and that $BD(a,b,q,i,i^{\prime })\neq
\{L_{q},L_{i},L_{i^{\prime }}\}$, where $R(a,b)\nsubseteq S_{i}$ and $%
b\notin S_{l_{i}}$. If $BD(a,b,q,i,i^{\prime })\setminus L_{i}$ does not
dominate all elements of $\{x\in X(a,b):x\cap (S_{i}\cup F_{i})\neq
\emptyset \}$ then%
\begin{equation}
BD(a,b,q,i,i^{\prime })=\min_{c,c^{\prime },q^{\prime }}%
\begin{cases}
BD(a,c,q,i,i^{\prime })\cup BD(c,b,q,i,i^{\prime }) \\ 
BD(a,c^{\prime },q,i,i^{\prime })\cup BD(c^{\prime },b,q^{\prime
},i,i^{\prime })%
\end{cases}%
,  \label{recursion-eq-2}
\end{equation}%
where the minimum is taken over all $c,c^{\prime },q^{\prime }$ that satisfy
the Conditions of Lemmas~\ref%
{bounded-dom-correctness-lem-2-first-direction-lem-1} and~\ref%
{bounded-dom-correctness-lem-2-first-direction-lem-2}, i.e.,

\begin{enumerate}
\item $c,c^{\prime }\in \mathcal{B}\cap R(a,b)$,

\item $c\in \mathbb{R}_{\text{right}}^{2}(\Gamma _{l_{i}}^{\text{vert}%
})\setminus F_{l_{i}}$ and $c^{\prime }\in F_{l_{i}}$,

\item $L_{i},L_{i^{\prime }}\in \mathcal{\mathcal{L}}_{q^{\prime }}^{\text{%
right}}$,

\item $L_{q^{\prime }}\in \mathcal{\mathcal{L}}_{q}^{\text{right}}\cap 
\mathcal{L}_{i,i^{\prime }}^{\text{left}}$ and $l_{q^{\prime }}\in F_{l_{i}}$%
,

\item $c^{\prime }\in \Gamma _{l_{q^{\prime }}}^{\text{diag}}$ or $c^{\prime
}\in \Gamma _{b}^{\text{diag}}$, and

\item $\mathcal{P}\cap X(a,b)\cap F_{c}\cap F_{i}=\emptyset $ and $\mathcal{P%
}\cap X(a,b)\cap F_{c^{\prime }}=\emptyset $.
\end{enumerate}
\end{lemma}

\begin{proof}
Assume that $BD(a,b,q,i,i^{\prime })\setminus L_{i}$ does not dominate all
elements of $\{x\in X(a,b):x\cap (S_{i}\cup F_{i})\neq \emptyset \}$. Recall
that $b\in \mathbb{R}_{\text{left}}^{2}(\Gamma _{r_{i}}^{\text{vert}})$ by
Definition~\ref{BD-def}. First we prove that also $b\in \mathbb{R}_{\text{%
right}}^{2}(\Gamma _{l_{i}}^{\text{vert}})$. Assume otherwise that $b\notin 
\mathbb{R}_{\text{right}}^{2}(\Gamma _{l_{i}}^{\text{vert}})$. Then, since $%
b\notin S_{l_{i}}$ by the assumption of the lemma, it follows that $b\in
B_{l_{i}}$. Thus $\left( S_{i}\cup F_{i}\right) \cap B_{b}=\emptyset $,
i.e.,~$L_{i}$ does not dominate any element of $X(a,b)$, cf.~Eq.~(\ref%
{X(a,b)-def-eq}). Therefore, since $BD(a,b,q,i,i^{\prime })\setminus L_{i}$
does not dominate all elements of $\{x\in X(a,b):x\cap (S_{i}\cup F_{i})\neq
\emptyset \}$ by assumption, it follows that $BD(a,b,q,i,i^{\prime })$ does
also not dominate all elements of $X(a,b)$, which is a contradiction to the
assumption that $BD(a,b,q,i,i^{\prime })\neq \bot $. Therefore $b\in \mathbb{%
R}_{\text{right}}^{2}(\Gamma _{l_{i}}^{\text{vert}})$.

Let $x_{0}\in X(a,b)$ be such that $x_{0}\cap (S_{i}\cup F_{i})\neq
\emptyset $ and $x_{0}$ is not dominated by $BD(a,b,q,i,i^{\prime
})\setminus L_{i}$. Let also $Z\subseteq \mathcal{L\cap \mathcal{L}}_{q}^{%
\text{right}}\mathcal{\cap \mathcal{L}}_{i,i^{\prime }}^{\text{left}}$ be an
arbitrary dominating set of $X(a,b)$ such that $L_{q}$ is the diagonally
leftmost line segment of $Z$ and $(i,i^{\prime })$ is the end-pair of $Z$.
Suppose that $|Z|=|BD(a,b,q,i,i^{\prime })|$ and that $x_{0}$ is dominated
by $L_{i}$ but not by $Z\setminus L_{i}$. Note that such a dominating set $Z$
always exists due to our assumption on $BD(a,b,q,i,i^{\prime })$. We
distinguish now two cases.

\medskip

\noindent \textbf{Case 1.} $x_{0}\cap \mathbb{R}_{\text{right}}^{2}(\Gamma
_{l_{i}}^{\text{diag}})\neq \emptyset $. Let $t\in 
%TCIMACRO{\U{211d} }%
%BeginExpansion
\mathbb{R}
%EndExpansion
^{2}$ be an arbitrary point of $x_{0}\cap \mathbb{R}_{\text{right}%
}^{2}(\Gamma _{l_{i}}^{\text{diag}})$. Since $x_{0}\in X(a,b)$ and $b\in 
\mathbb{R}_{\text{left}}^{2}(\Gamma _{r_{i}}^{\text{vert}})$ by Definition~%
\ref{BD-def}, it follows that $t\in S_{i}\cup F_{i}$. If $t\in S_{i}$ then
let $t^{\ast }\in R(a,b)$ be an arbitrary point on the intersection of the
line segment $L_{i}$ with the reverse shadow $F_{t}$ of the point $t$, i.e., 
$t^{\ast }\in R(a,b)\cap L_{i}\cap F_{t}$. Note that $t^{\ast }$ always
exists, since $x_{0}\in X(a,b)$, $R(a,b)\nsubseteq S_{i}$ by the assumption
of the lemma, and $b\in \mathbb{R}_{\text{right}}^{2}(\Gamma _{l_{i}}^{\text{%
vert}})$ as we proved above. Otherwise, if $t\in F_{i}$, then we define $%
t^{\ast }=t$. Since $t\in \mathbb{R}_{\text{right}}^{2}(\Gamma _{l_{i}}^{%
\text{diag}})$ by assumption, note that in both cases where $t\in S_{i}$ and 
$t\in F_{i}$, we have that $t\in S_{t^{\ast }}$ and that either $t^{\ast
}\in L_{i}$ or $t^{\ast }\in F_{i}\setminus L_{i}$.

Suppose that there exists a line segment $L_{k}\in Z\setminus L_{i}$ such
that $t^{\ast }\in S_{k}$. Then, since $t\in S_{t^{\ast }}$, it follows that
also $t\in S_{k}$. Thus the element $x_{0}\in X(a,b)$ is dominated by $%
L_{k}\in Z\setminus L_{i}$, which is a contradiction. Therefore $t^{\ast
}\notin S_{k}$ for every line segment $L_{k}\in Z\setminus L_{i}$.

Let $j$ be the greatest index such that for the line segment $L_{j}\in
Z\setminus L_{i}$ we have $r_{j}\in \mathbb{R}_{\text{left}}^{2}(\Gamma
_{t^{\ast }}^{\text{vert}})$. That is, for every other line segment $%
L_{s}\in Z\setminus L_{i}$ with $r_{s}\in \mathbb{R}_{\text{left}%
}^{2}(\Gamma _{t^{\ast }}^{\text{vert}})$, we have $r_{s}\in \mathbb{R}_{%
\text{left}}^{2}(\Gamma _{l_{j}}^{\text{vert}})$. If $r_{j}\in \mathbb{R}_{%
\text{right}}^{2}(\Gamma _{l_{i}}^{\text{vert}})$ then we define $%
t_{1}=r_{j} $. If $r_{j}\notin \mathbb{R}_{\text{right}}^{2}(\Gamma
_{l_{i}}^{\text{vert}})$ then we define $t_{1}=l_{i}$. Furthermore, if such
a line segment $L_{j}$ does not exist in $Z\setminus L_{i}$ (i.e.,~if $%
r_{s}\notin \mathbb{R}_{\text{left}}^{2}(\Gamma _{t^{\ast }}^{\text{vert}})$
for every $L_{s}\in Z\setminus L_{i}$), then we define again $t_{1}=l_{i}$.

Let $L_{j^{\prime }}\in Z\setminus L_{i}$ be a line segment such that $%
l_{j^{\prime }}\in \mathbb{R}_{\text{right}}^{2}(\Gamma _{t^{\ast }}^{\text{%
diag}})$ and that, for every other line segment $L_{s}\in Z\setminus L_{i}$
with $l_{s}\in \mathbb{R}_{\text{right}}^{2}(\Gamma _{t^{\ast }}^{\text{diag}%
})$, we have $l_{s}\in \mathbb{R}_{\text{right}}^{2}(\Gamma _{l_{j^{\prime
}}}^{\text{diag}})$. If $l_{j^{\prime }}\in \mathbb{R}_{\text{left}%
}^{2}(\Gamma _{b}^{\text{diag}})$ then we define $t_{2}=l_{j^{\prime }}$. If 
$l_{j^{\prime }}\notin \mathbb{R}_{\text{left}}^{2}(\Gamma _{b}^{\text{diag}%
})$ then we define $t_{2}=b$. Furthermore, if such a line segment $%
L_{j^{\prime }}$ does not exist in $Z\setminus L_{i}$ (i.e.,~if $l_{s}\notin 
\mathbb{R}_{\text{right}}^{2}(\Gamma _{t^{\ast }}^{\text{diag}})$ for every $%
L_{s}\in Z\setminus L_{i}$), then we define again $t_{2}=b$.

Now we define 
\begin{equation*}
c=\Gamma _{t_{1}}^{\text{vert}}\cap \Gamma _{t_{2}}^{\text{diag}}.
\end{equation*}%
It is easy to check by the above definition of $t_{1}$ and $t_{2}$ that $%
c\in \mathcal{B}\cap R(a,b)$ and that $c\in \mathbb{R}_{\text{right}%
}^{2}(\Gamma _{l_{i}}^{\text{vert}})\setminus F_{l_{i}}$.

Assume that there exists at least one point $p_{k}\in \mathcal{P}\cap
X(a,b)\cap F_{c}\cap F_{i}$. Then, since $BD(a,b,q,i,i^{\prime })\neq \bot $
by assumption, there must be a line segment $L_{k^{\prime }}\in Z\setminus
L_{i}$ such that $L_{k^{\prime }}$ dominates $p_{k}$. Since $p_{k}\in F_{c}$
by assumption, it follows that $L_{k^{\prime }}\cap F_{c}\neq \emptyset $.
If $r_{k^{\prime }}\in \mathbb{R}_{\text{left}}^{2}(\Gamma _{t^{\ast }}^{%
\text{vert}})$ then $r_{k^{\prime }}\in \mathbb{R}_{\text{left}}^{2}(\Gamma
_{c}^{\text{vert}})$ by the above definition of $c$, and thus the line
segment $L_{k^{\prime }}$ does not dominate the point $p_{k}$, which is a
contradiction. Therefore $r_{k^{\prime }}\notin \mathbb{R}_{\text{left}%
}^{2}(\Gamma _{t^{\ast }}^{\text{vert}})$. If $l_{k^{\prime }}\in \mathbb{R}%
_{\text{right}}^{2}(\Gamma _{t^{\ast }}^{\text{diag}})$ then $l_{k^{\prime
}}\in \mathbb{R}_{\text{right}}^{2}(\Gamma _{c}^{\text{diag}})$ by the above
definition of $c$, and thus the line segment $L_{k^{\prime }}$ does not
dominate the point $p_{k}$, which is a contradiction. Therefore $%
l_{k^{\prime }}\notin \mathbb{R}_{\text{right}}^{2}(\Gamma _{t^{\ast }}^{%
\text{diag}})$. Summarizing, $r_{k^{\prime }}\notin \mathbb{R}_{\text{left}%
}^{2}(\Gamma _{t^{\ast }}^{\text{vert}})$ and $l_{k^{\prime }}\notin \mathbb{%
R}_{\text{right}}^{2}(\Gamma _{t^{\ast }}^{\text{diag}})$, and thus $%
L_{k^{\prime }}\cap F_{t^{\ast }}\neq \emptyset $. That is, $t^{\ast }\in
S_{k^{\prime }}$ for some $L_{k^{\prime }}\in Z\setminus L_{i}$, which is a
contradiction as we proved above. Thus there does not exist such a point $%
p_{k}$, i.e.,~%
\begin{equation*}
\mathcal{P}\cap X(a,b)\cap F_{c}\cap F_{i}=\emptyset .
\end{equation*}

Assume that $t^{\ast }\in L_{i}$. Then, since $t^{\ast }\notin S_{k}$ for
every line segment $L_{k}\in Z\setminus L_{i}$ as we proved above, we can
partition the set $Z\setminus \{L_{q},L_{i},L_{i^{\prime }}\}$ into the sets 
$Z_{\text{below}}$, $Z_{\text{left}}$, and $Z_{\text{right}}$ as follows:%
\begin{eqnarray}
Z_{\text{below}} &=&\{L_{k}\in Z\setminus \{L_{q},L_{i},L_{i^{\prime
}}\}:L_{k}\cap S_{i}\neq \emptyset \},  \notag \\
Z_{\text{left}} &=&\{L_{k}\in Z\setminus \{L_{q},L_{i},L_{i^{\prime
}}\}:L_{k}\cap S_{i}=\emptyset ,L_{k}\subseteq \mathbb{R}_{\text{left}%
}^{2}(\Gamma _{t^{\ast }}^{\text{vert}})\},  \label{partition-case-1-eq-1} \\
Z_{\text{right}} &=&\{L_{k}\in Z\setminus \{L_{q},L_{i},L_{i^{\prime
}}\}:L_{k}\cap S_{i}=\emptyset ,L_{k}\subseteq \mathbb{R}_{\text{right}%
}^{2}(\Gamma _{t^{\ast }}^{\text{diag}})\}.  \notag
\end{eqnarray}

Assume now that $t^{\ast }\in F_{i}\setminus L_{i}$; then $t^{\ast }=t$ is a
point of $x_{0}$. Note that all points of $\mathcal{P}\cap X(a,b)\cap F_{i}$
are dominated by $Z\setminus L_{i}$, since they are not dominated by $L_{i}$
and $BD(a,b,q,i,i^{\prime })\neq \bot $ by assumption. Therefore $x_{0}$ is
a line segment, i.e.,~$x_{0}\in \mathcal{L}$. Assume that there exists a
line segment $L_{k}\in Z\setminus L_{i}$ such that $L_{k}\cap \left(
S_{t^{\ast }}\cup F_{t^{\ast }}\right) \neq \emptyset $. Then $x_{0}$ is
dominated by $L_{k}\in Z\setminus L_{i}$, which is a contradiction.
Therefore $L_{k}\cap \left( S_{t^{\ast }}\cup F_{t^{\ast }}\right)
=\emptyset $ for every line segment $L_{k}\in Z\setminus L_{i}$. That is,
for every $L_{k}\in Z\setminus L_{i}$ we have that either $L_{k}\subseteq
B_{t^{\ast }}$ or $L_{k}\subseteq A_{t^{\ast }}$. Therefore, in the case
where $t^{\ast }\in F_{i}\setminus L_{i}$, we can partition the set $%
Z\setminus \{L_{q},L_{i},L_{i^{\prime }}\}$ into the sets $Z_{\text{below}}$%
, $Z_{\text{left}}$, and $Z_{\text{right}}$ as follows:%
\begin{eqnarray}
Z_{\text{below}} &=&\emptyset ,  \notag \\
Z_{\text{left}} &=&\{L_{k}\in Z\setminus \{L_{q},L_{i},L_{i^{\prime
}}\}:L_{k}\subseteq B_{t^{\ast }}\},  \label{partition-case-1-eq-2} \\
Z_{\text{right}} &=&\{L_{k}\in Z\setminus \{L_{q},L_{i},L_{i^{\prime
}}\}:L_{k}\subseteq A_{t^{\ast }}\}.  \notag
\end{eqnarray}

Notice that, in both cases where $t^{\ast }\in L_{i}$ and $t^{\ast }\in
F_{i}\setminus L_{i}$, the set $Z_{1}=Z_{\text{below}}\cup Z_{\text{left}%
}\cup \{L_{q},L_{i},L_{i^{\prime }}\}$ is a dominating set of $X(a,c)$.
Furthermore the set $Z_{2}=Z_{\text{right}}\cup \{L_{q},L_{i},L_{i^{\prime
}}\}$ is a dominating set of $X(c,b)$. Moreover, $L_{q}$ is the diagonally
leftmost line segment and $(i,i^{\prime })$ is the end-pair of both $Z_{1}$
and $Z_{2}$. Therefore $|BD(a,c,q,i,i^{\prime })|\leq |Z_{1}|$ and $%
|BD(c,b,q,i,i^{\prime })|\leq |Z_{2}|$. Now, since $\{L_{q},L_{i},L_{i^{%
\prime }}\}\subseteq BD(a,c,q,i,i^{\prime })\cap BD(c,b,q,i,i^{\prime })$,
we have that%
\begin{eqnarray*}
|BD(a,c,q,i,i^{\prime })\cup BD(c,b,q,i,i^{\prime })| &\leq
&|BD(a,c,q,i,i^{\prime })|+|BD(c,b,q,i,i^{\prime
})|-|\{L_{q},L_{i},L_{i^{\prime }}\}| \\
&\leq &|Z_{1}|+|Z_{2}|-|\{L_{q},L_{i},L_{i^{\prime }}\}| \\
&=&|Z_{\text{below}}\cup Z_{\text{left}}\cup \{L_{q},L_{i},L_{i^{\prime }}\}|
\\
&&+|Z_{\text{right}}\cup \{L_{q},L_{i},L_{i^{\prime
}}\}|-|\{L_{q},L_{i},L_{i^{\prime }}\}| \\
&=&|Z_{\text{below}}|+|Z_{\text{left}}|+|Z_{\text{right}}|+|%
\{L_{q},L_{i},L_{i^{\prime }}\}| \\
&=&|Z|=|BD(a,b,q,i,i^{\prime })|.
\end{eqnarray*}%
\noindent

Finally Lemma~\ref{bounded-dom-correctness-lem-2-first-direction-lem-1}
implies that, if $BD(a,c,q,i,i^{\prime })\neq \bot $ and $%
BD(c,b,q,i,i^{\prime })\neq \bot $, then $BD(a,c,q,i,i^{\prime })\cup
BD(c,b,q,i,i^{\prime })$ is a dominating set of $X(a,b)$, in which $L_{q}$
is the diagonally leftmost line segment and $(i,i^{\prime })$ is the
end-pair. Therefore%
\begin{equation*}
|BD(a,b,q,i,i^{\prime })|\leq |BD(a,c,q,i,i^{\prime })\cup
BD(c,b,q,i,i^{\prime })|.
\end{equation*}%
It follows that $|BD(a,b,q,i,i^{\prime })|=|BD(a,c,q,i,i^{\prime })\cup
BD(c,b,q,i,i^{\prime })|$.

\medskip

\noindent \textbf{Case 2.} $x_{0}\cap \mathbb{R}_{\text{right}}^{2}(\Gamma
_{l_{i}}^{\text{diag}})=\emptyset $. Then, since $x_{0}\cap (S_{i}\cup
F_{i})\neq \emptyset $ by the initial assumption on $x_{0}$, it follows that 
$x_{0}\cap F_{i}\neq \emptyset $. Note that all points in $\mathcal{P}\cap
X(a,b)\cap F_{i}$ are dominated by $Z\setminus \{L_{i}\}$, since they are
not dominated by $L_{i}$ and $BD(a,b,q,i,i^{\prime })\neq \bot $ by
assumption. Therefore $x_{0}\in \mathcal{L}$. Let $t^{\ast }\in \mathbb{R}^{2}$ be an arbitrary point of $x_{0}\cap F_{i}$.

If $i^{\prime }\neq i$ and $l_{i^{\prime }}\in \mathbb{R}_{\text{left}%
}^{2}(\Gamma _{l_{i}}^{\text{diag}})$, then $L_{i^{\prime }}\in Z\setminus
\{L_{i}\}$ and $L_{i^{\prime }}$ dominates $x_{0}$, which is a
contradiction. Therefore, if $i^{\prime }\neq i$ then $l_{i^{\prime }}\notin 
\mathbb{R}_{\text{left}}^{2}(\Gamma _{l_{i}}^{\text{diag}})$. Furthermore,
it follows that if $L_{q}\neq L_{i}$ then also $L_{q}\neq L_{i^{\prime }}$.

Assume that there exists a line segment $L_{k}\in Z\setminus L_{i}$ such
that $L_{k}\cap \left( S_{t^{\ast }}\cup F_{t^{\ast }}\right) \neq \emptyset 
$. Then $x_{0}$ is dominated by $L_{k}\in Z\setminus L_{i}$, which is a
contradiction. Therefore $L_{k}\cap \left( S_{t^{\ast }}\cup F_{t^{\ast
}}\right) =\emptyset $ for every line segment $L_{k}\in Z\setminus L_{i}$.
That is, for every $L_{k}\in Z\setminus L_{i}$ we have that either $%
L_{k}\subseteq B_{t^{\ast }}$ or $L_{k}\subseteq A_{t^{\ast }}$. Therefore,
similarly to Eq.~(\ref{partition-case-1-eq-2}) in Case 1, we can partition
the set $Z\setminus \{L_{q},L_{i},L_{i^{\prime }}\}$ into the sets $Z_{\text{%
left}}$ and $Z_{\text{right}}$ as follows:%
\begin{eqnarray}
Z_{\text{left}} &=&\{L_{k}\in Z\setminus \{L_{q},L_{i},L_{i^{\prime
}}\}:L_{k}\subseteq B_{t^{\ast }}\},  \notag \\
Z_{\text{right}} &=&\{L_{k}\in Z\setminus \{L_{q},L_{i},L_{i^{\prime
}}\}:L_{k}\subseteq A_{t^{\ast }}\}.  \label{partition-case-2-eq}
\end{eqnarray}

Similarly to Case 1, let $j$ be the greatest index such that for the line
segment $L_{j}\in Z\setminus L_{i}$ we have $r_{j}\in \mathbb{R}_{\text{left}%
}^{2}(\Gamma _{t^{\ast }}^{\text{vert}})$. That is, for every other line
segment $L_{s}\in Z\setminus L_{i}$ with $r_{s}\in \mathbb{R}_{\text{left}%
}^{2}(\Gamma _{t^{\ast }}^{\text{vert}})$, we have $r_{s}\in \mathbb{R}_{%
\text{left}}^{2}(\Gamma _{l_{j}}^{\text{vert}})$. If $r_{j}\in \mathbb{R}_{%
\text{right}}^{2}(\Gamma _{l_{i}}^{\text{vert}})$ then we define $%
t_{1}=r_{j} $. If $r_{j}\notin \mathbb{R}_{\text{right}}^{2}(\Gamma
_{l_{i}}^{\text{vert}})$ then we define $t_{1}=l_{i}$. Furthermore, if such
a line segment $L_{j}$ does not exist in $Z\setminus L_{i}$ (i.e.,~if $%
r_{s}\notin \mathbb{R}_{\text{left}}^{2}(\Gamma _{t^{\ast }}^{\text{vert}})$
for every $L_{s}\in Z\setminus L_{i}$), then we define again $t_{1}=l_{i}$.

Let $L_{j^{\prime }}\in Z\setminus L_{i}$ be a line segment such that $%
l_{j^{\prime }}\in \mathbb{R}_{\text{right}}^{2}(\Gamma _{t^{\ast }}^{\text{%
diag}})$ and that, for every other line segment $L_{s}\in Z\setminus L_{i}$
with $l_{s}\in \mathbb{R}_{\text{right}}^{2}(\Gamma _{t^{\ast }}^{\text{diag}%
})$, we have $l_{s}\in \mathbb{R}_{\text{right}}^{2}(\Gamma _{l_{j^{\prime
}}}^{\text{diag}})$. If $l_{j^{\prime }}\in \mathbb{R}_{\text{left}%
}^{2}(\Gamma _{l_{i}}^{\text{diag}})$ then we define $L_{q^{\prime
}}=L_{j^{\prime }}$. If $l_{j^{\prime }}\notin \mathbb{R}_{\text{left}%
}^{2}(\Gamma _{l_{i}}^{\text{diag}})$ then we define $L_{q^{\prime }}=L_{i}$%
. Furthermore, if such a line segment $L_{j^{\prime }}$ does not exist in $%
Z\setminus L_{i}$ (i.e.,~if $l_{s}\notin \mathbb{R}_{\text{right}%
}^{2}(\Gamma _{t^{\ast }}^{\text{diag}})$ for every $L_{s}\in Z\setminus
L_{i}$), then we define again $L_{q^{\prime }}=L_{i}$.

Thus, in both cases where $L_{q^{\prime }}=L_{j^{\prime }}$ and $%
L_{q^{\prime }}=L_{i}$, it follows that $L_{q^{\prime }}\in \mathcal{%
\mathcal{L}}_{q}^{\text{right}}\cap \mathcal{L}_{i,i^{\prime }}^{\text{left}%
} $ and that $l_{q^{\prime }}\in F_{l_{i}}$. Note that it can be either $%
L_{q^{\prime }}\neq L_{q}$ or $L_{q^{\prime }}=L_{q}$. Furthermore recall
that, if $i^{\prime }\neq i$, then $l_{i^{\prime }}\notin \mathbb{R}_{\text{%
left}}^{2}(\Gamma _{l_{i}}^{\text{diag}})$ as we proved above. Therefore $%
L_{i},L_{i^{\prime }}\in \mathcal{\mathcal{L}}_{q^{\prime }}^{\text{right}}$.

Now we define the point $t_{2}$ as follows. If $l_{q^{\prime }}\in \mathbb{R}%
_{\text{left}}^{2}(\Gamma _{b}^{\text{diag}})$ then we define $%
t_{2}=l_{q^{\prime }}$. Otherwise, if $l_{q^{\prime }}\notin \mathbb{R}_{%
\text{left}}^{2}(\Gamma _{b}^{\text{diag}})$ then we define $t_{2}=b$.
Furthermore we define 
\begin{equation*}
c^{\prime }=\Gamma _{t_{1}}^{\text{vert}}\cap \Gamma _{t_{2}}^{\text{diag}}.
\end{equation*}%
Therefore, due to the above definition of $t_{1}$ and $t_{2}$, it follows
that $c^{\prime }\in \Gamma _{l_{q^{\prime }}}^{\text{diag}}$ or $c^{\prime
}\in \Gamma _{b}^{\text{diag}}$. Furthermore note that $c^{\prime }\in
S_{t^{\ast }}$. It is easy to check by the definition of $t_{1}$ and $t_{2}$
that $c^{\prime }\in \mathcal{B}\cap R(a,b)$ and that $c^{\prime }\in
F_{l_{i}}$. Since $c^{\prime }\in F_{l_{i}}$, note that $F_{c^{\prime
}}\subseteq F_{i}$, and thus $F_{c^{\prime }}\cap F_{i}=F_{c^{\prime }}$.
Thus, similarly to Case 1, we can prove that%
\begin{equation*}
\mathcal{P}\cap X(a,b)\cap F_{c^{\prime }}=\emptyset .
\end{equation*}

Now recall the partition of the set $Z\setminus \{L_{q},L_{i},L_{i^{\prime
}}\}$ into the sets $Z_{\text{left}}$ and $Z_{\text{right}}$, cf.~Eq.~(\ref%
{partition-case-2-eq}). Notice that the set $Z_{1}=Z_{\text{left}}\cup
\{L_{q},L_{i},L_{i^{\prime }}\}$ is a dominating set of $X(a,c^{\prime })$
and that the set $Z_{2}=Z_{\text{right}}\cup \{L_{q^{\prime
}},L_{i},L_{i^{\prime }}\}$ is a dominating set of $X(c^{\prime },b)$.
Furthermore, $L_{q}$ is the diagonally leftmost line segment of $Z_{1}$ and $%
(i,i^{\prime })$ is the end-pair of $Z_{1}$. Similarly, $L_{q^{\prime }}$ is
the diagonally leftmost line segment of $Z_{2}$ and $(i,i^{\prime })$ is the
end-pair of $Z_{2}$. Therefore $|BD(a,c^{\prime },q,i,i^{\prime })|\leq
|Z_{1}|$ and $|BD(c^{\prime },b,q^{\prime },i,i^{\prime })|\leq |Z_{2}|$.

Let first $L_{q}=L_{q^{\prime }}$. Then, since $\{L_{q},L_{i},L_{i^{\prime
}}\}\subseteq BD(a,c^{\prime },q,i,i^{\prime })\cup BD(c^{\prime
},b,q^{\prime },i,i^{\prime })$, it follows that%
\begin{eqnarray*}
|BD(a,c^{\prime },q,i,i^{\prime })\cup BD(c^{\prime },b,q^{\prime
},i,i^{\prime })| &\leq &|BD(a,c^{\prime },q,i,i^{\prime })|+|BD(c^{\prime
},b,q^{\prime },i,i^{\prime })|-|\{L_{q},L_{i},L_{i^{\prime }}\}| \\
&\leq &|Z_{1}|+|Z_{2}|-|\{L_{q},L_{i},L_{i^{\prime }}\}| \\
&=&|Z_{\text{left}}\cup \{L_{q},L_{i},L_{i^{\prime }}\}|+|Z_{\text{right}%
}\cup \{L_{q^{\prime }},L_{i},L_{i^{\prime }}\}|-|\{L_{q},L_{i},L_{i^{\prime
}}\}| \\
&=&|Z_{\text{left}}|+|Z_{\text{right}}|+|\{L_{q},L_{i},L_{i^{\prime }}\}| \\
&=&|Z|=|BD(a,b,q,i,i^{\prime })|.
\end{eqnarray*}

Let now $L_{q}\neq L_{q^{\prime }}$. Then $l_{q^{\prime }}\in \mathbb{R}_{%
\text{right}}^{2}(\Gamma _{l_{q}}^{\text{diag}})$, since $L_{q^{\prime }}\in 
\mathcal{\mathcal{L}}_{q}^{\text{right}}$ as we proved above. Furthermore,
since $l_{q^{\prime }}\in \mathbb{R}_{\text{left}}^{2}(\Gamma _{l_{i}}^{%
\text{diag}})$ by definition of $q^{\prime }$, it follows that $L_{q}\neq
L_{i}$. Therefore also $L_{q}\neq L_{i^{\prime }}$, as we proved above.
Moreover, if $L_{q^{\prime }}\neq L_{i}$ then $L_{q^{\prime }}=L_{j^{\prime
}}$ by the above definition of $q^{\prime }$, and thus $L_{q^{\prime }}\in
Z_{\text{right}}$. Therefore, in both cases where $L_{q^{\prime }}=L_{i}$
and $L_{q^{\prime }}\neq L_{i}$, we have $Z_{2}=Z_{\text{right}}\cup
\{L_{q^{\prime }},L_{i},L_{i^{\prime }}\}=Z_{\text{right}}\cup
\{L_{i},L_{i^{\prime }}\}$. Thus, since $\{L_{i},L_{i^{\prime }}\}\subseteq
BD(a,c^{\prime },q,i,i^{\prime })\cap BD(c^{\prime },b,q^{\prime
},i,i^{\prime })$, it follows that%
\begin{eqnarray*}
|BD(a,c^{\prime },q,i,i^{\prime })\cup BD(c^{\prime },b,q^{\prime
},i,i^{\prime })| &\leq &|BD(a,c^{\prime },q,i,i^{\prime })|+|BD(c^{\prime
},b,q^{\prime },i,i^{\prime })|-|\{L_{i},L_{i^{\prime }}\}| \\
&\leq &|Z_{1}|+|Z_{2}|-|\{L_{i},L_{i^{\prime }}\}| \\
&=&|Z_{\text{left}}\cup \{L_{q},L_{i},L_{i^{\prime }}\}|+|Z_{\text{right}%
}\cup \{L_{i},L_{i^{\prime }}\}|-|\{L_{i},L_{i^{\prime }}\}| \\
&=&|Z_{\text{left}}\cup \{L_{q}\}|+|Z_{\text{right}}|+|\{L_{i},L_{i^{\prime
}}\}| \\
&=&|Z_{\text{left}}|+|Z_{\text{right}}|+|\{L_{q}\}|+|\{L_{i},L_{i^{\prime
}}\}| \\
&=&|Z|=|BD(a,b,q,i,i^{\prime })|.
\end{eqnarray*}

Finally Lemma~\ref{bounded-dom-correctness-lem-2-first-direction-lem-2}
implies that, if $BD(a,c^{\prime },q,i,i^{\prime })\neq \bot $ and $%
BD(c^{\prime },b,q^{\prime },i,i^{\prime })\neq \bot $ then $BD(a,c^{\prime
},q,i,i^{\prime })\cup BD(c^{\prime },b,q^{\prime },i,i^{\prime })$ is a
dominating set of $X(a,b)$, in which $L_{q}$ is the diagonally leftmost line
segment and $(i,i^{\prime })$ is the end-pair. Therefore%
\begin{equation*}
|BD(a,b,q,i,i^{\prime })|\leq |BD(a,c^{\prime },q,i,i^{\prime })\cup
BD(c^{\prime },b,q^{\prime },i,i^{\prime })|.
\end{equation*}%
It follows that $|BD(a,b,q,i,i^{\prime })|=|BD(a,c^{\prime },q,i,i^{\prime
})\cup BD(c^{\prime },b,q^{\prime },i,i^{\prime })|$.

\medskip

Summarizing Case 1 and Case 2, it follows that the value of $%
BD(a,b,q,i,i^{\prime })$ can be computed by Eq.~(\ref{recursion-eq-2}),
where the minimum is taken over all values of $c,c^{\prime },q^{\prime }$,
as stated in the lemma.
\end{proof}

\medskip

Using the recursive computations of Lemmas~\ref%
{bounded-dom-correctness-lem-0},~\ref{bounded-dom-correctness-lem-1}, and~%
\ref{bounded-dom-correctness-lem-2}, we are now ready to present Algorithm~%
\ref{bounded-dominating-tolerance-alg} for computing \textsc{Bounded
Dominating Set} on tolerance graphs in polynomial time.

\begin{algorithm}[t!]
\caption{\textsc{Bounded Dominating Set} on Tolerance Graphs} \label{bounded-dominating-tolerance-alg}
\begin{algorithmic} [1]
\REQUIRE{A horizontal shadow representation $(\mathcal{P},\mathcal{L})$, where ${\mathcal{P} = \{p_{1},p_{2}, \ldots, p_{|\mathcal{P}|}\}}$ and 
${\mathcal{L} = \{L_{1},L_{2}, \ldots, L_{|\mathcal{L}|}\}}$}
\ENSURE{A set $Z\subseteq \mathcal{L}$ of minimum size that dominates $(\mathcal{P},\mathcal{L})$, 
or the announcement that $\mathcal{L}$ does not dominate~${(\mathcal{P},\mathcal{L})}$}

\medskip

\STATE{Add two dummy line segments $L_{0}$ and $L_{|\mathcal{L}|+1}$ completely to the left and to the right of $\mathcal{P} \cup \mathcal{L}$, respectively} \label{alg-1-line-1}

\STATE{$\mathcal{L} \leftarrow \mathcal{L} \cup \{L_{0}, L_{|\mathcal{L}|+1}\}$; \ \ 
denote ${\mathcal{L} = \{L_{1},L_{2}, \ldots, L_{|\mathcal{L}|}\}}$, where now $L_{1}$ and $L_{|\mathcal{L}|}$ are dummy} \label{alg-1-line-2}  %%%%line segments

\vspace{0,1cm}

\STATE{$\mathcal{A} \leftarrow \{l_{i},r_{i}:1\leq i\leq |\mathcal{L}|\}$; \ \ 
$\mathcal{B} \leftarrow \{\Gamma _{t}^{\text{diag} }\cap \Gamma _{t^{\prime }}^{\text{vert}}:t,t^{\prime }\in \mathcal{A}\}$} \label{alg-1-line-3}

\vspace{0,1cm}

\FOR[initialization]{every pair of points $(a,b) \in \mathcal{A} \times \mathcal{B}$ such that $b \in \mathbb{R}_{\text{right}}^{2}(\Gamma_{a}^{\text{diag}})$} \label{alg-1-line-4}

     \STATE{$X(a,b) \leftarrow \{x\in \mathcal{P}\cup \mathcal{L}:x\subseteq \left( B_{b}\setminus \Gamma _{b}^{\text{vert}}\right) \cap \mathbb{R}_{\text{right}}^{2}(\Gamma_{a}^{\text{diag}})\}$} \label{alg-1-line-5}

     \FOR{every $q,i,i^{\prime} \in \{1,2,\ldots, |\mathcal{L}|\}$} \label{alg-1-line-6}
          \IF[$(i,i^{\prime})$ is a right-crossing pair]{$r_{i^{\prime}} \in S_{r_{i}}$}\label{alg-1-line-7}
               
               \IF{$L_{q}\in \mathcal{L}_{i,i^{\prime }}^{\text{left}}$, \ $L_{i},L_{i^{\prime }}\in \mathcal{\mathcal{L}}_{q}^{\text{right}}$, \ and \ $b\in \mathbb{R}_{\text{left}}^{2}(\Gamma _{r_{i}}^{\text{vert}})$}\label{alg-1-line-8}

                    \vspace{0,1cm}
                    \STATE{$\mathcal{L}_{i,i^{\prime }}^{\text{left}} \leftarrow \{x\in \mathcal{P}\cup \mathcal{L} : x\subseteq B_{t}$,\ where $t=\Gamma _{r_{i}}^{\text{vert}}\cap \Gamma_{r_{i^{\prime }}}^{\text{diag}}\}$}\label{alg-1-line-9}
                    \STATE{$\mathcal{L}_{q}^{\text{right}} \leftarrow \{x\in \mathcal{P}\cup \mathcal{L}:x\subseteq 
\mathbb{R}_{\text{right}}^{2}(\Gamma _{l_{q}}^{\text{diag}})\}$}\label{alg-1-line-10}

                    \vspace{0,1cm}
                    
                    \IF{$\mathcal{L\cap \mathcal{L}}_{q}^{\text{right}}\mathcal{\cap L}_{i,i^{\prime}}^{\text{left}}$ does not dominate all elements of $X(a,b)$}\label{alg-1-line-11}
                         \STATE{$BD(a,b,q,i,i^{\prime}) \leftarrow \bot$}\label{alg-1-line-12}
                    \ELSIF{$\{L_{q},L_{i},L_{i^{\prime }}\}$ dominates all elements of $X(a,b)$}\label{alg-1-line-13}
                         \STATE{$BD(a,b,q,i,i^{\prime}) \leftarrow \{L_{q},L_{i},L_{i^{\prime }}\}$}\label{alg-1-line-14}
                    \ELSE \label{alg-1-line-15}
                         \STATE{$BD(a,b,q,i,i^{\prime}) \leftarrow \mathcal{L\cap \mathcal{L}}_{q}^{\text{right}}\mathcal{\cap L}_{i,i^{\prime}}^{\text{left}}$} \COMMENT{initialization}\label{alg-1-line-16}
                    \ENDIF
                      
               \ENDIF
          \ENDIF
     \ENDFOR 
\ENDFOR

\medskip

\FOR{every pair of points $(a,b) \in \mathcal{A} \times \mathcal{B}$ such that $b \in \mathbb{R}_{\text{right}}^{2}(\Gamma_{a}^{\text{diag}})$}\label{alg-1-line-17}
     \FOR{every $q,i,i^{\prime} \in \{1,2,\ldots, |\mathcal{L}|\}$}\label{alg-1-line-18}
          \IF[$(i,i^{\prime})$ is a right-crossing pair]{$r_{i^{\prime}} \in S_{r_{i}}$}\label{alg-1-line-19}
               \IF{$L_{q}\in \mathcal{L}_{i,i^{\prime }}^{\text{left}}$, \ $L_{i},L_{i^{\prime }}\in \mathcal{\mathcal{L}}_{q}^{\text{right}}$, \ and \ $b\in \mathbb{R}_{\text{left}}^{2}(\Gamma _{r_{i}}^{\text{vert}})$}\label{alg-1-line-20}

                    \vspace{0,1cm}
                    \STATE{Compute the solutions $Z_{1}, Z_{2}, Z_{3}$ by Lemmas~\ref{bounded-dom-correctness-lem-0},~\ref{bounded-dom-correctness-lem-1}, and~\ref{bounded-dom-correctness-lem-2}, respectively}\label{alg-1-line-21}
                    \FOR{$k=1$ to $3$}\label{alg-1-line-22}
                         \STATE{\textbf{if} $|Z_{k}| < |BD(a,b,q,i,i^{\prime})|$ \textbf{then} $BD(a,b,q,i,i^{\prime}) \leftarrow Z_{k}$}\label{alg-1-line-23}
                    \ENDFOR
               \ENDIF
          \ENDIF
     \ENDFOR 
\ENDFOR 

\medskip

\STATE{\textbf{if} $BD(l_{1},r_{\mathcal{L}},1,|\mathcal{L}|,|\mathcal{L}|) = \bot$ \textbf{then} \textbf{return} \ $\mathcal{L}$ does not dominate~${(\mathcal{P},\mathcal{L})}$}\label{alg-1-line-24}

\STATE{\ \ \ \ \textbf{else} \textbf{return} \ $BD(l_{1},r_{\mathcal{L}},1,|\mathcal{L}|,|\mathcal{L}|) \setminus \{L_{1}, L_{|\mathcal{L}|}\}$}\label{alg-1-line-25}
\end{algorithmic}
\end{algorithm}

\begin{theorem}
\label{bounded-correctness-thm}Given a horizontal shadow representation $(%
\mathcal{P},\mathcal{L})$ of a tolerance graph $G$ with $n$ vertices,
Algorithm~\ref{bounded-dominating-tolerance-alg} solves \textsc{Bounded
Dominating Set} in $O(n^{9})$ time.
\end{theorem}

\begin{proof}
In the first line, Algorithm~\ref{bounded-dominating-tolerance-alg} augments
the horizontal shadow representation $(\mathcal{P},\mathcal{L})$ by adding
to $\mathcal{L}$ the two dummy line segments $L_{0}$ and $L_{|\mathcal{L}%
|+1} $ (with endpoints $l_{0},r_{0}$ and $l_{|\mathcal{L}|+1},r_{|\mathcal{L}%
|+1}$, respectively) such that all elements of $\mathcal{P\cup L}$ are
contained in $A_{r_{0}}$ and in $B_{l_{|\mathcal{L}|+1}}$. In the second
line the algorithm renumbers the elements of the set $\mathcal{L}$ such that 
$\mathcal{L}=\{L_{1},L_{2},\ldots ,L_{|\mathcal{L}|}\}$, where in this new
enumeration the line segments $L_{1}$ and $L_{|\mathcal{L}|}$ are dummy.
Furthermore, in line~\ref{alg-1-line-3}, the algorithm computes the point
sets $\mathcal{A}$ and $\mathcal{B}$ (cf.~Section~\ref%
{terminology-bounded-domination-subsec}).

In lines~\ref{alg-1-line-4}-\ref{alg-1-line-16} the algorithm performs all
initializations. In particular, first in line~\ref{alg-1-line-5} the
algorithm computes the sets $X(a,b)\subseteq \mathcal{P}\cup \mathcal{L}$
for all feasible pairs $(a,b)\in \mathcal{A\times B}$ (cf.~Eq.~(\ref%
{X(a,b)-def-eq})). Then the algorithm iteratively executes lines~\ref%
{alg-1-line-9}-\ref{alg-1-line-16} for all values of $q,i,i^{\prime }\in
\{1,2,\ldots ,|\mathcal{L}|\}$ for which $BD(a,b,q,i,i^{\prime })$ can be
defined (these conditions on $q,i,i^{\prime }$ are tested in lines~\ref%
{alg-1-line-6}-\ref{alg-1-line-8}, cf.~Definition~\ref{BD-def}). For all
such values of $q,i,i^{\prime }$, the algorithm computes an initial value
for $BD(a,b,q,i,i^{\prime })$ in lines~\ref{alg-1-line-9}-\ref{alg-1-line-16}%
. In particular, in lines~\ref{alg-1-line-12} and~\ref{alg-1-line-14} it
computes the values of $BD(a,b,q,i,i^{\prime })$ which can be computed
directly (cf.~Observations~\ref{obs:botcs} and~\ref{DB-init-feasible-obs}).
In the case where $BD(a,b,q,i,i^{\prime })\neq \bot $ and $%
BD(a,b,q,i,i^{\prime })\neq \{L_{q},L_{i},L_{i^{\prime }}\}$, the set $%
\mathcal{L\cap \mathcal{L}}_{q}^{\text{right}}\mathcal{\cap L}_{i,i^{\prime
}}^{\text{left}}$ is a feasible (but not necessarily optimal) solution (cf.
Definition~\ref{BD-def}), therefore in this case the algorithm initializes
in line~\ref{alg-1-line-16} the value of $BD(a,b,q,i,i^{\prime })$ to $%
\mathcal{L\cap \mathcal{L}}_{q}^{\text{right}}\mathcal{\cap L}_{i,i^{\prime
}}^{\text{left}}$.

The main computations of the algorithm are performed in lines~\ref%
{alg-1-line-17}-\ref{alg-1-line-23}. In particular, the algorithm
iteratively executes lines~\ref{alg-1-line-21}-\ref{alg-1-line-23} for all
values of $a,b,q,i,i^{\prime }$ for which $BD(a,b,q,i,i^{\prime })$ can be
defined (these conditions on $a,b,q,i,i^{\prime }$ are tested in lines~\ref%
{alg-1-line-17}-\ref{alg-1-line-20}, cf.~Definition~\ref{BD-def}). In line~%
\ref{alg-1-line-21} the algorithm computes all the necessary values that are
the candidates for the value $BD(a,b,q,i,i^{\prime })$ and in lines~\ref%
{alg-1-line-22}-\ref{alg-1-line-23} the algorithm computes $%
BD(a,b,q,i,i^{\prime })$ from these candidate values. The correctness of
this computation of $BD(a,b,q,i,i^{\prime })$ follows by Lemmas~\ref%
{bounded-dom-correctness-lem-0},~\ref{bounded-dom-correctness-lem-1}, and~%
\ref{bounded-dom-correctness-lem-2}, respectively.

Finally, the algorithm computes the final output in lines~\ref{alg-1-line-24}%
-\ref{alg-1-line-25}. Indeed, since in the (augmented) horizontal shadow
representation $(\mathcal{P},\mathcal{L})$ the two dummy horizontal line
segments are isolated (i.e.,~the line segments $L_{1}$ and $L_{|\mathcal{L}%
|} $ in the augmented representation, cf.~lines~\ref{alg-1-line-1}-\ref%
{alg-1-line-2} of the algorithm), they must be included in every minimum
bounded dominating set of the (augmented) tolerance graph. Therefore the
algorithm correctly returns in line~\ref{alg-1-line-25} the computed set $%
BD(l_{1},r_{|\mathcal{L}|},1,|\mathcal{L}|,|\mathcal{L}|)\setminus
\{L_{1},L_{|\mathcal{L}|}\}$, as long as $BD(l_{1},r_{|\mathcal{L}|},1,|%
\mathcal{L}|,|\mathcal{L}|)\neq \bot $. Furthermore, if $BD(l_{1},r_{|%
\mathcal{L}|},1,|\mathcal{L}|,|\mathcal{L}|)=\bot $ then the whole
(augmented) set $\mathcal{L}$ does not dominate all elements of the
(augmented) set $\mathcal{P}\cup \mathcal{L}$, and thus in this case the
algorithm correctly returns a negative announcement in line~\ref%
{alg-1-line-24}.

Regarding the running time of Algorithm~\ref%
{bounded-dominating-tolerance-alg}, first recall that the sets $\mathcal{A}$
and~$\mathcal{B}$ have~$O(n)$ and~$O(n^{2})$ elements, respectively. Thus
the first three lines of the algorithm can be implemented in~$O(n^{2})$
time. Due to the for-loop of line~\ref{alg-1-line-4}, the lines~\ref%
{alg-1-line-5}-\ref{alg-1-line-16} are executed at most~$O(n^{3})$ times.
Recall by Eq.~(\ref{region-R(a,b)-def-eq}) and~(\ref{X(a,b)-def-eq}) that,
for every pair $(a,b)\in \mathcal{A}\times \mathcal{B}$, the region~$R(a,b)$
can be specified in constant time (cf.~the shaded region in Figure~\ref%
{X-a-b-fig}) and the vertex set $X(a,b)$ can be computed in $O(n)$ time.
That is, line~\ref{alg-1-line-5} of the algorithm can be executed in $O(n)$
time. For every fixed pair $(a,b)$, the lines~\ref{alg-1-line-7}-\ref%
{alg-1-line-16} are executed at most $O(n^{3})$ times, due to the for-loop
of line~\ref{alg-1-line-6}. Furthermore the if-statements of lines~\ref%
{alg-1-line-7} and~\ref{alg-1-line-8} can be executed in constant time,
while the computations of $\mathcal{L}_{i,i^{\prime }}^{\text{left}}$ and $%
\mathcal{\mathcal{L}}_{q}^{\text{right}}$ in lines~\ref{alg-1-line-9} and %
\ref{alg-1-line-10} can be computed in $O(n)$ time each. The if-statement of
line~\ref{alg-1-line-11} can be executed in $O(n^{2})$ time, since in the
worst case we check adjacency between each element of $\mathcal{\mathcal{%
L\cap L}}_{q}^{\text{right}}\cap \mathcal{L}_{i,i^{\prime }}^{\text{left}}$\
and each element of $X(a,b)$. Moreover, each of the lines~\ref{alg-1-line-12}%
-\ref{alg-1-line-16} can be trivially executed in at most $O(n)$ time.
Therefore the total execution time of lines~\ref{alg-1-line-4}-\ref%
{alg-1-line-16} is $O(n^{8})$.

Due to the for-loop of lines~\ref{alg-1-line-17} and~\ref{alg-1-line-18},
the lines~\ref{alg-1-line-19}-\ref{alg-1-line-23} are executed at most~$%
O(n^{6})$ times, since there exist at most $O(n^{3})$ pairs $(a,b)$ and at
most $O(n^{3})$ triples $\{q,i,i^{\prime }\}$. Furthermore, since each of
the lines~\ref{alg-1-line-19} and~\ref{alg-1-line-20} can be executed in
constant time, the execution time of the lines~\ref{alg-1-line-19}-\ref%
{alg-1-line-23} is dominated by the execution time of line~\ref%
{alg-1-line-21}, i.e., by the recursive computation of the set $%
BD(a,b,q,i,i^{\prime })$ from Lemmas~\ref{bounded-dom-correctness-lem-0}, %
\ref{bounded-dom-correctness-lem-1}, and~\ref{bounded-dom-correctness-lem-2}%
. Note that we have already computed in lines~\ref{alg-1-line-12} and~\ref%
{alg-1-line-14} of the algorithm whether $BD(a,b,q,i,i^{\prime })\neq \bot $
and $BD(a,b,q,i,i^{\prime })\neq \{L_{q},L_{i},L_{i^{\prime }}\}$. Moreover
it can also be checked in constant time whether $R(a,b)\nsubseteq S_{i}$ and
whether $b\in S_{l_{i}}$, and thus we can decide in constant time in line~%
\ref{alg-1-line-21} whether Lemmas~\ref{bounded-dom-correctness-lem-0},~\ref%
{bounded-dom-correctness-lem-1}, and~\ref{bounded-dom-correctness-lem-2} can
be applied. If Lemma~\ref{bounded-dom-correctness-lem-0} can be applied, the
corresponding candidate for $BD(a,b,q,i,i^{\prime })$ can be computed in
constant time by a previously computed value (cf.~Eq.~(\ref{recursion-eq-0})).

Assume now that Lemma~\ref{bounded-dom-correctness-lem-1} can be applied.
Then the corresponding candidate for $BD(a,b,q,i,i^{\prime })$ is computed
by the right-hand side of Eq.~(\ref{recursion-eq-1}), for all values of $%
c,q^{\prime },j,j^{\prime }$ that satisfy the conditions of Lemma~\ref%
{bounded-dom-correctness-lem-1-first-direction}. Note by Condition~2 of
Lemma~\ref{bounded-dom-correctness-lem-1-first-direction} that, if $i\neq
i^{\prime }$, then $j^{\prime }=i^{\prime }$. Therefore every feasible
quadruple $(i,i^{\prime },j,j^{\prime })$ is either $(i,i,j,j^{\prime })$ or 
$(i,i^{\prime },j,i^{\prime })$, i.e., there exist at most $O(n^{3})$
feasible quadruples $(i,i^{\prime },j,j^{\prime })$. Thus, since we already
considered $O(n^{2})$ iterations for all pairs $(i,i^{\prime })$ in line~\ref%
{alg-1-line-18} of the algorithm, we only need to consider another~$O(n)$
iterations (multiplicatively) in line~\ref{alg-1-line-21} for all feasible
pairs $(j,j^{\prime })$ in the execution of Lemma~\ref%
{bounded-dom-correctness-lem-1}. Furthermore there are at most~$O(n)$
feasible values of $q^{\prime }$ by Conditions~1 and~3 of Lemma~\ref%
{bounded-dom-correctness-lem-1-first-direction}. Moreover the value of $c$
is uniquely determined (in constant time) by the values of $j$ and $b$
(cf.~Condition~4 of Lemma~\ref{bounded-dom-correctness-lem-1-first-direction}%
); once $c$ has been computed, we also need $O(n)$ additional time to check
Condition~5 of Lemma~\ref{bounded-dom-correctness-lem-1-first-direction}.
Therefore, Lemma~\ref{bounded-dom-correctness-lem-1} can be applied in $%
O(n^{3})$ time in line~\ref{alg-1-line-21} of the algorithm.

Assume finally that Lemma~\ref{bounded-dom-correctness-lem-2} can be
applied. Then the corresponding candidate for $BD(a,b,q,i,i^{\prime })$ is
computed by the right-hand side of Eq.~(\ref{recursion-eq-2}), for all
values of $c,c^{\prime },q^{\prime }$ that satisfy the conditions of Lemma~%
\ref{bounded-dom-correctness-lem-2}. Note that there exist $O(n^{2})$
feasible values for $c$, cf.~Conditions~1 and~2 of Lemma~\ref%
{bounded-dom-correctness-lem-2}. Furthermore, once the value of $c$ has been
chosen, we need $O(n)$ additional time to check Condition~6 of Lemma~\ref%
{bounded-dom-correctness-lem-2}. Thus, the upper part of the right-hand side
of Eq.~(\ref{recursion-eq-2}) can be computed in $O(n^{3})$ time. On the
other hand, there exist $O(n)$ feasible values for $q^{\prime }$, cf.
Conditions~3 and~4 of Lemma~\ref{bounded-dom-correctness-lem-2}. For every
value of $q^{\prime }$ there exist $O(n)$ feasible values for $c^{\prime }$,
cf.~Condition~5 of Lemma~\ref{bounded-dom-correctness-lem-2}; once the value
of $c^{\prime }$ has been chosen, we need $O(n)$ additional time to check
Condition~6 of Lemma~\ref{bounded-dom-correctness-lem-2}. Thus, the lower
part of the right-hand side of Eq.~(\ref{recursion-eq-2}) can be also
computed in $O(n^{3})$ time. That is, Lemma~\ref%
{bounded-dom-correctness-lem-2} can be applied in $O(n^{3})$ time in line~%
\ref{alg-1-line-21} of the algorithm.

Summarizing, the total execution time of the lines~\ref{alg-1-line-17}-\ref%
{alg-1-line-23} is~$O(n^{9})$. Therefore, since the execution time of lines~%
\ref{alg-1-line-4}-\ref{alg-1-line-16} is $O(n^{8})$, the total running time
of Algorithm~\ref{bounded-dominating-tolerance-alg} is~$O(n^{9})$.
\end{proof}

\section{Restricted bounded dominating set on tolerance graphs\label%
{Restricted-domination-sec}}

In this section we use Algorithm~\ref{bounded-dominating-tolerance-alg} of
Section~\ref{Bounded-dominating-sec} to provide a polynomial time algorithm
(cf.~Algorithm~\ref{restricted-bounded-alg}) for a slightly modified version
of \textsc{Bounded Dominating Set} on tolerance graphs, which we call 
\textsc{Restricted Bounded Dominating Set}, formally defined below.

\vspace{0,2cm} \noindent \fbox{ 
\begin{minipage}{0.96\textwidth}
 \begin{tabular*}{\textwidth}{@{\extracolsep{\fill}}lr} \rbdsprob & \\ \end{tabular*}
 
 \vspace{1.2mm}
{\bf{Input:}} A 6-tuple ${\cal I}=(\PP, \LL,j,j',i,i')$, where $(\PP,\LL)$ is a horizontal shadow representation of a tolerance graph $G$, 
$(j,j')$ is a left-crossing pair of $G$, and $(i,i')$ is a right-crossing pair of $G$.
 \\
{\bf{Output:}} A set $Z \subseteq \LL$ of minimum size that dominates $(\PP,\LL)$, where $(j,j')$ is the start-pair~and $(i,i')$ is the end-pair of $Z$, 
or the announcement that $\LL \cap \LL^{\text{right}}_{j,j'}\cap \LL^{\text{left}}_{i,i'}$ does not dominate $(\PP,\LL)$.
\end{minipage}} \vspace{0,2cm}

In order to present Algorithm~\ref{restricted-bounded-alg} for \textsc{%
Restricted Bounded Dominating Set} on tolerance graphs, we first reduce this
problem to \textsc{Bounded Dominating Set} on tolerance graphs, cf.~Lemma~%
\ref{lem:sbssol2-first}. Before we present this reduction to \textsc{Bounded
Dominating Set}, we first need to prove some properties in the following
auxiliary Lemmas~\ref{lem:badvrt1}-\ref{obs:irlvvrt2}. These properties will
motivate the definition of bad and irrelevant points $p\in \mathcal{P}$ and
of bad and irrelevant line segments $L_{t}\in \mathcal{L}$, cf.~Definition %
\ref{bad-irrelevant-def}. The main idea behind Definition~\ref%
{bad-irrelevant-def} is the following. If an instance contains a bad point $%
p\in \mathcal{P}$ or a bad line segment $L_{t}\in \mathcal{L}$, then $%
\mathcal{L\cap L}_{j,j^{\prime }}^{\text{right}}\cap \mathcal{L}%
_{i,i^{\prime }}^{\text{left}}$ does not dominate $(\mathcal{P},\mathcal{L})$%
. On the other hand, if an instance contains an irrelevant point $p\in 
\mathcal{P}$ or an irrelevant line segment $L_{t}\in \mathcal{L}$, we can
safely ignore $p$ (resp. $L_{t}$).

\begin{lemma}
\label{lem:badvrt1}Let $\mathcal{I}=(\mathcal{P},\mathcal{L},j,j^{\prime
},i,i^{\prime })$ be an instance of \textsc{Restricted Bounded Dominating Set%
} on tolerance graphs. Let $l=\Gamma _{l_{j}}^{\text{vert}}\cap \Gamma
_{l_{j^{\prime }}}^{\text{diag}}$ and $r=\Gamma _{r_{i}}^{\text{vert}}\cap
\Gamma _{r_{i^{\prime }}}^{\text{diag}}$. If there exists a point $p\in 
\mathcal{P}$ such that $p\in \mathbb{R}_{\text{left}}^{2}(\Gamma _{l}^{\text{%
diag}})$ or $p\in \mathbb{R}_{\text{right}}^{2}(\Gamma _{r}^{\text{vert}})$,
then $\mathcal{L\cap L}_{j,j^{\prime }}^{\text{right}}\cap \mathcal{L}%
_{i,i^{\prime }}^{\text{left}}$ does not dominate $(\mathcal{P},\mathcal{L})$%
.
\end{lemma}

\begin{proof}
Assume otherwise that $Z\subseteq \mathcal{L}$ is a solution of $\mathcal{I}$%
. First suppose that there exists a point $p\in \mathcal{P}$ such that $p\in 
\mathbb{R}_{\text{left}}^{2}(\Gamma _{l}^{\text{diag}})$, where $l=\Gamma
_{l_{j}}^{\text{vert}}\cap \Gamma _{l_{j^{\prime }}}^{\text{diag}}$. Then,
by Lemma~\ref{shadow-correctness-lem-2}, there must exist a line segment $%
L_{k}\in Z$ such that $p\in S_{k}$. Thus $l_{k}\in \mathbb{R}_{\text{left}%
}^{2}(\Gamma _{l_{j^{\prime }}}^{\text{diag}})$, which is a contradiction to
the fact that $(j,j^{\prime })$ is the start-pair of $Z$.

Now suppose that there exists a point $p\in \mathcal{P}$ such that $p\in 
\mathbb{R}_{\text{right}}^{2}(\Gamma _{r}^{\text{vert}})$, where $r=\Gamma
_{r_{i}}^{\text{vert}}\cap \Gamma _{r_{i^{\prime }}}^{\text{diag}}$. Then,
by Lemma~\ref{shadow-correctness-lem-2}, there must exist a line segment $%
L_{k}\in Z$ such that $p\in S_{k}$. Thus $r_{k}\in \mathbb{R}_{\text{right}%
}^{2}(\Gamma _{r_{i}}^{\text{vert}})$, which is a contradiction to the fact
that $(i,i^{\prime })$ is the end-pair of $Z$.
\end{proof}

\begin{lemma}
\label{irrelevant-points-lem}Let $\mathcal{I}=(\mathcal{P},\mathcal{L}%
,j,j^{\prime },i,i^{\prime })$ be an instance of \textsc{Restricted Bounded
Dominating Set} on tolerance graphs. Let $l=\Gamma _{l_{j}}^{\text{vert}%
}\cap \Gamma _{l_{j^{\prime }}}^{\text{diag}}$ and $r=\Gamma _{r_{i}}^{\text{%
vert}}\cap \Gamma _{r_{i^{\prime }}}^{\text{diag}}$. If there exists a point 
$p\in \mathcal{P}$ such that $p\in S_{l}\cup S_{r}$ then at least one of the
line segments $\{L_{j^{\prime }},L_{i}\}$ is a neighbor of $p$.
\end{lemma}

\begin{proof}
Recall by Definition~\ref{left-right-crossing-pair-def} in Section~\ref%
{terminology-bounded-domination-subsec} that ${l_{j}\in S_{l_{j^{\prime }}}}$
and $r_{i^{\prime }}\in S_{r_{i}}$, since $(j,j^{\prime })$ is a
left-crossing pair and $(i,i^{\prime })$ is a right-crossing pair.
Therefore, since $l=\Gamma _{l_{j}}^{\text{vert}}\cap \Gamma _{l_{j^{\prime
}}}^{\text{diag}}$ and $r=\Gamma _{r_{i}}^{\text{vert}}\cap \Gamma
_{r_{i^{\prime }}}^{\text{diag}}$ by the assumptions of the lemma, it
follows that $l\in {S_{l_{j^{\prime }}}}$ and $r\in S_{r_{i}}$.

If $p\in S_{l}$ then also $p\in {S_{l_{j^{\prime }}}}$ (since $l\in {%
S_{l_{j^{\prime }}}}$ as we proved above), and thus $L_{j^{\prime }}$ is a
neighbor of $p$ by Lemma~\ref{shadow-correctness-lem-2}. Similarly, if $p\in
S_{r}$ then also $p\in S_{r_{i}}$ (since $r\in S_{r_{i}}$ as we proved
above), and thus $L_{i}$ is a neighbor of $p$ by Lemma~\ref%
{shadow-correctness-lem-2}.
\end{proof}

\begin{lemma}
\label{lem:badvrt2} Let $\mathcal{I}=(\mathcal{P},\mathcal{L},j,j^{\prime
},i,i^{\prime })$ be an instance of \textsc{Restricted Bounded Dominating Set%
} on tolerance graphs. Let $l=\Gamma _{l_{j}}^{\text{vert}}\cap \Gamma
_{l_{j^{\prime }}}^{\text{diag}}$ and $r=\Gamma _{r_{i}}^{\text{vert}}\cap
\Gamma _{r_{i^{\prime }}}^{\text{diag}}$. If there exists a line segment $%
L_{t}\in \mathcal{L}$ such that $L_{t}\subseteq B_{l}$ or $L_{t}\subseteq
A_{r}$, then $\mathcal{L\cap L}_{j,j^{\prime }}^{\text{right}}\cap \mathcal{L%
}_{i,i^{\prime }}^{\text{left}}$ does not dominate $(\mathcal{P},\mathcal{L}%
) $.
\end{lemma}

\begin{proof}
Assume otherwise that $Z\subseteq \mathcal{L}$ is a solution of $\mathcal{I}$%
. First suppose that there exists a line segment $L_{t}\in \mathcal{L}$ such
that $L_{t}\subseteq B_{l}$, where $l=\Gamma _{l_{j}}^{\text{vert}}\cap
\Gamma _{l_{j^{\prime }}}^{\text{diag}}$. Then, by Lemma~\ref%
{shadow-correctness-lem-1}, there must exist a line segment $L_{k}\in Z$
such that $L_{t}\cap S_{k}\neq \emptyset $ or $L_{k}\cap S_{t}\neq \emptyset 
$. If $L_{t}\cap S_{k}\neq \emptyset $ then $l_{k}\in \mathbb{R}_{\text{left}%
}^{2}(\Gamma _{l_{j^{\prime }}}^{\text{diag}})$, which is a contradiction to
the fact that $(j,j^{\prime })$ is the start-pair of $Z$. If $L_{k}\cap
S_{t}\neq \emptyset $ then $l_{k}\in \mathbb{R}_{\text{left}}^{2}(\Gamma
_{l_{j}}^{\text{vert}})$, which is again a contradiction to the fact that $%
(j,j^{\prime })$ is the start-pair of $Z$.

Now suppose that there exists a line segment $L_{t}\in \mathcal{L}$ such
that $L_{t}\subseteq A_{r}$, where $r=\Gamma _{r_{i}}^{\text{vert}}\cap
\Gamma _{r_{i^{\prime }}}^{\text{diag}}$. Then, by Lemma~\ref%
{shadow-correctness-lem-1}, there exists a line segment $L_{k}\in Z$ such
that $L_{t}\cap S_{k}\neq \emptyset $ or $L_{k}\cap S_{t}\neq \emptyset $.
If $L_{t}\cap S_{k}\neq \emptyset $ then $r_{k}\in \mathbb{R}_{\text{right}%
}^{2}(\Gamma _{r_{i}}^{\text{vert}})$, which is a contradiction to the fact
that $(i,i^{\prime })$ is the end-pair of $Z$. If $L_{k}\cap S_{t}\neq
\emptyset $ then $r_{k}\in \mathbb{R}_{\text{right}}^{2}(\Gamma
_{r_{i^{\prime }}}^{\text{diag}})$, which is again a contradiction to the
fact that $(i,i^{\prime })$ is the end-pair of $Z$.
\end{proof}

\begin{lemma}
\label{obs:irlvvrt1} Let $\mathcal{I}=(\mathcal{P},\mathcal{L},j,j^{\prime
},i,i^{\prime })$ be an instance of \textsc{Restricted Bounded Dominating Set%
} on tolerance graphs. Let $l=\Gamma _{l_{j}}^{\text{vert}}\cap \Gamma
_{l_{j^{\prime }}}^{\text{diag}}$ and $r=\Gamma _{r_{i}}^{\text{vert}}\cap
\Gamma _{r_{i^{\prime }}}^{\text{diag}}$. If there exists a line segment $%
L_{t}\in \mathcal{L}$ with one of its endpoints in $B_{l}\cup A_{r}$ and one
point (not necessarily an endpoint) in $\overline{B_{l}}\cap \overline{A_{r}}
$, then at least one of the line segments $\{L_{j},L_{j^{\prime
}},L_{i},L_{i^{\prime }}\}$ is a neighbor of $L_{t}$. Moreover, $L_{t}$ does
not belong to any optimum solution $Z$ of \textsc{Restricted Bounded
Dominating Set}.
\end{lemma}

\begin{proof}
Let $Z$ be an optimum solution of \textsc{Restricted Bounded Dominating Set}%
. Let $L_{t}\in \mathcal{L}$ be a line segment with one of its endpoints in $%
B_{l}\cup A_{r}$ and one point (not necessarily an endpoint) in $\overline{%
B_{l}}\cap \overline{A_{r}}$. Notice that $r_{t}\in A_{r}$ or $l_{t}\in
B_{l} $. Let first $r_{t}\in A_{r}$. Since $L_{t}$ has also a point in $%
\overline{B_{l}}\cap \overline{A_{r}}$, it follows that $L_{t}$ has a point
in $(S_{i}\cup F_{i})\cup (S_{i^{\prime }}\cup F_{i^{\prime }})$. Therefore $%
L_{t}$ is a neighbor of $L_{i}$ or $L_{i^{\prime }}$ by Lemma~\ref%
{shadow-correctness-lem-1}. Let now $l_{t}\in B_{l}$. Since $L_{t}$ has also
a point in $\overline{B_{l}}\cap \overline{A_{r}}$, it follows that $L_{t}$
has a point in $(S_{j}\cup F_{j})\cup (S_{j^{\prime }}\cup F_{j^{\prime }})$%
. Therefore $L_{t}$ is a neighbor of $L_{i}$ or $L_{i^{\prime }}$ by Lemma~%
\ref{shadow-correctness-lem-1}. Finally, since $r_{t}\in A_{r}$ or $l_{t}\in
B_{l}$, it follows that $r_{t}\in \mathbb{R}_{\text{right}}^{2}(\Gamma
_{r_{i}}^{\text{vert}})$ or $l_{t}\in \mathbb{R}_{\text{left}}^{2}(\Gamma
_{l_{j}}^{\text{vert}})$. Therefore $L_{t}\notin \mathcal{L}_{j,j^{\prime
}}^{\text{right}}$ or $L_{t}\notin \mathcal{L}_{i,i^{\prime }}^{\text{left}}$%
. Thus, since $Z\subseteq \mathcal{L\cap L}_{j,j^{\prime }}^{\text{right}%
}\cap \mathcal{L}_{i,i^{\prime }}^{\text{left}}$, it follows that $%
L_{t}\notin Z$.
\end{proof}

\begin{lemma}
\label{obs:irlvvrt2}Let $\mathcal{I}=(\mathcal{P},\mathcal{L},j,j^{\prime
},i,i^{\prime })$ be an instance of \textsc{Restricted Bounded Dominating Set%
} on tolerance graphs. Let $l=\Gamma _{l_{j}}^{\text{vert}}\cap \Gamma
_{l_{j^{\prime }}}^{\text{diag}}$ and $r=\Gamma _{r_{i}}^{\text{vert}}\cap
\Gamma _{r_{i^{\prime }}}^{\text{diag}}$. If there exists a line segment $%
L_{t}\in \mathcal{L}$ such that $L_{t}\subseteq \overline{B_{l}}\cap 
\overline{A_{r}}$ and $L_{t}\notin \mathcal{L}_{j,j^{\prime }}^{\text{right}%
}\cap \mathcal{L}_{i,i^{\prime }}^{\text{left}}$ then at least one of the
line segments $\{L_{j},L_{j^{\prime }},L_{i},L_{i^{\prime }}\}$ is a
neighbor of $L_{t}$. Moreover, $L_{t}$ does not belong to any optimum
solution $Z$ of \textsc{Restricted Bounded Dominating Set}.
\end{lemma}

\begin{proof}
Suppose first that $L_{t}\notin \mathcal{L}_{j,j^{\prime }}^{\text{right}}$.
Then $l_{t}\in \mathbb{R}_{\text{left}}^{2}(\Gamma _{l_{j}}^{\text{vert}})$
or $l_{t}\in \mathbb{R}_{\text{left}}^{2}(\Gamma _{l_{j^{\prime }}}^{\text{%
diag}})$. We first consider the case where $l_{t}\in \mathbb{R}_{\text{left}%
}^{2}(\Gamma _{l_{j}}^{\text{vert}})$. Then, since $l_{t}\in \overline{B_{l}}%
\cap \overline{A_{r}}$ by assumption, it follows that $l_{t}\in \mathbb{R}_{%
\text{right}}^{2}(\Gamma _{l_{i^{\prime }}}^{\text{diag}})$. This implies
that $l_{t}\in S_{j^{\prime }}$, and thus $L_{j^{\prime }}$ is a neighbor of 
$L_{t}$. We now consider the case where $l_{t}\in \mathbb{R}_{\text{left}%
}^{2}(\Gamma _{l_{j^{\prime }}}^{\text{diag}})$. Then, since $l_{t}\in 
\overline{B_{l}}\cap \overline{A_{r}}$ by assumption, it follows that $%
l_{t}\in \mathbb{R}_{\text{right}}^{2}(\Gamma _{l_{j}}^{\text{vert}})$. This
implies that $l_{t}\in F_{j}$, and thus $L_{j}$ is a neighbor of $L_{t}$.

The case where $L_{t}\notin \mathcal{L}_{i,i^{\prime }}^{\text{left}}$ can
be dealt with in exactly the same way, implying that, in this case, $L_{i}$
or $L_{i^{\prime }}$ is a neighbor of $L_{t}$.
\end{proof}

\medskip

From Lemmas~\ref{lem:badvrt1} and~\ref{lem:badvrt2} we define now the
notions of a \emph{bad point} $p\in \mathcal{P}$ and a \emph{bad line segment%
} $L_{t}\in \mathcal{L}$, respectively. Moreover, from Lemmas~\ref%
{irrelevant-points-lem},~\ref{obs:irlvvrt1}, and~\ref{obs:irlvvrt2} we
define the notions of an \emph{irrelevant point} $p\in \mathcal{P}$ and of
an \emph{irrelevant line segment} $L_{t}\in \mathcal{L}$, as follows.

\begin{definition}
\label{bad-irrelevant-def}Let $\mathcal{I}=(\mathcal{P},\mathcal{L}%
,j,j^{\prime },i,i^{\prime })$ be an instance of \textsc{Restricted Bounded
Dominating Set} on tolerance graphs. Let $l=\Gamma _{l_{j}}^{\text{vert}%
}\cap \Gamma _{l_{j^{\prime }}}^{\text{diag}}$ and $r=\Gamma _{r_{i}}^{\text{%
vert}}\cap \Gamma _{r_{i^{\prime }}}^{\text{diag}}$. A point $p\in \mathcal{P%
}$ is a \emph{bad point} if $p\in \mathbb{R}_{\text{left}}^{2}(\Gamma _{l}^{%
\text{diag}})$ or $p\in \mathbb{R}_{\text{right}}^{2}(\Gamma _{r}^{\text{vert%
}})$. A point $p\in \mathcal{P}$ is an \emph{irrelevant point} if $p\in
S_{l}\cup S_{r}$. A line segment $L_{t}\in \mathcal{L}$ is a \emph{bad line
segment} if $L_{t}\subseteq B_{l}$ or $L_{t}\subseteq A_{r}$. Finally a line
segment $L_{t}\in \mathcal{L}$ is an \emph{irrelevant line segment} if
either $L_{t}\subseteq \overline{B_{l}}\cap \overline{A_{r}}$ and $%
L_{t}\notin \mathcal{L}_{j,j^{\prime }}^{\text{right}}\cap \mathcal{L}%
_{i,i^{\prime }}^{\text{left}}$, or $L_{t}$ has an endpoint in $B_{l}\cup
A_{r}$ and another point in $\overline{B_{l}}\cap \overline{A_{r}}$.
\end{definition}

The next lemma will enable us to reduce \textsc{Restricted Bounded
Dominating Set} to \textsc{Bounded Dominating Set} on tolerance graphs,
cf.~Lemma~\ref{lem:sbssol2-first}.

\begin{lemma}
\label{lem:cosntrnwinstnc}Let $\mathcal{I}=(\mathcal{P},\mathcal{L}%
,j,j^{\prime },i,i^{\prime })$ be an instance of \textsc{Restricted Bounded
Dominating Set} on tolerance graphs, which has no bad or irrelevant points $%
p\in \mathcal{P}$ and no bad or irrelevant line segments $L\in \mathcal{L}$.
Then we can add a new line segment $L_{j,1}$ to the set $\mathcal{P}\cup 
\mathcal{L}$ such that $L_{j}$ is the only neighbor of $L_{j,1}$.
\end{lemma}

\begin{proof}
Since there are no bad or irrelevant points $p\in \mathcal{P}$ and no bad or
irrelevant line segments $L\in \mathcal{L}$ by assumption, there exists a
point $x\in \mathbb{R}^{2}$ such that, for every $p\in \mathcal{P}$ and for
every $L_{t}\in \mathcal{L}\setminus \{L_{j}\}$, we have that~${p,L_{t}\in 
\mathbb{R}_{\text{right}}^{2}(\Gamma _{x}^{\text{vert}})}$. That is, no
element of $\mathcal{P}\cup \left( \mathcal{L}\setminus \{L_{j}\}\right) $
has any point in the interior of the region $R_{1}=\mathbb{R}_{\text{right}%
}^{2}(\Gamma _{l_{j}}^{\text{vert}})\cap \mathbb{R}_{\text{left}}^{2}(\Gamma
_{x}^{\text{vert}})$. Furthermore we define the region $R_{1}^{\prime
}\subseteq R_{1}$, where $R_{1}^{\prime }=R_{1}\cap \mathbb{R}_{\text{left}%
}^{2}(\Gamma _{l_{j^{\prime}}}^{\text{diag}})$. This region $R_{1}^{\prime }$
is illustrated in Figure~\ref{construction-restricted-fig} for the case
where $j^{\prime }\neq j$; the case where $j^{\prime }=j$ is similar. Now we
add to $\mathcal{L}$ a new line segment $L_{j,1}$ arbitrarily within the
interior of the region $R_{1}^{\prime }$, cf.~Figure~\ref%
{construction-restricted-fig}. By the definition of $R_{1}^{\prime }$ it is
easy to verify that $L_{j,1}$ is adjacent only to~$L_{j}$.
\end{proof}

\begin{figure}[tbh]
\centering 
\includegraphics[scale=0.68]{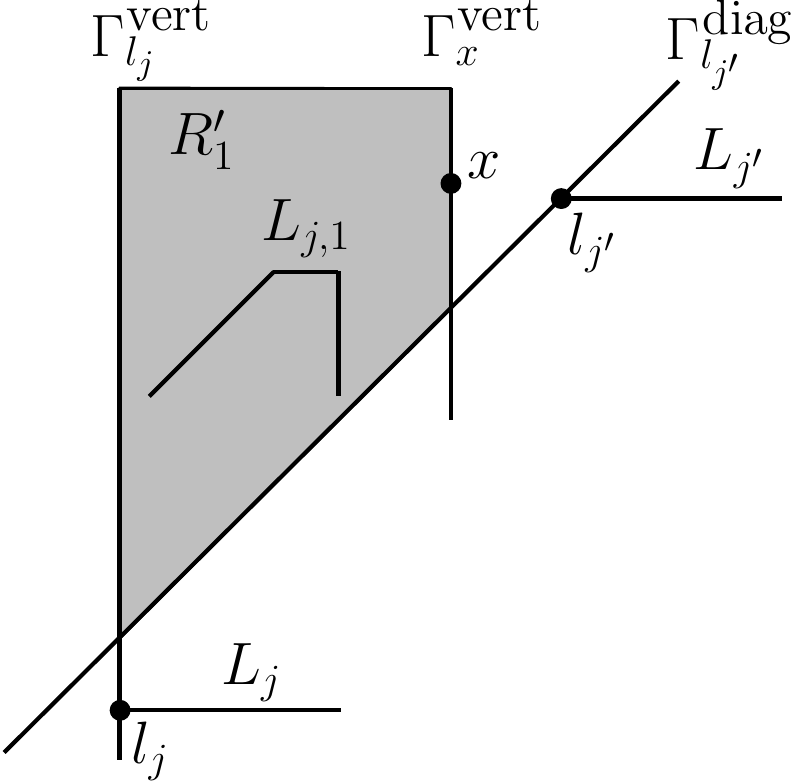}
\caption{The addition of the line segment $L_{j,1}$, in the case where $%
j^{\prime }\neq j$.}
\label{construction-restricted-fig}
\end{figure}

\medskip

In the following we denote by $l_{j,1}$ the left endpoint of the new line
segment $L_{j,1}$. Similarly to Definition~\ref{BD-def} in Section~\ref%
{bounded-alg-subsec}, we present in the next definition the quantity $RD_{(%
\mathcal{P},\mathcal{L})}(j,j^{\prime },i,i^{\prime })$ for the \textsc{%
Restricted Bounded Dominating Set} problem on tolerance graphs.

\begin{definition}
\label{RD-def}Let $\mathcal{I}=(\mathcal{P},\mathcal{L},j,j^{\prime
},i,i^{\prime })$ be an instance of \textsc{Restricted Bounded Dominating Set%
} on tolerance graphs. Then $RD_{(\mathcal{P},\mathcal{L})}(j,j^{\prime
},i,i^{\prime })$ is a dominating set ${Z\subseteq \mathcal{L\cap L}%
_{j,j^{\prime }}^{\text{right}}\cap \mathcal{L}_{i,i^{\prime }}^{\text{left}}%
}$ of $(\mathcal{P},\mathcal{L})$ with the smallest size, in which $%
(j,j^{\prime })$ and $(i,i^{\prime })$ are the start-pair and the end-pair,
respectively. If such a dominating set $Z$ does not exist, we define ${RD_{(%
\mathcal{P},\mathcal{L})}(j,j^{\prime },i,i^{\prime })=\bot }$ and $%
\left\vert {RD_{(\mathcal{P},\mathcal{L})}(j,j^{\prime },i,i^{\prime })}%
\right\vert {=\infty }$.
\end{definition}

\begin{observation}
\label{restricted-obs}${RD_{(\mathcal{P},\mathcal{L})}(j,j^{\prime
},i,i^{\prime })\neq \bot }$ if and only if $L_{j},L_{j^{\prime }}\in {%
\mathcal{L}_{i,i^{\prime }}^{\text{left}}}$, $L_{i},L_{i^{\prime }}\in {%
\mathcal{L}_{j,j^{\prime }}^{\text{right}}}$, and ${\mathcal{L\cap L}%
_{j,j^{\prime }}^{\text{right}}\cap \mathcal{L}_{i,i^{\prime }}^{\text{left}}%
}$ is a dominating set of $(\mathcal{P},\mathcal{L})$.
\end{observation}

For simplicity of the presentation we may refer to the set ${RD_{(\mathcal{P}%
,\mathcal{L})}(j,j^{\prime },i,i^{\prime })}$ as ${RD_{G}(j,j^{\prime
},i,i^{\prime })}$, where $(\mathcal{P},\mathcal{L})$ is the horizontal
shadow representation of the tolerance graph $G$. In the next lemma we
reduce the computation of ${RD_{(\mathcal{P},\mathcal{L})}(j,j^{\prime
},i,i^{\prime })}$ to the computation of an appropriate value for the
bounded dominating set problem (cf.~Section~\ref{Bounded-dominating-sec}).

\begin{lemma}
\label{lem:sbssol2-first}Let $\mathcal{I}=(\mathcal{P},\mathcal{L}%
,j,j^{\prime },i,i^{\prime })$ be an instance of \textsc{Restricted Bounded
Dominating Set} on tolerance graphs, which has no bad or irrelevant points $%
p\in \mathcal{P}$ and no bad or irrelevant line segments $L\in \mathcal{L}$.
Let $(\mathcal{P},\widehat{\mathcal{L}})$ be the augmented representation
that is obtained from $(\mathcal{P},\mathcal{L})$ by adding the line segment 
{$L_{j,1}$} as in Lemma~\ref{lem:cosntrnwinstnc}. Furthermore let $r=\Gamma
_{r_{i}}^{\text{vert}}\cap \Gamma _{r_{i^{\prime }}}^{\text{diag}}$. If ${%
RD_{(\mathcal{P},\mathcal{L})}(j,j^{\prime },i,i^{\prime })\neq \bot }$ then 
$RD_{(\mathcal{P},\mathcal{L})}(j,j^{\prime },i,i^{\prime })=BD_{(\mathcal{P}%
,\widehat{\mathcal{L}})}(l_{j,1},r,j^{\prime },i,i^{\prime })$.
\end{lemma}

\begin{proof}
Let $l=\Gamma _{l_{j}}^{\text{vert}}\cap \Gamma _{l_{j^{\prime }}}^{\text{%
diag}}$ and $r=\Gamma _{r_{i}}^{\text{vert}}\cap \Gamma _{r_{i^{\prime }}}^{%
\text{diag}}$. Then, since by assumption there are no bad or irrelevant
points $p\in \mathcal{P}$ or line segments $L\in \mathcal{L}$ in the
instance $\mathcal{I}=(\mathcal{P},\mathcal{L},j,j^{\prime },i,i^{\prime })$%
, it follows that all elements of $\mathcal{P}\cup \mathcal{L}$ are entirely
contained in the region $A_{l}\cap B_{r}$ of $\mathbb{R}^{2}$,
cf.~Definition~\ref{bad-irrelevant-def}. Therefore all elements of $\mathcal{%
P}\cup \mathcal{L}$ belong to the set $\{L_{i}\}\cup X(l,r)$, cf.~Eq.~(\ref%
{X(a,b)-def-eq}) in Section~\ref{bounded-alg-subsec}. Now recall from the
construction of the augmented representation $(\mathcal{P},\widehat{\mathcal{%
L}})$ from $(\mathcal{P},\mathcal{L})$ in the proof of Lemma~\ref%
{lem:cosntrnwinstnc} that $L_{j,1}$ is the only element of $\mathcal{P}\cup 
\widehat{\mathcal{L}}$ that does not belong to the set $\{L_{i}\}\cup X(l,r)$%
, cf.~Figure~\ref{construction-restricted-fig}. Furthermore, it is easy to
check that the set of elements of $\mathcal{P}\cup \widehat{\mathcal{L}}$ is
exactly the set $\{L_{i}\}\cup X(l_{j,1},r)$.

Since ${RD_{(\mathcal{P},\mathcal{L})}(j,j^{\prime },i,i^{\prime })\neq \bot 
}$ by assumption, it follows by Observation~\ref{restricted-obs} that $%
L_{j},L_{j^{\prime }}\in {\mathcal{L}_{i,i^{\prime }}^{\text{left}}}$ and $%
L_{i},L_{i^{\prime }}\in {\mathcal{L}_{j,j^{\prime }}^{\text{right}}}$ as
well as that ${\mathcal{L\cap L}_{j,j^{\prime }}^{\text{right}}\cap \mathcal{%
L}_{i,i^{\prime }}^{\text{left}}}$ is a dominating set of $(\mathcal{P},%
\mathcal{L})$. Furthermore, since $L_{j}$ is the only neighbor of $L_{j,1}$
in the augmented representation $(\mathcal{P},\widehat{\mathcal{L}})$, it
follows that ${\mathcal{L\cap L}_{j,j^{\prime }}^{\text{right}}\cap \mathcal{%
L}_{i,i^{\prime }}^{\text{left}}}$ is also a dominating set of $(\mathcal{P},%
\widehat{\mathcal{L}})$. Moreover, since ${\mathcal{L}_{j,j^{\prime }}^{%
\text{right}}\subseteq \mathcal{L}_{j^{\prime }}^{\text{right}}}$
(cf.~Definition~\ref{left-right-crossing-pair-def} in Section~\ref%
{terminology-bounded-domination-subsec}), it follows that also ${\mathcal{%
L\cap L}_{j^{\prime }}^{\text{right}}\cap \mathcal{L}_{i,i^{\prime }}^{\text{%
left}}}$ is a dominating set of $(\mathcal{P},\widehat{\mathcal{L}})$.
Therefore ${BD_{(\mathcal{P},\widehat{\mathcal{L}})}(l_{j,1},r,j^{\prime
},i,i^{\prime })\neq \bot }$ by Observation~\ref{obs:botcs}. That is, $BD_{(%
\mathcal{P},\widehat{\mathcal{L}})}(l_{j,1},r,j^{\prime },i,i^{\prime })$ is
a dominating set $Z\subseteq \widehat{\mathcal{L}}$ of $X(l_{j,1},r)$ with
the smallest size, in which $(i,i^{\prime })$ is its end-pair and $%
L_{j^{\prime }}$ is its diagonally leftmost line segment (cf.~Definition~\ref%
{BD-def} in Section~\ref{bounded-alg-subsec}). Since $L_{j^{\prime }}$ is
the diagonally leftmost line segment of $BD_{(\mathcal{P},\widehat{\mathcal{L%
}})}(l_{j,1},r,j^{\prime }i,i^{\prime })$, it follows that $L_{j,1}\notin
BD_{(\mathcal{P},\widehat{\mathcal{L}})}(l_{j,1},r,j^{\prime },i,i^{\prime
}) $. Therefore $L_{j}\in BD_{(\mathcal{P},\widehat{\mathcal{L}}%
)}(l_{j,1},r,j^{\prime },i,i^{\prime })$, since $L_{j}$ is the only neighbor
of $L_{j,1}$ in $(\mathcal{P},\widehat{\mathcal{L}})$. Thus $(j,j^{\prime })$
is the start-pair of $BD_{(\mathcal{P},\widehat{\mathcal{L}}%
)}(l_{j,1},r,j^{\prime },i,i^{\prime })$. Finally, since also $\mathcal{P}%
\cup \widehat{\mathcal{L}}=\{L_{i}\}\cup X(l_{j,1},r)$ as we proved above,
it follows that $RD_{(\mathcal{P},\mathcal{L})}(j,j^{\prime },i,i^{\prime
})=BD_{(\mathcal{P},\widehat{\mathcal{L}})}(l_{j,1},r,j^{\prime
},i,i^{\prime })$.
\end{proof}

\medskip

We are now ready to present Algorithm~\ref{restricted-bounded-alg} which,
given an instance $\mathcal{I}=(\mathcal{P},\mathcal{L},j,j^{\prime
},i,i^{\prime })$ of \textsc{Restricted Bounded Dominating Set} on tolerance
graphs, either outputs a set $Z\subseteq \mathcal{L\cap L}_{j,j^{\prime }}^{%
\text{right}}\cap \mathcal{L}_{i,i^{\prime }}^{\text{left}}$ of minimum size
that dominates all elements of $(\mathcal{P},\mathcal{L})$, or it announces
that such a set $Z$ does not exist. Algorithm~\ref{restricted-bounded-alg}
uses Algorithm~\ref{bounded-dominating-tolerance-alg} (which solves \textsc{%
Bounded Dominating Set} on tolerance graphs, cf.~Section~\ref%
{Bounded-dominating-sec}) as a subroutine.

\begin{algorithm}[t!]
\caption{\textsc{Restricted Bounded Dominating Set} on Tolerance Graphs} 
\label{restricted-bounded-alg}
\begin{algorithmic} [1]
\REQUIRE{A 6-tuple ${\cal I}=(\PP, \LL,j,j',i,i')$, where $(\PP,\LL)$ is a horizontal shadow representation of a tolerance graph $G$, 
$(j,j')$ is a left-crossing pair and $(i,i')$ is a right-crossing pair of $(\PP, \LL)$.}
\ENSURE{A set $Z\subseteq \mathcal{L}$ of minimum size that dominates $(\mathcal{P},\mathcal{L})$, where $(j,j')$ is the start-pair~and $(i,i')$ is the end-pair of $Z$, 
or the value $\bot$.}%

\medskip

\IF{$(\PP,\LL)$ contains a bad point $p\in \PP$ or a bad line segment $L_{k}\in \LL$ \ (cf.~Definition~\ref{bad-irrelevant-def})} \label{alg-2-line-1}
\vspace{0,1cm}
     \RETURN{$\bot$} \label{alg-2-line-2}
\ENDIF 

\vspace{0,1cm}

\IF{$L_{j},L_{j^{\prime }}\in {\mathcal{L}_{i,i^{\prime }}^{\text{left}}}$, \ $L_{i},L_{i^{\prime }}\in {\mathcal{L}_{j,j^{\prime }}^{\text{right}}}$, 
and ${\mathcal{L\cap L}_{j,j^{\prime }}^{\text{right}}\cap \mathcal{L}_{i,i^{\prime }}^{\text{left}}}$ is a dominating set of $(\mathcal{P},\mathcal{L})$} \label{alg-2-line-3}
     
     \vspace{0,1cm}
     
     \STATE{Compute the sets $\PP_{1}\subseteq \PP$ and $\LL_{1}\subseteq \LL$ of irrelevant points and line segments \ (cf.~Definition~\ref{bad-irrelevant-def})}\label{alg-2-line-4}%
     
     \STATE{$\mathcal{P} \leftarrow \mathcal{P} \setminus \PP_{1}$; \ \ $\mathcal{L} \leftarrow \mathcal{L} \setminus \LL_{1}$; \ \ $r \leftarrow \Gamma_{r_{i}}^{\text{vert}} \cap \Gamma_{r_{i^{\prime}}}^{\text{diag}}$} \label{alg-2-line-5}

     \vspace{0,1cm}
     
     \STATE{$\widehat{\mathcal{L}} \leftarrow \mathcal{L} \cup \{L_{j,1}\}$ \ (cf.~Lemma~\ref{lem:cosntrnwinstnc})} \label{alg-2-line-6}
     
     \vspace{0,1cm}
     
     \RETURN{$BD_{(\mathcal{P},\widehat{\mathcal{L}})}(l_{j,1},r,j',i,i^{\prime})$} \COMMENT{by calling Algorithm~\ref{bounded-dominating-tolerance-alg}} \label{alg-2-line-7}
    
\ENDIF

\vspace{0,1cm}

\STATE{\textbf{else \ return} $\bot$}\label{alg-2-line-8}
\end{algorithmic}
\end{algorithm}

\begin{theorem}
\label{restricted-correctness-thm}Given a $6$-tuple $\mathcal{I}=(\mathcal{P}%
,\mathcal{L},j,j^{\prime },i,i^{\prime })$, where $(\mathcal{P},\mathcal{L})$
is a horizontal shadow representation of a tolerance graph $G$ with $n$
vertices, $(j,j^{\prime })$ is a left-crossing pair and $(i,i^{\prime })$ is
a right-crossing pair of $(\mathcal{P},\mathcal{L})$, Algorithm~\ref%
{restricted-bounded-alg} computes \textsc{Restricted Bounded Dominating Set}
in $O(n^{9})$ time.
\end{theorem}

\begin{proof}
If the horizontal shadow representation $(\mathcal{P},\mathcal{L)}$ contains
at least one bad point $p\in \mathcal{P}$ or at least one bad line segment $%
L_{k}\in \mathcal{L}$ (cf.~Definition~\ref{bad-irrelevant-def}) then $%
\mathcal{L\cap L}_{j,j^{\prime }}^{\text{right}}\cap \mathcal{L}%
_{i,i^{\prime }}^{\text{left}}$ does not dominate $(\mathcal{P},\mathcal{L})$
by Lemmas~\ref{lem:badvrt1} and~\ref{lem:badvrt2}. Thus, in the case where
such a bad point or bad line segment exists in $(\mathcal{P},\mathcal{L)}$,
Algorithm~\ref{restricted-bounded-alg} correctly returns ${\bot }$,
cf.~lines~\ref{alg-2-line-1}-\ref{alg-2-line-2}. Furthermore, due to
Observation~\ref{restricted-obs}, the algorithm correctly returns ${\bot }$
in~line~\ref{alg-2-line-8} if at least one of the conditions checked in line~%
\ref{alg-2-line-3} is not satisfied.

Assume now that all conditions that are checked in line~\ref{alg-2-line-3}
are satisfied. Then ${RD_{(\mathcal{P},\mathcal{L})}(j,j^{\prime
},i,i^{\prime })\neq \bot }$ by Observation~\ref{restricted-obs}. Let ${%
\mathcal{P}_{1}\subseteq \mathcal{P}}$ and ${\mathcal{L}_{1}\subseteq 
\mathcal{L}}$ be the set of all irrelevant points and line segments,
respectively (cf.~Definition~\ref{bad-irrelevant-def}). Then, by Lemmas~\ref%
{irrelevant-points-lem},~\ref{obs:irlvvrt1}, and~\ref{obs:irlvvrt2}, every
point $p\in \mathcal{P}_{1}$ and every line segment $L_{t}\in \mathcal{L}%
_{1} $ is dominated by at least one of the line segments $%
\{L_{j},L_{j^{\prime }},L_{i},L_{i^{\prime }}\}$. Furthermore, by Lemmas~\ref%
{obs:irlvvrt1} and~\ref{obs:irlvvrt2}, no line segment $L_{t}\in \mathcal{L}%
_{1}$ is contained in any optimum solution $Z$ of \textsc{Restricted Bounded
Dominating Set}. Thus Algorithm~\ref{restricted-bounded-alg} correctly
removes the sets $\mathcal{P}_{1}$ and $\mathcal{L}_{1}$ of the irrelevant
points and line segments from the instance, cf.~lines~\ref{alg-2-line-4}-\ref%
{alg-2-line-5} of the algorithm.

In line~\ref{alg-2-line-6} the algorithm augments the set $\mathcal{L}$ of
line segments to the set $\widehat{\mathcal{L}}$ by adding to it the line
segment $L_{j,1}$ as in Lemma~\ref{lem:cosntrnwinstnc}. Then the algorithm
returns in line~\ref{alg-2-line-7} the value $BD_{(\mathcal{P},\widehat{%
\mathcal{L}})}(l_{j,1},r,j^{\prime },i,i^{\prime })$ by calling Algorithm~%
\ref{bounded-dominating-tolerance-alg} as a subroutine (cf.~Section~\ref%
{Bounded-dominating-sec}). The correctness of this computation in line~\ref%
{alg-2-line-7} follows immediately by Lemma~\ref{lem:sbssol2-first}.

Regarding the running time of Algorithm~\ref{restricted-bounded-alg}, note
by Definition~\ref{bad-irrelevant-def} that we can check in constant time
whether a given point $p\in \mathcal{P}$ (resp.~a given line segment $%
L_{t}\in \mathcal{L}$) is bad or irrelevant. Therefore each of the lines~\ref%
{alg-2-line-1},~\ref{alg-2-line-2}, and~\ref{alg-2-line-4} of the algorithm
can be executed in $O(n)$ time. The execution time of the if-statement of
line~\ref{alg-2-line-3} is dominated by the $O(n^{2})$ time that is needed to check
whether $\mathcal{L\cap L}_{j,j^{\prime }}^{\text{right}}\cap \mathcal{L}%
_{i,i^{\prime }}^{\text{left}}$ is a dominating set of $(\mathcal{P},%
\mathcal{L})$. Furthermore lines~\ref{alg-2-line-5}-\ref{alg-2-line-6} can
be executed trivially in total $O(n)$ time. Finally, line~\ref{alg-2-line-7}
can be executed in $O(n^{9})$ time by Theorem~\ref{bounded-correctness-thm},
and thus the total running time of Algorithm~\ref{restricted-bounded-alg} is 
$O(n^{9})$.
\end{proof}

\section{Dominating set on tolerance graphs\label{tolerance-domination-sec}}

In this section we present our main algorithm of the paper (cf.~Algorithm~%
\ref{dominating-tol-alg}) which computes in polynomial time a minimum
dominating set of a tolerance graph $G$, given by a horizontal shadow
representation $(\mathcal{P},\mathcal{L})$. Algorithm~\ref%
{dominating-tol-alg} uses as subroutines Algorithms~\ref%
{bounded-dominating-tolerance-alg} and~\ref{restricted-bounded-alg}, which
solve \textsc{Bounded Dominating Set} and \textsc{Restricted Bounded
Dominating Set} on tolerance graphs, respectively (cf.~Sections~\ref%
{Bounded-dominating-sec} and~\ref{Restricted-domination-sec}). Throughout
this section we assume without loss of generality that the given tolerance
graph $G$ is connected and that $G$ is given with a \emph{canonical}
horizontal shadow representation $(\mathcal{P},\mathcal{L})$. It is
important to note here that, in contrast to Algorithms~\ref%
{bounded-dominating-tolerance-alg} and~\ref{restricted-bounded-alg}, the
minimum dominating set $D$ that is computed by Algorithm~\ref%
{dominating-tol-alg} can also contain unbounded vertices. Thus always $D\neq
\bot $, since in the worst case $D$ contains the whole set $\mathcal{P}\cup 
\mathcal{L}$.

For every $p\in \mathcal{P}$ we denote by $N(p)=\{L_{k}\in \mathcal{L}:p\in
S_{k}\}$ and $H(p)=\{x\in \mathcal{P}\cup \mathcal{L}:x\cap S_{p}\neq
\emptyset \}$. Note that, due to Lemmas~\ref{shadow-correctness-lem-2} and~%
\ref{shadow-hovering-lem}, $N(p)$ is the set of neighbors of $p$ and $H(p)$
is the set of hovering vertices of $p$. Furthermore, for every $L_{k}\in 
\mathcal{L}$ we denote by $N(L_{k})=\{p\in \mathcal{P}:p\in S_{k}\}\cup
\{L_{t}\in \mathcal{L}:L_{t}\cap S_{k}\neq \emptyset $ or $L_{k}\cap
S_{t}\neq \emptyset \}$. Note that, due to Lemmas~\ref%
{shadow-correctness-lem-1} and~\ref{shadow-correctness-lem-2}, $N(L_{k})$ is
the set of neighbors of~$L_{k}$.

\begin{observation}
\label{neighbors-hovering-obs}Let $(\mathcal{P},\mathcal{L})$ be a canonical
representation of a connected tolerance graph $G$, and let $p\in \mathcal{P}$%
. Then $N(p)\subseteq N(x)$ for every $x\in H(p)$ by Lemma~\ref%
{neighbors-hovering}. Furthermore $H(p)\cap \mathcal{L}\neq \emptyset $ by
Lemma~\ref{hovering-bounded}.
\end{observation}

\begin{lemma}
\label{lem:nhvornigb}Let $(\mathcal{P},\mathcal{L})$ be a canonical
horizontal shadow representation of a connected tolerance graph $G$ and let $%
D$ be a minimum dominating set of $(\mathcal{P},\mathcal{L})$. If there
exists a point $p\in \mathcal{P}$ such that $p\in D$ and $(N(p)\cup
H(p))\cap D\neq \emptyset $, then there exists a dominating set $D^{\prime }$
of $(\mathcal{P},\mathcal{L})$ such that $|D^{\prime }|=|D|$ and $|D^{\prime
}\cap \mathcal{P}|=|D\cap \mathcal{P}|-1$.
\end{lemma}

\begin{proof}
We may assume without loss of generality that $\mathcal{P}\neq \emptyset$
and $\mathcal{L}\neq \emptyset$. Indeed, if $\mathcal{P}= \emptyset$ then we
can just solve the problem \textsc{Bounded Dominating Set} (see Section~\ref%
{Bounded-dominating-sec}); furthermore, if $\mathcal{L}= \emptyset$, then
the graph $G$ is an independent set. Consider a point $p\in \mathcal{P}$
such that $p\in D$. Suppose first that $x\in D$ for some $x\in N(p)$, i.e., $%
N(p)\cap D\neq \emptyset $. Recall by Observation~\ref%
{neighbors-hovering-obs} that $H(p)\cap \mathcal{L}\neq \emptyset $ and
consider a line segment $L_{k}\in H(p)\cap \mathcal{L}$. We will prove that
the set $D^{\prime }=(D\setminus \{p\})\cup \{L_{k}\}$ is a minimum
dominating set of $G$. First note that $p$ is dominated by $x\in D\setminus
\{p\}\subseteq D^{\prime }$. Furthermore $N(p)\subseteq N(L_{k})$ by
Observation~\ref{neighbors-hovering-obs}, since $L_{k}\in H(p)$. This
implies that $N(p)$ is dominated by $L_{k}$ in $D^{\prime }$. Thus, since $%
|D^{\prime }|=|D|$, it follows that $D^{\prime }$ is a minimum dominating
set of $G$.

Suppose now that $x\in D$ for some $x\in H(p)$, i.e., $H(p)\cap D\neq
\emptyset $. Since $G$ is assumed to be connected, it follows that $N(p)\neq
\emptyset $. Let $L_{k}\in N(p)$. We will prove that the set $D^{\prime
}=(D\setminus \{p\})\cup \{L_{k}\}$ is a minimum dominating set of $G$.
First note that $p$ is dominated by $L_{k}\in D^{\prime }$. Recall by
Observation~\ref{neighbors-hovering-obs} that $N(p)\subseteq N(x)$. This
implies that $N(p)$ is dominated by $x$ in $D^{\prime }$. Thus, since $%
|D^{\prime }|=|D|$, it follows that $D^{\prime }$ is a minimum dominating
set of $G$.

To finish the proof of the lemma, note that $|D^{\prime }\cap \mathcal{P}%
|=|D\cap \mathcal{P}|-1$ follows from the construction of $D^{\prime }$, as
we always replace in $D^{\prime }$ the point $p\in \mathcal{P}$ by a line
segment $L_{k}\in \mathcal{L}$.
\end{proof}

\medskip

Define now the subset $\mathcal{P}^{\ast }\subseteq \mathcal{P}$ of points
as follows:%
\begin{equation}
\mathcal{P}^{\ast }=\{p\in \mathcal{P}:p\notin H(p^{\prime })\text{ for
every point }p^{\prime }\in \mathcal{P}\setminus \{p\}\}.  \label{P-ast-eq}
\end{equation}%
Equivalently, $\mathcal{P}^{\ast }$ contains all points $p\in \mathcal{P}$
such that $p\notin S_{p^{\prime }}$ for every other point $p^{\prime }\in 
\mathcal{P}\setminus \{p\}$. Note by the definition of the set $\mathcal{P}%
^{\ast }$ that for every $p_{1},p_{2}\in \mathcal{P}^{\ast }$ we have $%
p_{1}\notin S_{p_{2}}\cup F_{p_{2}}$. Furthermore recall that the points of $%
\mathcal{P}=\{p_{1},p_{2},\ldots ,p_{|\mathcal{P}|}\}$ have been assumed to
be ordered increasingly with respect to their $x$-coordinates. Therefore,
since $\mathcal{P}^{\ast }\subseteq \mathcal{P}$, the points of $\mathcal{P}%
^{\ast }$ are also ordered increasingly with respect to their $x$%
-coordinates.

\begin{definition}
\label{normal-dominating-def}Let $(\mathcal{P},\mathcal{L})$ be a horizontal
shadow representation. A dominating set $D$ of $(\mathcal{P},\mathcal{L})$
is \emph{normalized} if:

\begin{enumerate}
\item \label{lem:itm1}$(N(p)\cup H(p))\cap D=\emptyset $ whenever $p\in
D\cap \mathcal{P}$, and

\item \label{lem:itm2 copy(1)}$D\cap \mathcal{P}\subseteq \mathcal{P}^{\ast
} $.
\end{enumerate}
\end{definition}

\begin{lemma}
\label{lem:dswrlk4}Let $(\mathcal{P},\mathcal{L})$ be a canonical horizontal
shadow representation of a connected tolerance graph $G$. Then there exists
a minimum dominating set $D$ of $(\mathcal{P},\mathcal{L})$ that is
normalized.
\end{lemma}

\begin{proof}
Let $D$ be a minimum dominating set of $G$ that contains the smallest
possible number of points from the set $\mathcal{P}$. That is, $|D\cap 
\mathcal{P}|\leq |D^{\prime }\cap \mathcal{P}|$ for every minimum dominating
set $D^{\prime }$ of $G$. Let $p\in D\cap \mathcal{P}$.

First assume that $(N(p)\cup H(p))\cap D\neq \emptyset $. Then Lemma~\ref%
{lem:nhvornigb} implies that there exists another minimum dominating set $%
D^{\prime }$ of $G$ such that $|D^{\prime }\cap \mathcal{P}|=|D\cap \mathcal{%
P}|-1<|D\cap \mathcal{P}|$, which is a contradiction to the choice of $D$.
Therefore $(N(p)\cup H(p))\cap D=\emptyset $ for every $p\in D\cap \mathcal{P%
}$.

Now assume that $p\in (\mathcal{P}\setminus \mathcal{P}^{\ast })\cap D$.
Then, by the definition of the set $\mathcal{P}^{\ast }$, there exists a
point $p^{\prime }\in \mathcal{P}$ such that $p\in H(p^{\prime })$. Note by
Observation~\ref{neighbors-hovering-obs} that $N(p^{\prime })\subseteq N(p)$%
. Suppose that $p^{\prime }\in D$. Then, since $p\in H(p^{\prime })$, Lemma~%
\ref{lem:nhvornigb} implies that there exists a minimum dominating set $%
D^{\prime }$ such that $|D^{\prime }\cap \mathcal{P}|=|D\cap \mathcal{P}%
|-1<|D\cap \mathcal{P}|$, which is a contradiction to the choice of $D$.
Therefore $p^{\prime }\notin D$. Thus, since $D$ is a dominating set of $G$
and $p^{\prime }\notin D$, there must exist an $L_{k}\in N(p^{\prime })$
such that $L_{k}\in D$. Therefore, since $N(p^{\prime })\subseteq N(p)$, it
follows that $L_{k}\in N(p)\cap D$. Then Lemma~\ref{lem:nhvornigb} implies
that there exists a minimum dominating set $D^{\prime }$ of $G$ such that $%
|D^{\prime }\cap \mathcal{P}|=|D\cap \mathcal{P}|-1<|D\cap \mathcal{P}|$,
which is again a contradiction to the choice of $D$. This implies that $(%
\mathcal{P}\setminus \mathcal{P}^{\ast })\cap D=\emptyset $ and therefore $%
D\cap \mathcal{P}\subseteq \mathcal{P}^{\ast }$. Thus the dominating set $D$
is normalized.
\end{proof}

\medskip

In the remainder of this section, whenever we refer to a minimum dominating
set $D$ of a connected tolerance graph $G$ that is given by a canonical
horizontal shadow representation $(\mathcal{P},\mathcal{L})$, we will always
assume (due to Lemma~\ref{lem:dswrlk4}) that $D$ is \emph{normalized}.
Moreover, given such a canonical horizontal shadow representation $(\mathcal{%
P},\mathcal{L})$, where $\mathcal{P}=\{p_{1},p_{2},\ldots ,p_{|\mathcal{P}%
|}\}$ and $\mathcal{L}=\{L_{1},L_{2},\ldots ,L_{|\mathcal{L}|}\}$, we add
two dummy line segments $L_{0}$ and $L_{|\mathcal{L}|+1}$ (with endpoints $%
l_{0},r_{0}$ and $l_{|\mathcal{L}|+1},r_{|\mathcal{L}|+1}$, respectively)
such that all elements of $\mathcal{P\cup L}$ are contained in $A_{r_{0}}$
and in $B_{l_{|\mathcal{L}|+1}}$. Denote $\mathcal{L}^{\prime }=\mathcal{L}%
\cup \{L_{0},L_{|\mathcal{L}|+1}\}$. Furthermore we add one dummy point $p_{|%
\mathcal{P}|+1}$ such that all elements of $\mathcal{P\cup L}^{\prime }$ are
contained in $B_{p_{|\mathcal{P}|+1}}$. Denote $\mathcal{P}^{\prime }=%
\mathcal{P}\cup \{p_{|\mathcal{P}|+1}\}$.

Note that $(\mathcal{P}^{\prime },\mathcal{L}^{\prime })$ is a horizontal
shadow representation of some tolerance graph $G^{\prime }$, where the
bounded vertices $V_{B}^{\prime }$ of $G^{\prime }$ correspond to the line
segments of $\mathcal{L}^{\prime }$ and the unbounded vertices $%
V_{U}^{\prime }$ of $G^{\prime }$ correspond to the points of $\mathcal{P}%
^{\prime }$. Furthermore note that, although $G$ is connected, $G^{\prime }$
is not connected as it contains the three isolated vertices that correspond
to $L_{0}$, $L_{|\mathcal{L}|+1}$, and $p_{|\mathcal{P}|+1}$. However, since
there exists by Lemma~\ref{lem:dswrlk4} a minimum dominating set $D$ of $G$
that is normalized, it is easy to verify that $G^{\prime }$ also admits a
normalized minimum dominating set. Therefore, whenever we refer to a minimum
dominating set $D^{\prime }$ of the augmented tolerance graph $G^{\prime }$,
we will always assume that $D^{\prime }$ is normalized.

For simplicity of the presentation, we refer in the following to the
augmented sets $\mathcal{P}^{\prime }$ and $\mathcal{L}^{\prime }$ of points
and horizontal line segments by $\mathcal{P}$ and $\mathcal{L}$,
respectively. In the remainder of this section we will write $\mathcal{P}%
=\{p_{1},p_{2},\ldots ,p_{|\mathcal{P}|}\}$ and $\mathcal{L}%
=\{L_{1},L_{2},\ldots ,L_{|\mathcal{L}|}\}$ with the understanding that the
last point $p_{|\mathcal{P}|}$ of $\mathcal{P}$, as well as the first and
the last line segments $L_{1}$ and $L_{|\mathcal{L}|}$ of $\mathcal{L}$, are
dummy. Note that the last point $p_{|\mathcal{P}|}$ (i.e., the new dummy
point) belongs to the set $\mathcal{P}^{\ast }$. Furthermore, we will refer
to the augmented tolerance graph $G^{\prime }$ by $G$. For every $%
p_{i},p_{j}\in \mathcal{P}^{\ast }$ with $i<j$, we denote by%
\begin{eqnarray}
G_{j} &=&\{x\in \mathcal{P}\cup \mathcal{L}:x\subseteq B_{p_{j}}\setminus
\Gamma _{p_{j}}^{\text{vert}}\},  \label{G_j-eq1} \\
G(i,j) &=&\{x\in G_{j}:x\subseteq A_{p_{i}}\}.  \label{G-i-j-eq2}
\end{eqnarray}%
that is, $G_{j}$ is set of elements of $\mathcal{P}\cup \mathcal{L}$ that
are entirely contained in the region $B_{p_{j}}\setminus \Gamma _{p_{j}}^{%
\text{vert}}$, and $G(i,j)$ is the subset of $G_{j}$ that contains the
elements of $\mathcal{P}\cup \mathcal{L}$ that are entirely contained in the
region $A_{p_{i}}$. Note that $p_{j}\notin G_{j}$ and $p_{j}\notin G(i,j)$.

\begin{definition}
\label{D-general-def}Let $p_{j}\in \mathcal{P}^{\ast }$ and $(i,i^{\prime })$
be a right-crossing pair in $G_{j}$. Then $D(j,i,i^{\prime })$ is a \emph{%
minimum normalized dominating set} of $G_{j}$ whose end-pair is $%
(i,i^{\prime })$. If there exists no dominating set $Z$ of $G_{j}$ whose
end-pair is $(i,i^{\prime })$, we define $D(j,i,i^{\prime })=\bot $.
\end{definition}

\begin{observation}
\label{general-infeasible-characterization-obs}$D(j,i,i^{\prime })\neq \bot $
if and only if $\mathcal{L}_{i,i^{\prime }}^{\text{left}}$ is a dominating
set of $G_{j}$.
\end{observation}

\begin{observation}
\label{infeasible-dominating-obs}If $X(r_{i^{\prime }},p_{j})$ is not
dominated by the set $\{L_{i},L_{i^{\prime }}\}$ then $D(j,i,i^{\prime
})=\bot $. Furthermore, if there exists a point $p\in \mathcal{P}\cap G_{j}$
such that $p\in \mathbb{R}_{\text{right}}^{2}(\Gamma _{r_{i}}^{\text{vert}})$
then $D(j,i,i^{\prime })=\bot $.
\end{observation}

Due to Observation~\ref{general-infeasible-characterization-obs}, without
loss of generality we assume below (in Lemmas~\ref%
{domination-correctness-lem-1-first-direction} and~\ref%
{domination-correctness-lem}) that $D(j,i,i^{\prime })\neq \bot $. Before we
provide our recursive computation for $D(j,i,i^{\prime })$ in Lemma~\ref%
{domination-correctness-lem} (cf.~Eq.~(\ref{recursion-domination-eq})), we
first prove in the next lemma that the upper part of the right hand side of
Eq.~(\ref{recursion-domination-eq}) is indeed a normalized dominating set of 
$G_{j}$, in which $(i,i^{\prime })$ is its end-pair.

\begin{lemma}
\label{domination-correctness-lem-1-first-direction}Let $G$ be a tolerance
graph, $(\mathcal{P},\mathcal{L})$ be a canonical representation of $G$, $%
p_{j}\in \mathcal{P}^{\ast }$, and $(i,i^{\prime })$ be a a right-crossing
pair of $G_{j}$. Assume that $D(j,i,i^{\prime })\neq \bot $. Let $%
q,q^{\prime },z,z^{\prime },w,w^{\prime }$ such that:

\begin{enumerate}
\item \label{domination-correctness-condition-1}$p_{q^{\prime }}\in \mathcal{%
P}^{\ast }$, where $1\leq q^{\prime }<j$,

\item \label{domination-correctness-condition-2}$L_{i},L_{i^{\prime }}\notin
N(p_{q^{\prime }})\cup H(p_{q^{\prime }})$,

\item \label{domination-correctness-condition-3}$(w,w^{\prime })$ is a
left-crossing pair of $G(q^{\prime },j)$,

\item \label{domination-correctness-condition-4}$(z,z^{\prime })$ is a
right-crossing pair of $G_{q^{\prime }}$,

\item \label{domination-correctness-condition-5}$q=\min \{1\leq k\leq
q^{\prime }:p_{k}\in \mathcal{P}^{\ast },\ p_{k}\in A_{\zeta }\}$, where $%
\zeta =\Gamma _{r_{z}}^{\text{vert}}\cap \Gamma _{r_{z^{\prime }}}^{\text{%
diag}}$,

\item \label{domination-correctness-condition-6}${\left( H(p_{q})\cup
H(p_{q^{\prime }})\right) \setminus \left( \bigcup_{q\leq k\leq q^{\prime
}}N(p_{k})\right) }$ are dominated by the line segments $\{L_{z},L_{z^{%
\prime }},L_{w},L_{w^{\prime }}\}$,

\item \label{domination-correctness-condition-7}$G(q,q^{\prime })$ is
dominated by $\{p_{k}\in \mathcal{P}^{\ast }:q\leq k\leq q^{\prime }\}$.
\end{enumerate}

If $D(q,z,z^{\prime })\neq \bot $ and $RD_{G(q^{\prime },j)}(w,w^{\prime
},i,i^{\prime })\neq \bot $ then the set%
\begin{equation*}
D(q,z,z^{\prime })\cup \left\{ p_{k}\in \mathcal{P}^{\ast }:q\leq k\leq
q^{\prime }\right\} \cup RD_{G(q^{\prime },j)}(w,w^{\prime },i,i^{\prime })
\end{equation*}
is a normalized dominating set of $G_{j}$, in which $(i,i^{\prime })$ is its
end-pair.
\end{lemma}

\begin{proof}
The choices of $q,q^{\prime },z,z^{\prime },w,w^{\prime },i,i^{\prime }$, as
described in the assumptions of the lemma, are illustrated in Figure~\ref%
{general-recursion-fig}. Assume that $D(q,z,z^{\prime })\neq \bot $ and that 
$RD_{G(q^{\prime },j)}(w,w^{\prime },i,i^{\prime })\neq \bot $. We denote
for simplicity $D=D_{1}\cup D_{2}\cup D_{3}$, where%
\begin{eqnarray}
D_{1} &=&D(q,z,z^{\prime }),  \notag \\
D_{2} &=&\left\{ p_{k}\in \mathcal{P}^{\ast }:q\leq k\leq q^{\prime
}\right\},  \label{domination-correctness-D1-D2-D3-eq} \\
D_{3} &=&RD_{G(q^{\prime },j)}(w,w^{\prime },i,i^{\prime }).  \notag
\end{eqnarray}

First we prove that $D$ is a dominating set of $G_{j}$ and that $%
(i,i^{\prime })$ is the end-pair of $D$. Since $D_{1}\neq \bot $ and $%
D_{3}\neq \bot $, note that the set $G_{q}$ is dominated by $D_{1}$ and that
the set $G(q^{\prime },j)$ is dominated by~$D_{3}$. Furthermore, by
Condition~\ref{domination-correctness-condition-7} of the lemma, the set $%
G(q,q^{\prime })$ is dominated by~$D_{2}$. It remains to prove that, if $%
x\notin D$ is an element of $G_{j}$ such that $x\cap F_{p_{q}}\neq \emptyset 
$, or $x\cap F_{p_{q^{\prime }}}\neq \emptyset $, or $x\cap S_{p_{q}}\neq
\emptyset $, or $x\cap S_{p_{q^{\prime }}}\neq \emptyset $, then $x$ is
dominated by some element of $D$.

Assume that $x\notin D$ is an element of $G_{j}$ such that $x\cap
S_{p_{q}}\neq \emptyset $ or $x\cap S_{p_{q^{\prime }}}\neq \emptyset $.
Then $x\in H(p_{q})\cup H(p_{q^{\prime }})$ by Lemma~\ref%
{shadow-hovering-lem}. If $x\in \bigcup_{q\leq k\leq q^{\prime }}N(p_{k})$
then $x$ is clearly dominated by $D_{2}$, cf.~Eq.~(\ref%
{domination-correctness-D1-D2-D3-eq}). Otherwise $x\in {\left( H(p_{q})\cup
H(p_{q^{\prime }})\right) \setminus \left( \bigcup_{q\leq k\leq q^{\prime
}}N(p_{k})\right) }$, and thus $x$ is dominated by the line segments $%
\{L_{z},L_{z^{\prime }},L_{w},L_{w^{\prime }}\}$ by Condition~\ref%
{domination-correctness-condition-6} of the lemma.

Now assume that $x\notin D$ is an element of $G_{j}$ such that $x\cap
F_{p_{q}}\neq \emptyset $ or $x\cap F_{p_{q^{\prime }}}\neq \emptyset $.
Suppose that $x\in \mathcal{P}$, i.e., $x\in F_{p_{q}}$ or $x\in
F_{p_{q^{\prime }}}$. If $x\in F_{p_{q}}$ then $p_{q}\in S_{x}$, and thus $%
p_{q}\in H(x)$ by Lemma~\ref{shadow-hovering-lem}. This is a contradiction,
since $p_{q}\in \mathcal{P}^{\ast }$ by Condition~\ref%
{domination-correctness-condition-5} of the lemma, cf.~the definition of $%
\mathcal{P}^{\ast }$ in Eq.~(\ref{P-ast-eq}). Similarly, if $x\in
F_{p_{q^{\prime }}}$ then we arrive again to a contradiction, since $%
p_{q^{\prime }}\in \mathcal{P}^{\ast }$ by Condition~\ref%
{domination-correctness-condition-1} of the lemma. Therefore $x\notin 
\mathcal{P}$, i.e., $x\in \mathcal{L}$. Let $x=L_{k}$. Since $L_{k}\cap
F_{p_{q}}\neq \emptyset $ or $L_{k}\cap F_{p_{q^{\prime }}}\neq \emptyset $,
it follows that $p_{q}\in S_{k}$ or $p_{q^{\prime }}\in S_{k}$, and thus $%
x=L_{k}\in N(p_{q})\cup N(p_{q^{\prime }})$. That is, $x$ is dominated by $%
\{p_{q},p_{q^{\prime }}\}$. Therefore $D$ is a dominating set of $G_{j}$.
Furthermore, since $(i,i^{\prime })$ is the end-pair of $D_{3}$, it follows
that $(i,i^{\prime })$ is also the end-pair of $D=D_{1}\cup D_{2}\cup D_{3}$.

We now prove that $D$ is normalized. First note that $D_{1}=D(q,z,z^{\prime
})$ is normalized by Definition~\ref{D-general-def} and that $D_{2}$ is
normalized as it only contains elements of $\mathcal{P}^{\ast }$, cf.
Definition~\ref{normal-dominating-def}. Moreover, due to Definition~\ref%
{normal-dominating-def}, $D_{3}$ is normalized as it contains only elements
of $\mathcal{L}$, cf.~Definition~\ref{RD-def} in Section~\ref%
{Restricted-domination-sec}. That is, each of $D_{1}$, $D_{2}$, and $D_{3}$
is normalized. Furthermore note that, due to the Conditions~\ref%
{domination-correctness-condition-2},~\ref%
{domination-correctness-condition-3}, and~\ref%
{domination-correctness-condition-4} of the lemma, for any two elements $%
x,x^{\prime }$ that belong to different sets among $D_{1},D_{2},D_{3}$, no
point of $x$ belongs to the shadow of $x^{\prime }$. Therefore the whole set 
$D$ is normalized. Summarizing, $D$ is a normalized dominating set of $G_{j}$
whose end-pair is $(i,i^{\prime })$.
\end{proof}

\medskip

Given the statement of Lemma~\ref%
{domination-correctness-lem-1-first-direction}, we are now ready to provide
our recursive computation of the sets $D(j,i,i^{\prime })$.

\begin{lemma}
\label{domination-correctness-lem} Let $G$ be a tolerance graph, $(\mathcal{P%
},\mathcal{L})$ be a canonical representation of $G$, $p_{j}\in \mathcal{P}%
^{\ast }$, and $(i,i^{\prime })$ be a a right-crossing pair of $G_{j}$ such
that $D(j,i,i^{\prime })\neq \bot $. Then 
\begin{equation}
D(j,i,i^{\prime })=\min_{q^{\prime },z,z^{\prime },w,w^{\prime }}%
\begin{cases}
D(q,z,z^{\prime })\cup \left\{ p_{k}\in \mathcal{P}^{\ast }:q\leq k\leq
q^{\prime }\right\} \cup RD_{G(q^{\prime },j)}(w,w^{\prime },i,i^{\prime })
\\ 
BD_{G_{j}}(l_{1},b,1,i,i^{\prime })\text{, where }b=\Gamma _{r_{i}}^{\text{%
vert}}\cap \Gamma _{r_{i^{\prime }}}^{\text{diag}}%
\end{cases}%
.  \label{recursion-domination-eq}
\end{equation}%
where the minimum is taken over all $q^{\prime },z,z^{\prime },w,w^{\prime }$
that satisfy\footnote{%
Note that the value of $q$ is uniquely determined by the value of $q^{\prime
}$ and by the pair $(z,z^{\prime })$, cf.~Condition~5 of Lemma~\ref%
{domination-correctness-lem-1-first-direction}.} the Conditions~\ref%
{domination-correctness-condition-1}-\ref{domination-correctness-condition-7}
of Lemma~\ref{domination-correctness-lem-1-first-direction}.
\end{lemma}

\begin{proof}
Let $Z$ be a normalized dominating set of $G_{j}$ such that $(i,i^{\prime })$
is its end-pair and $Z=|D(j,i,i^{\prime })|$. We distinguish the following
two cases.

\medskip

\noindent \textbf{Case 1.} $Z\cap \mathcal{P}^{\ast }=\emptyset $, i.e., $%
Z\subseteq \mathcal{L}$. Denote $b=\Gamma _{r_{i}}^{\text{vert}}\cap \Gamma
_{r_{i^{\prime }}}^{\text{diag}}$ and observe that $X(l_{1},b)\subseteq
G_{j} $. Therefore, since $Z$ is a dominating set of $G_{j}$, it follows
that $Z$ is also a dominating set of $X(l_{1},b)$. Moreover recall that $%
L_{1}$ is a dummy isolated line segment, and thus $L_{1}\in Z$. In
particular, $L_{1}$ is the diagonally leftmost line segment of $Z$.
Therefore $|BD_{G_{j}}(l_{1},b,1,i,i^{\prime })|\leq |Z| $, since $%
Z\subseteq \mathcal{L}$ and $(i,i^{\prime })$ is the end-pair of $Z $ by
assumption.

Since $D(j,i,i^{\prime })\neq \bot $ by assumption, it follows by
Observation~\ref{infeasible-dominating-obs} that there are no points $p\in 
\mathcal{P}\cap G_{j}$ such that $p\in \mathbb{R}_{\text{right}}^{2}(\Gamma
_{r_{i}}^{\text{vert}})$, and that $X(r_{i^{\prime }},p_{j})$ is dominated
by $L_{i}$ and $L_{i^{\prime }}$. Therefore $BD_{G_{j}}(l_{1},b,1,i,i^{%
\prime })$ is a dominating set of $G_{j}$ that has $(i,i^{\prime })$ as its
end-pair. Moreover, due to Definition~\ref{normal-dominating-def}, $%
BD_{G_{j}}(l_{1},b,1,i,i^{\prime })$ is normalized as it contains only
elements of $\mathcal{L}$ (cf.~Definition~\ref{BD-def} in Section~\ref%
{bounded-alg-subsec}). Thus $|Z|\leq |BD_{G_{j}}(l_{1},b,1,i,i^{\prime })|$.
That is, $|Z|=|BD_{G_{j}}(l_{1},b,1,i,i^{\prime })|$.

\medskip

\noindent \textbf{Case 2.} $Z\cap \mathcal{P}^{\ast }\neq \emptyset $. Let $%
q^{\prime }=\max \{k<j:p_{k}\in \mathcal{P}^{\ast }\cap Z\}$, cf.~Figure~\ref%
{general-recursion-fig}. From the assumption that $Z$ is normalized, it
follows that for every line segment $L_{k}\in Z\cap \mathcal{L}$, either $%
L_{k}\subseteq B_{p_{q^{\prime }}}$ or $L_{k}\subseteq A_{p_{q^{\prime }}}$.
Therefore the set $Z\cap \mathcal{L}$ can be partitioned into two sets $Z_{%
\mathcal{L},1}$ and $Z_{\mathcal{L},2}$, where 
\begin{eqnarray*}
Z_{\mathcal{L},1} &=&\{L_{k}\in Z\cap \mathcal{L}:L_{k}\subseteq
B_{p_{q^{\prime }}}\}, \\
Z_{\mathcal{L},2} &=&\{L_{k}\in Z\cap \mathcal{L}:L_{k}\subseteq
A_{p_{q^{\prime }}}\}.
\end{eqnarray*}%
In particular, note that $L_{i},L_{i^{\prime }}\notin N(p_{q^{\prime }})\cup
H(p_{q^{\prime }})$. Now we prove that $L_{i},L_{i^{\prime }}\in Z_{\mathcal{%
L},2}$. Assume otherwise $L_{i}\in Z_{\mathcal{L},1}$, i.e., $L_{i}\subseteq
B_{p_{q^{\prime }}}$. Then $r_{i}\in B_{p_{q^{\prime }}}$, and thus $%
p_{q^{\prime }}\in \mathbb{R}_{\text{right}}^{2}(\Gamma _{r_{i}}^{\text{vert}%
})$. This is a contradiction by Observation~\ref{infeasible-dominating-obs},
since $D(j,i,i^{\prime })\neq \bot $ by assumption. Now assume that $%
L_{i^{\prime }}\in Z_{\mathcal{L},1}$, i.e., $L_{i^{\prime }}\subseteq
B_{p_{q^{\prime }}}$. Then $r_{i^{\prime }}\in B_{p_{q^{\prime }}}$, and
thus $p_{q^{\prime }}\in \mathbb{R}_{\text{right}}^{2}(\Gamma _{r_{i^{\prime
}}}^{\text{diag}})$. This is a contradiction to the assumption that $%
(i,i^{\prime })$ is the end-pair of $D(j,i,i^{\prime })$. Summarizing, $%
L_{i},L_{i^{\prime }}\in Z_{\mathcal{L},2}$.

\begin{figure}[t]
\centering
\includegraphics[width=\textwidth]{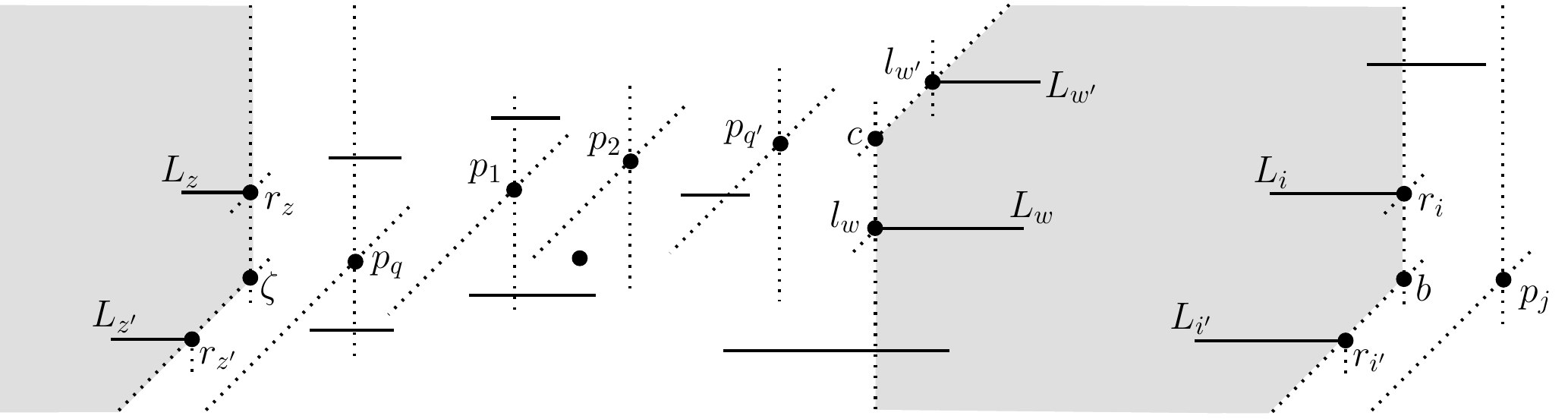}
\caption{The recursion for Case 2 of Lemma~\protect\ref%
{domination-correctness-lem}, where $p_{q}, p_{1}, p_{2}, p_{q^{\prime }}
\in P^{\ast}$.}
\label{general-recursion-fig}
\end{figure}

Notice that $Z_{\mathcal{L},2}\subseteq \mathcal{L}$ is a bounded dominating
set of $G(q^{\prime },j)$ with $(i,i^{\prime })$ as its end-pair, and thus $%
Z_{\mathcal{L},2}\neq \emptyset $. Since $Z_{\mathcal{L},2}\subseteq 
\mathcal{L}$, Observation~\ref{unique-start-end-pair-L-obs} implies that $Z_{%
\mathcal{L},2}$ contains a unique start-pair. Let $(w,w^{\prime })$ be the
left-crossing pair of $G(q^{\prime },j)$ which is the start-pair of $Z_{%
\mathcal{L},2}$. 
Then%
\begin{equation}
|RD_{G(q^{\prime },j)}(w,w^{\prime },i,i^{\prime })|\leq |Z_{\mathcal{L},2}|,
\label{eq:eqzl2}
\end{equation}%
and thus $RD_{G(q^{\prime },j)}(w,w^{\prime },i,i^{\prime })\neq \bot $.

Recall that $G_{j}$ contains the isolated (dummy) line segment $L_{1}$, and
thus $L_{1}\in Z_{\mathcal{L},1}$. Therefore $Z_{\mathcal{L},1}\neq
\emptyset $. Since $Z_{\mathcal{L},1}\subseteq \mathcal{L}$, Observation~\ref%
{unique-start-end-pair-L-obs} implies that $Z_{\mathcal{L},1}$ contains a
unique end-pair. Let $(z,z^{\prime })$ be the right-crossing pair of $%
G_{q^{\prime }}$ which is the end-pair of $Z_{\mathcal{L},1}$. Denote $\zeta
=\Gamma _{r_{z}}^{\text{vert}}\cap \Gamma _{r_{z^{\prime }}}^{\text{diag}}$,
cf.~Figure~\ref{general-recursion-fig}.

Consider now an arbitrary point $p\in \mathcal{P}^{\ast }\cap Z$. We will
prove that $p\notin F_{\zeta }\cup S_{\zeta }$. Assume otherwise that $p\in
F_{\zeta }$. Then $p\in \mathbb{R}_{\text{right}}^{2}(\Gamma _{r_{z}}^{\text{%
vert}})$, and thus also $p\in \mathbb{R}_{\text{right}}^{2}(\Gamma
_{r_{z^{\prime }}}^{\text{vert}})$. Moreover $p\in \mathbb{R}_{\text{left}%
}^{2}(\Gamma _{r_{z^{\prime }}}^{\text{diag}})$. This implies that $p\in
F_{r_{z^{\prime }}}$. That is, $r_{z^{\prime }}\in S_{p}$, and thus Lemma~%
\ref{shadow-hovering-lem} implies that $L_{z^{\prime }}\in H(p)$. This is a
contradiction to the assumption that $Z$ is normalized, since both $%
p,L_{z^{\prime }}\in Z$. Thus $p\notin F_{\zeta }$. Now assume that $p\in
S_{\zeta }$. Then $p\in \mathbb{R}_{\text{right}}^{2}(\Gamma _{r_{z^{\prime
}}}^{\text{diag}})$, and thus also $p\in \mathbb{R}_{\text{right}%
}^{2}(\Gamma _{r_{z}}^{\text{diag}})$. Furthermore $p\in \mathbb{R}_{\text{%
left}}^{2}(\Gamma _{r_{z}}^{\text{vert}})$. This implies that $p\in
S_{r_{z}} $, and thus $L_{z}\in N(p)$. This is again a contradiction to the
assumption that $Z$ is normalized, since both $p,L_{z}\in Z$. Thus $p\notin
S_{\zeta }$. Summarizing, for every $p\in P^{\ast }\cap Z$ we have that $%
p\notin F_{\zeta }\cup S_{\zeta }$, i.e., either $p\in A_{\zeta }$ or $p\in
B_{\zeta } $. Therefore the set $P^{\ast }\cap Z$ can be partitioned into
two sets $Z_{\mathcal{P}^{\ast },1}$ and $Z_{\mathcal{P}^{\ast },2}$, where 
\begin{eqnarray*}
Z_{\mathcal{P}^{\ast },1} &=&\{p\in P^{\ast }\cap Z:p\in B_{\zeta }\}, \\
Z_{\mathcal{P}^{\ast },2} &=&\{p\in P^{\ast }\cap Z:p\in A_{\zeta }\}.
\end{eqnarray*}

Note that $p_{q}\in Z_{\mathcal{P}^{\ast },2}$. Furthermore, since $%
(z,z^{\prime })$ is the end-pair of $Z_{\mathcal{L},1}$, note that all line
segments of $Z_{\mathcal{L},1}$ are contained in $B_{\zeta }$. Therefore all
elements of the set $Z_{1}=Z_{\mathcal{L},1}\cup Z_{\mathcal{P}^{\ast },1}$
are contained in $B_{\zeta }$, and thus $(z,z^{\prime })$ is the end-pair of 
$Z_{1}$. Define now $q=\min \{1\leq k\leq q^{\prime }:p_{k}\in \mathcal{P}%
^{\ast },\ p_{k}\in A_{\zeta }\}$, cf.~Figure~\ref{general-recursion-fig}.
Recall that $p_{q}\notin G_{q}$, cf.~Eq.~(\ref{G_j-eq1}). It is easy to
check that no line segment of $Z_{\mathcal{L},2}$ dominates any element of $%
G_{q}$, cf.~Figure~\ref{general-recursion-fig}. Similarly, no point of $Z_{%
\mathcal{P}^{\ast },2}$ dominates any element of $G_{q}$. Thus the set $%
Z_{1} $ is a dominating set of $G_{q}$. Furthermore $Z_{1}$ is normalized,
since $Z_{1}\subseteq Z$ and $Z$ is normalized by assumption. That is, $%
Z_{1} $ is a normalized dominating set of $G_{q}$ with $(z,z^{\prime })$ as
its end-pair. Therefore, 
\begin{equation}
|D(q,z,z^{\prime })|\leq |Z_{1}|,  \label{eq:eqz1}
\end{equation}%
and thus $D(q,z,z^{\prime })\neq \bot $.

We now prove that $Z_{\mathcal{P}^{\ast },2}=\{p_{k}\in \mathcal{P}^{\ast
}:q\leq k\leq q^{\prime }\}$. Clearly $Z_{\mathcal{P}^{\ast },2}\subseteq
\{p_{k}\in \mathcal{P}^{\ast }:q\leq k\leq q^{\prime }\}$ by the definition
of the index $q$ and of the set $Z_{\mathcal{P}^{\ast },2}$. Recall that for
every line segment $L_{t}\in Z$, either $L_{t}\in Z_{\mathcal{L},1}$ or $%
L_{t}\in Z_{\mathcal{L},2}$. If $L_{t}\in Z_{\mathcal{L},1}$ then $%
L_{t}\subseteq B_{\zeta }\subseteq B_{p_{q}}$. Denote $c=\Gamma _{l_{w}}^{%
\text{vert}}\cap \Gamma _{l_{w^{\prime }}}^{\text{diag}}$, cf.~Figure~\ref%
{general-recursion-fig}. If $L_{t}\in Z_{\mathcal{L},2}$ then $%
L_{t}\subseteq A_{c}\subseteq A_{p_{q^{\prime }}}$, since $(w,w^{\prime })$
is the start-pair of $Z_{\mathcal{L},2}$. Thus, for every line segment $%
L_{t}\in Z$, either $L_{t}\subseteq B_{p_{q}}$ or $L_{t}\subseteq
A_{p_{q^{\prime }}}$. Therefore $N(p_{k})\cap Z=\emptyset $, for every $k\in
\{q,q+1,\ldots ,q^{\prime }\}$, and thus all points $p_{k}\in \mathcal{P}%
^{\ast }$, where $q\leq k\leq q^{\prime }$, must belong to $Z$. That is, $%
\{p_{k}\in \mathcal{P}^{\ast }:q\leq k\leq q^{\prime }\}\subseteq Z_{%
\mathcal{P}^{\ast },2}$. Therefore,%
\begin{equation}
Z_{\mathcal{P}^{\ast },2}=\{p_{k}\in \mathcal{P}^{\ast }:q\leq k\leq
q^{\prime }\}.  \label{eq:eqzp22222}
\end{equation}

Recall that for every line segment $L_{k}\in Z$, either $L_{k}\subseteq
B_{p_{q}}$ or $L_{k}\subseteq A_{p_{q^{\prime }}}$, as we proved above.
Therefore $G(q,q^{\prime })$ must be dominated by $Z_{\mathcal{P}^{\ast },2}$%
. Furthermore, due to Eq.~(\ref{eq:eqzp22222}), $Z_{\mathcal{P}^{\ast },2}$
clearly dominates the set $\bigcup_{q\leq k\leq q^{\prime }}N(p_{k})$.
Moreover every hovering vertex of $p_{q}$ and of $p_{q^{\prime }}$ must be
dominated by $Z_{\mathcal{P}^{\ast },2}$ or by the set $\{L_{z},L_{z^{\prime
}},L_{w},L_{w^{\prime }}\}$. Therefore $\{L_{z},L_{z^{\prime
}},L_{w},L_{w^{\prime }}\}$ must dominate the set $(H(p_{q})\cup
H(p_{q^{\prime }}))\setminus \left( \bigcup_{q\leq k\leq q^{\prime
}}N(p_{k})\right) $.

Now note that the sets $D(q,z,z^{\prime })$, $Z_{\mathcal{P}^{\ast },2}$,
and $RD_{G(q^{\prime },j)}(w,w^{\prime },i,i^{\prime })$ are mutually
disjoint. Furthermore, it follows by Eq.~(\ref{eq:eqzl2}) and~(\ref{eq:eqz1}%
) that%
\begin{eqnarray}
\left\vert D(q,z,z^{\prime })\right\vert +\left\vert Z_{\mathcal{P}^{\ast
},2}\right\vert +\left\vert RD_{G(q^{\prime },j)}(w,w^{\prime },i,i^{\prime
})\right\vert &\leq &|Z_{1}|+|Z_{\mathcal{P}^{\ast },2}|+|Z_{\mathcal{L},2}|
\notag \\
&=&|Z_{\mathcal{L},1}\cup Z_{\mathcal{P}^{\ast },1}|+|Z_{\mathcal{P}^{\ast
},2}|+|Z_{\mathcal{L},2}|
\label{domination-correctness-case-2-inequality-eq-1} \\
&=&|Z|=|D(j,i,i^{\prime })|.  \notag
\end{eqnarray}%
Therefore $\left\vert D(q,z,z^{\prime })\cup Z_{\mathcal{P}^{\ast },2}\cup
RD_{G(q^{\prime },j)}(w,w^{\prime },i,i^{\prime })\right\vert \leq
|D(j,i,i^{\prime })|$. On the other hand, since $Z_{\mathcal{P}^{\ast
},2}=\{p_{k}\in \mathcal{P}^{\ast }:q\leq k\leq q^{\prime }\}$ by Eq.~(\ref%
{eq:eqzp22222}), Lemma~\ref{domination-correctness-lem-1-first-direction}
implies that, if $D(q,z,z^{\prime })\neq \bot $ and $RD_{G(q^{\prime
},j)}(w,w^{\prime },i,i^{\prime })\neq \bot $, then $D(q,z,z^{\prime })\cup
Z_{\mathcal{P}^{\ast },2}\cup RD_{G(q^{\prime },j)}(w,w^{\prime
},i,i^{\prime })$ is a normalized dominating set of $G_{j}$, in which $%
(i,i^{\prime })$ is its end-pair. Therefore%
\begin{equation}
|D(j,i,i^{\prime })|\leq \left\vert D(q,z,z^{\prime })\cup Z_{\mathcal{P}%
^{\ast },2}\cup RD_{G(q^{\prime },j)}(w,w^{\prime },i,i^{\prime
})\right\vert.  \label{domination-correctness-case-2-inequality-eq-2}
\end{equation}%
The lemma follows by Eq.~(\ref{domination-correctness-case-2-inequality-eq-1}%
) and~(\ref{domination-correctness-case-2-inequality-eq-2}).
\end{proof}

\medskip

We are now ready to present Algorithm~\ref{dominating-tol-alg} which, given
a canonical horizontal shadow representation $(\mathcal{P},\mathcal{L}) $ of
a connected tolerance graph $G$, computes a (normalized) minimum dominating
set $D$ of $G$. The correctness of Algorithm~\ref{dominating-tol-alg} is
proved in Theorem~\ref{domination-correctness-thm}.

\begin{algorithm}[t!]
\caption{\textsc{Dominating Set} on Tolerance Graphs} \label{dominating-tol-alg}
\begin{algorithmic} [1]
\REQUIRE{A canonical horizontal shadow representation $(\mathcal{P},\mathcal{L})$, where ${\mathcal{P} = \{p_{1},p_{2}, \ldots, p_{|\mathcal{P}|}\}}$ and 
${\mathcal{L} = \{L_{1},L_{2}, \ldots, L_{|\mathcal{L}|}\}}$.}
\ENSURE{A set $D\subseteq \mathcal{L} \cup \mathcal{P}$ of minimum size that dominates $(\mathcal{P},\mathcal{L})$.}

\medskip

\STATE{Add two dummy line segments $L_{0}$ (resp.~$L_{|\mathcal{L}|+1}$) completely to the left (resp.~right) of $\mathcal{P} \cup \mathcal{L}$} \label{alg-3-line-1}

\STATE{Add a dummy point $p_{|\mathcal{P}|+1}$ completely to the right of $L_{|\mathcal{L}|+1}$} \label{alg-3-line-2}

\STATE{$\mathcal{P} \leftarrow \mathcal{P} \cup \{p_{|\mathcal{P}|+1}\}$; \ \ 
$\mathcal{L} \leftarrow \mathcal{L} \cup \{L_{0}, L_{|\mathcal{L}|+1}\}$} \label{alg-3-line-3}

\STATE{Denote ${\mathcal{P} = \{p_{1},p_{2}, \ldots, p_{|\mathcal{P}|}\}}$ and ${\mathcal{L} = \{L_{1},L_{2}, \ldots, L_{|\mathcal{L}|}\}}$, 
where now $p_{|\mathcal{P}|}$, $L_{1}$, and $L_{|\mathcal{L}|}$ are dummy} \label{alg-3-line-4}

\vspace{0,1cm}

\STATE{$\mathcal{P}^{\ast }=\{p\in \mathcal{P}:p\notin H(p^{\prime })\text{ for every point }p^{\prime }\in \mathcal{P}\setminus \{p\}\}$} \label{alg-3-line-5}

\vspace{0,1cm}

\FOR{every pair of points $(a,b) \in \mathcal{A} \times \mathcal{B}$ such that $b \in \mathbb{R}_{\text{right}}^{2}(\Gamma_{a}^{\text{diag}})$} \label{alg-3-line-6}
     \STATE{$X(a,b) \leftarrow \{x\in \mathcal{P}\cup \mathcal{L}:x\subseteq \left( B_{b}\setminus \Gamma _{b}^{\text{vert}}\right) \cap \mathbb{R}_{\text{right}}^{2}(\Gamma_{a}^{\text{diag}})\}$} \label{alg-3-line-7}
\ENDFOR 

\vspace{0,1cm}

\FOR{every $p_{j} \in \mathcal{P}^{\ast}$} \label{alg-3-line-8}
     \STATE{$G_{j} \leftarrow \{x\in \mathcal{P}\cup \mathcal{L}:x\subseteq B_{p_{j}}\setminus \Gamma _{p_{j}}^{\text{vert}}\}$} \label{alg-3-line-9}
     \FOR{every $i,i^{\prime} \in \{1,2,\ldots, |\mathcal{L}|\}$} \label{alg-3-line-10}
     
     \vspace{0,1cm}
     
          \IF[$(i,i^{\prime})$ is a right-crossing pair of $G_{j}$]{$L_{i}, L_{i^{\prime}} \in G_{j}$ and $r_{i^{\prime}} \in S_{r_{i}}$} \label{alg-3-line-11}
               \STATE{\textbf{if} \ $\mathcal{L}_{i,i^{\prime }}^{\text{left}}$ does not dominate all elements of $G_{j}$ \ \textbf{then} \ $D(j,i,i^{\prime}) \leftarrow \bot $} \label{alg-3-line-12}
               \STATE{\ \ \ \textbf{else} \ Compute $D(j,i,i^{\prime})$ by Lemma~\ref{domination-correctness-lem}} \COMMENT{by calling Algorithms~\ref{bounded-dominating-tolerance-alg}~and~\ref{restricted-bounded-alg}} \label{alg-3-line-13}

          \ENDIF
     \ENDFOR
\ENDFOR

\medskip

\RETURN{$D(|\PP|, |\LL|, |\LL|) \setminus \{L_{1}, L_{\LL}\}$} \label{alg-3-line-14}
\end{algorithmic}
\end{algorithm}

\begin{theorem}
\label{domination-correctness-thm}Given a canonical horizontal shadow
representation $(\mathcal{P},\mathcal{L})$ of a connected tolerance graph $G$
with $n$ vertices, Algorithm~\ref{dominating-tol-alg} computes in $O(n^{15})$
time a (normalized) minimum dominating set~$D$ of~$G$.
\end{theorem}

\begin{proof}
In the first line, Algorithm~\ref{dominating-tol-alg} augments the given
canonical horizontal shadow representation $(\mathcal{P},\mathcal{L})$ by
adding to $\mathcal{L}$ the dummy line segments $L_{0}$ and $L_{|\mathcal{L}%
|+1}$ (with endpoints $l_{0},r_{0}$ and $l_{|\mathcal{L}|+1},r_{|\mathcal{L}%
|+1}$, respectively) such that all elements of $\mathcal{P\cup L}$ are
contained in $A_{r_{0}}$ and in $B_{l_{|\mathcal{L}|+1}}$. Furthermore, in
the second line, the algorithm further augments the set of points $\mathcal{P%
}$ by adding to it the dummy point $p_{|\mathcal{P}|+1}$ such that all
elements of $\mathcal{P\cup L}^{\prime }$ are contained in $B_{p_{|\mathcal{P%
}|+1}}$. In lines~\ref{alg-3-line-3} and~\ref{alg-3-line-4} the algorithm
renumbers the elements of the sets $\mathcal{P}$ and $\mathcal{L}$ such that 
$\mathcal{P}=\{p_{1},p_{2},\ldots ,p_{|\mathcal{P}|}\}$ and $\mathcal{L}%
=\{L_{1},L_{2},\ldots ,L_{|\mathcal{L}|}\}$, where in this new enumeration
the point $p_{|\mathcal{P}|}$ is dummy and the line segments $L_{1}$ and $%
L_{|\mathcal{L}|}$ are dummy as well. In lines~\ref{alg-3-line-5}-\ref%
{alg-3-line-9} the algorithm computes the subset $\mathcal{P}^{\ast
}\subseteq \mathcal{P}$ (cf.~Eq.~(\ref{P-ast-eq})), all feasible subsets $%
X(a,b)\subseteq \mathcal{P\cup L}$ (cf.~Eq.~(\ref{X(a,b)-def-eq}) in Section~%
\ref{bounded-alg-subsec}), and all sets $G_{j}$, where $p_{j}\in \mathcal{P}%
^{\ast }$ (cf.~Eq.~(\ref{G_j-eq1})).

The main computations of the algorithm are performed in lines~\ref%
{alg-3-line-12}-\ref{alg-3-line-13}, which are executed for every point $%
p_{j}\in \mathcal{P}^{\ast }$ and for every right-crossing pair $%
(i,i^{\prime })$ of the set $G_{j}$. In line~\ref{alg-3-line-12} the
algorithm checks whether $\mathcal{L}_{i,i^{\prime }}^{\text{left}}$
dominates all elements of $G_{j}$. If it is not the case, it correctly
computes $D(j,i,i^{\prime })=\bot $ by Observation~\ref%
{general-infeasible-characterization-obs}. Otherwise, if $\mathcal{L}%
_{i,i^{\prime }}^{\text{left}}$ is a dominating set of $G_{j}$, then the
algorithm computes in line~\ref{alg-3-line-13} the value of $D(j,i,i^{\prime
})$ with the recursive formula of Lemma~\ref{domination-correctness-lem}.
Note that, to compute all the necessary values for this recursive formula,
Algorithm~\ref{dominating-tol-alg} needs to call Algorithms~\ref%
{bounded-dominating-tolerance-alg} and~\ref{restricted-bounded-alg} as
subroutines, cf.~Lemma~\ref{domination-correctness-lem}.

Once all values $D(j,i,i^{\prime })$ have been computed, the set $D(|%
\mathcal{P}|,|\mathcal{L}|,|\mathcal{L}|)$ is a minimum normalized
dominating set of $G_{|\mathcal{P}|}$ whose end-pair is $(|\mathcal{L}|,|%
\mathcal{L}|)$, cf.~Definition~\ref{D-general-def}. Recall that $p_{|%
\mathcal{P}|}\notin G_{|\mathcal{P}|}$, i.e., $G_{|\mathcal{P}|}=(\mathcal{P}%
\setminus \{p_{|\mathcal{P}|}\})\cup \mathcal{L}$. Therefore, since the two
dummy line segments are isolated, they must belong to the dominating set $D(|%
\mathcal{P}|,|\mathcal{L}|,|\mathcal{L}|)$ of $G_{|\mathcal{P}|}$. Thus the
algorithm correctly returns in line~\ref{alg-3-line-14} the value $D(|%
\mathcal{P}|,|\mathcal{L}|,|\mathcal{L}|)\setminus \{L_{1},L_{|\mathcal{L}%
|}\}$ as a minimum normalized dominating set for the input tolerance graph $%
G $.

Regarding the running time of Algorithm~\ref{dominating-tol-alg}, first note
that the execution time of lines~\ref{alg-3-line-1}-\ref{alg-3-line-5} is
dominated by the computation of the set $\mathcal{P}^{\ast }$ in line~\ref%
{alg-3-line-5}; this can be done in at most $O(n^{2})$ time, since we check
in the worst case for every two points $p,p^{\prime }\in \mathcal{P}$
whether $p\in H(p^{\prime })$. Due to the for-loop of line~\ref{alg-3-line-6}%
, line~\ref{alg-3-line-7} is executed at most~$O(n^{3})$ times. Furthermore
recall by Eq.~(\ref{region-R(a,b)-def-eq}) and~(\ref{X(a,b)-def-eq}) that,
for every pair $(a,b)\in \mathcal{A}\times \mathcal{B}$, the vertex set $%
X(a,b)$ can be computed in $O(n)$ time. Therefore, lines~\ref{alg-3-line-6}-%
\ref{alg-3-line-7} are executed in $O(n^{4})$ time. Due to the for-loop of
line~\ref{alg-3-line-8}, the lines~\ref{alg-3-line-9}-\ref{alg-3-line-13}
are executed~$O(n)$ times, since there are at most $O(n)$ points in the set $%
\mathcal{P}^{\ast }$. For every fixed $p_{j}\in \mathcal{P}^{\ast }$, line~\ref%
{alg-3-line-9} can be trivially executed in $O(n)$ time. For every fixed $%
p_{j}\in \mathcal{P}^{\ast }$, the lines~\ref{alg-3-line-11}-\ref%
{alg-3-line-13} are executed $O(n^{2})$ times, due to the for-loop of line~%
\ref{alg-3-line-10}. Furthermore, for every fixed triple $(j,i,i^{\prime })$%
, line~\ref{alg-3-line-11} can be executed in constant time and line~\ref%
{alg-3-line-12} can be easily executed in $O(n^{2})$ time.

It remains to upper bound the execution time of line~\ref{alg-3-line-13}
using Lemma~\ref{domination-correctness-lem}. Before we execute line~\ref%
{alg-3-line-13} for the first time, we perform two preprocessing steps. In
the first preprocessing step we compute, for each of the $O(n)$ possible
values for $j$, the graph $G_{j}$ in $O(n)$ time (cf.~Eq.~(\ref{G_j-eq1}))
and then we compute by Algorithm~\ref{bounded-dominating-tolerance-alg} in $%
O(n^{9})$ time the values $BD_{G_{j}}(l_{1},b,1,i,i^{\prime })$ for every
feasible pair $(i,i^{\prime })$, cf.~Theorem~\ref{bounded-correctness-thm} in Section~\ref{Bounded-dominating-sec}. 
That is, we compute in the first preprocessing step the values $%
BD_{G_{j}}(l_{1},b,1,i,i^{\prime })$ for every triple $(j,i,i^{\prime })$ in 
$O(n^{10})$ time. In the second preprocessing step we compute, for each of
the $O(n^{6})$ possible values for $q^{\prime },j,w,w^{\prime },i,i^{\prime
} $, the graph $G(q^{\prime },j)$ in $O(n)$ time (cf.~Eq.~(\ref{G-i-j-eq2}))
and then we compute by Algorithm~\ref{restricted-bounded-alg} in $O(n^{9})$
time the values $RD_{G(q^{\prime },j)}(w,w^{\prime },i,i^{\prime })$, 
cf.~Theorem~\ref{restricted-correctness-thm} in Section~\ref{Restricted-domination-sec}. 
That is, we compute in the second
preprocessing step all values $RD_{G(q^{\prime },j)}(w,w^{\prime
},i,i^{\prime })$ in $O(n^{15})$ time.

Consider a fixed value for the triple $(j,i,i^{\prime })$. Then there exist $%
O(n)$ feasible values for $q^{\prime }$, cf.~Conditions~1 and~2 of Lemma~\ref%
{domination-correctness-lem-1-first-direction}. Furthermore there exist $O(n^{2})$
feasible values for the pair $(z,z^{\prime })$, cf.~Condition~4 of
Lemma~\ref{domination-correctness-lem-1-first-direction}. Once
the values of $q,z,z^{\prime }$ have been chosen, we can compute in $O(n)$
time the value of $q$, cf.~Conditions~5 and~6 of Lemma~\ref%
{domination-correctness-lem-1-first-direction}. Furthermore, once the values
of $q^{\prime }$ and $q$ have been chosen, we can check Condition~7 of 
Lemma~\ref{domination-correctness-lem-1-first-direction} in $O(n^{2})$ time. Thus,
given a fixed value for the triple $(j,i,i^{\prime })$, we can compute in $%
O(n^{5})$ time the sets $D(q,z,z^{\prime })\cup \left\{ p_{k}\in \mathcal{P}%
^{\ast }:q\leq k\leq q^{\prime }\right\} $, for all feasible values of the
triples $(q,z,z^{\prime })$. Moreover, for each of the $O(n^{2})$ feasible
pairs $(w,w^{\prime })$ (cf.~Condition~3 of Lemma~\ref{domination-correctness-lem-1-first-direction}) we can compute in $O(n)$ time the set $%
D(q,z,z^{\prime })\cup \left\{ p_{k}\in \mathcal{P}^{\ast }:q\leq k\leq
q^{\prime }\right\} \cup RD_{G(q^{\prime },j)}(w,w^{\prime },i,i^{\prime })$%
, cf.~Lemma~\ref{domination-correctness-lem-1-first-direction}. That is, for
a fixed value of the triple $(j,i,i^{\prime })$, we can compute all these
sets in $O(n^{8})$ time, and thus we can compute all values of $%
D(j,i,i^{\prime })$ in $O(n^{11})$ time.

Summarizing, the running time of the algorithm is dominated by the two
preprocessing steps for computing in advance all values $%
BD_{G_{j}}(l_{1},b,1,i,i^{\prime })$ and $RD_{G(q^{\prime },j)}(w,w^{\prime
},i,i^{\prime })$, and thus the running time of Algorithm~\ref%
{dominating-tol-alg} is $O(n^{15})$.
\end{proof}

\section{Concluding Remarks\label{conclusions-sec}}

In this paper we introduced two new geometric representations for tolerance
and multitolerance graphs, called the \emph{horizontal shadow representation} and
the \emph{shadow representation}, respectively. Using these new representations we
first proved that the dominating set problem is \textsf{APX}-hard on
multitolerance graphs and then we provided a polynomial time algorithm for
this problem on tolerance graphs, thus answering to a longstanding open question.
Therefore, given the (seemingly) small difference between the definition of
tolerance and multitolerance graphs, this dichotomy result appears to be
surprising. 

The two new representations have the potential for further
exploitation via sweep line algorithms. For example, using the shadow
representation, it is not very difficult to design a polynomial sweep line
algorithm for the independent dominating set problem, even on the larger
class of multitolerance graphs. In particular, although the complexity of
the dominating set problem has been established in this paper for both
tolerance and multitolerance graphs, an interesting research direction would
be to use these new representations also for other related problems, e.g., 
for the connected dominating set problem. A major open problem in tolerance and
multitolerance graphs is to establish the computational complexity of the \emph{Hamiltonicity problems}. 
We hope that the two new geometric representations can provide new insights also for these problems.

Our algorithm for tolerance graphs is highly non-trivial and its running
time is upper-bounded by $O(n^{15})$, where $n$ is the number of vertices in
the input tolerance graph. Using more sophisticated data structures our
algorithm could run slightly faster. As our main aim in this paper was to
establish the \emph{first} polynomial-time algorithm for this problem,
rather than finding an optimized efficient algorithm, an interesting
research direction is to explore to what extend the running time can be
reduced. The existence of a \emph{practically efficient} polynomial-time
algorithm for the dominating set problem on tolerance graphs remains widely
open.

\paragraph*{Acknowledgments.}

The second author wishes to thank Steven Chaplick for insightful initial
discussions.

{

}


\begin{thebibliography}{10}

\bibitem{Altschul90}
S.~F. Altschul, W.~Gish, W.~Miller, E.~W. Myers, and D.~J. Lipman.
\newblock {Basic Local Alignment Search Tool}.
\newblock {\em Journal of Molecular Biology}, 215(3):403--410, 1990.

\bibitem{Bogart95}
K.~P. Bogart, P.~C. Fishburn, G.~Isaak, and L.~Langley.
\newblock Proper and unit tolerance graphs.
\newblock {\em Discrete Applied Mathematics}, 60(1-3):99--117, 1995.

\bibitem{Booth82}
K.~S. Booth and J.~H. Johnson.
\newblock Dominating sets in chordal graphs.
\newblock {\em {SIAM} Journal on Computing}, 11(1):191--199, 1982.

\bibitem{Bus06}
A.~H. Busch.
\newblock A characterization of triangle-free tolerance graphs.
\newblock {\em Discrete Applied Mathematics}, 154(3):471--477, 2006.

\bibitem{ChanG14}
T.~M. Chan and E.~Grant.
\newblock Exact algorithms and {APX}-hardness results for geometric packing and
  covering problems.
\newblock {\em Computational Geometry}, 47(2):112--124, 2014.

\bibitem{DeogunS94}
J.~S. Deogun and G.~Steiner.
\newblock Polynomial algorithms for hamiltonian cycle in cocomparability
  graphs.
\newblock {\em {SIAM} Journal on Computing}, 23(3):520--552, 1994.

\bibitem{Fel98}
S.~Felsner.
\newblock Tolerance graphs and orders.
\newblock {\em Journal of Graph Theory}, 28(3):129--140, 1998.

\bibitem{GoMo82}
M.~C. Golumbic and C.~L. Monma.
\newblock A generalization of interval graphs with tolerances.
\newblock In {\em Proceedings of the 13th Southeastern Conference on
  Combinatorics, Graph Theory and Computing, Congressus Numerantium 35}, pages
  321--331, 1982.

\bibitem{GolumbicMonma84}
M.~C. Golumbic, C.~L. Monma, and W.~T. Trotter.
\newblock Tolerance graphs.
\newblock {\em Discrete Applied Mathematics}, 9(2):157--170, 1984.

\bibitem{GolSi02}
M.~C. Golumbic and A.~Siani.
\newblock Coloring algorithms for tolerance graphs: reasoning and scheduling
  with interval constraints.
\newblock In {\em Proceedings of the Joint International Conferences on
  Artificial Intelligence, Automated Reasoning, and Symbolic Computation
  (AISC/Calculemus)}, pages 196--207, 2002.

\bibitem{GolTol04}
M.~C. Golumbic and A.~N. Trenk.
\newblock {\em Tolerance Graphs}.
\newblock Cambridge studies in advanced mathematics, 2004.

\bibitem{Kaufmann06}
M.~Kaufmann, J.~Kratochvil, K.~A. Lehmann, and A.~R. Subramanian.
\newblock Max-tolerance graphs as intersection graphs: cliques, cycles, and
  recognition.
\newblock In {\em Proceedings of the 17th annual ACM-SIAM symposium on Discrete
  Algorithms (SODA)}, pages 832--841, 2006.

\bibitem{KratschStewart93}
D.~Kratsch and L.~Stewart.
\newblock Domination on cocomparability graphs.
\newblock {\em SIAM Journal on Discrete Mathematics}, 6(3):400--417, 1993.

\bibitem{Lehmann06}
K.~A. Lehmann, M.~Kaufmann, S.~Steigele, and K.~Nieselt.
\newblock On the maximal cliques in $c$-max-tolerance graphs and their
  application in clustering molecular sequences.
\newblock {\em Algorithms for Molecular Biology}, 1, 2006.

\bibitem{Multitol-Mertzios14}
G.~B. Mertzios.
\newblock An intersection model for multitolerance graphs: Efficient algorithms
  and hierarchy.
\newblock {\em Algorithmica}, 69(3):540--581, 2014.

\bibitem{MSZ-Model-SIDMA-09}
G.~B. Mertzios, I.~Sau, and S.~Zaks.
\newblock A new intersection model and improved algorithms for tolerance
  graphs.
\newblock {\em {SIAM} Journal on Discrete Mathematics}, 23(4):1800--1813, 2009.

\bibitem{MSZ-SICOMP-11}
G.~B. Mertzios, I.~Sau, and S.~Zaks.
\newblock The recognition of tolerance and bounded tolerance graphs.
\newblock {\em SIAM Journal on Computing}, 40(5):1234--1257, 2011.

\bibitem{Muller96}
H.~M{\"u}ller.
\newblock Hamiltonian circuits in chordal bipartite graphs.
\newblock {\em Discrete Mathematics}, 156(1-3):291--298, 1996.

\bibitem{Parra98}
A.~Parra.
\newblock Triangulating multitolerance graphs.
\newblock {\em Discrete Applied Mathematics}, 84(1-3):183--197, 1998.

\bibitem{Spinrad03}
J.~P. Spinrad.
\newblock {\em Efficient graph representations}, volume~19 of {\em Fields
  Institute Monographs}.
\newblock American Mathematical Society, 2003.

\bibitem{Williamson11}
D.~P. Williamson and D.~B. Shmoys.
\newblock {\em The Design of Approximation Algorithms}.
\newblock Cambridge University Press, New York, NY, USA, 2011.

\end{thebibliography}
\end{document}